\documentclass[usenames,dvipsnames,letterpaper,11pt]{article}
\usepackage{setspace}
\usepackage{amsmath}
\usepackage[T1]{fontenc}
\usepackage[margin=1in]{geometry}
\usepackage[numbers, semicolon]{natbib}
\setcitestyle{authoryear,open={(},close={)},aysep={}} 
\usepackage{amsfonts}
\usepackage{float}
\usepackage{graphicx}
\usepackage{tikz}
\usetikzlibrary{backgrounds, fit, positioning}
\usepackage{mathrsfs}
\usepackage{caption}
\usepackage{subcaption}
\usepackage[titles,subfigure]{tocloft}
\usepackage{layouts}
\usepackage{hyperref}
\usepackage[normalem]{ulem}
\usepackage{authblk}
\usepackage{bm}
\usepackage{hyperref}
\newcommand{\simiid}{\stackrel{\text{i.i.d.}}{\sim}}
\usepackage[ruled,procnumbered]{algorithm2e}
\usepackage{amsthm}
\usepackage{lscape}

\newtheorem{thm}{Theorem}
\newtheorem{defi}{Definition}
\newtheorem{prop}{Proposition}
\newtheorem{coro}{Corollary}
\newtheorem{lemma}{Lemma}

\newcommand{\rr}{\mathbb R}

\newcommand{\p}{\mathbb P}
\newcommand{\ind}{\mathbf{1}}
\newcommand{\di}[1]{\mathop{d#1}}
\newcommand{\Xno}[1]{X_{\text{-}{#1}}}
\newcommand{\xno}[1]{x_{\text{-}{#1}}}
\DeclareMathOperator{\Var}{Var}
\DeclareMathOperator{\Cov}{Cov}

\DeclareMathOperator{\Expo}{Expo}
\DeclareMathOperator{\Bern}{Bern}

\DeclareMathOperator{\diag}{diag}

\DeclareMathOperator{\MAC}{MAC}
\newcommand{\e}{\mathbb E}

\newcommand{\ci}{\mathrel{\perp\mspace{-10mu}\perp}}

\providecommand{\abs}[1]{\lvert#1\rvert}

\newcommand{\Xk}{\tilde{X}}
\newcommand{\xk}{\tilde{x}}
\newcommand{\bXk}{\tilde{X}}
\newcommand{\bxk}{\tilde{x}}
\newcommand{\eqd}{\stackrel{d}{=}}
\newcommand{\law}{\mathcal{L}}
\newcommand{\swap}{\textnormal{swap}}
\newcommand{\ch}{\textnormal{ch}}
\newcommand{\density}{f}
\newcommand\independent{\ci}

\usepackage{etoolbox}

\makeatletter
\AtBeginEnvironment{procedure}{\let\c@algocf\c@procedure}
\makeatother

\title{Metropolized Knockoff Sampling}
\author[1]{%\normalsize
Stephen Bates\thanks{Authors are listed in alphabetical order.}}
\author[1,2]{Emmanuel Cand\`es}
\author[3]{Lucas Janson}
\author[3]{Wenshuo Wang}
\date{\today}
\affil[1]{Department of Statistics, Stanford University, Stanford, CA 94305}
\affil[2]{Department of Mathematics, Stanford University, Stanford, CA 94305}
\affil[3]{Department of Statistics, Harvard University, Cambridge, MA 02138}

\begin{document}

\maketitle

\begin{abstract}
 % Contemporary scientific investigations routinely involve searching
 % large-scale data sets for a small number of true effects from many
 % candidates, creating the urgent need for variable selection
 % algorithms with statistical guarantees. 
Model-X knockoffs is a wrapper that transforms essentially any feature importance measure 
into a variable selection algorithm, which discovers true effects while rigorously controlling 
the expected fraction of false positives.  
  % in situations where  
  % nothing can be assumed about the dependence between the response and the explanatory variables. 
%  while only requiring the
%  generation of a set of fake synthetic variables known as {\em knockoffs}. 
%  These knockoffs can be used as negative controls. 
A frequently discussed  
challenge to apply this method is to construct knockoff variables, which are synthetic variables
obeying a crucial exchangeability property with the explanatory variables under study. This paper 
introduces techniques for knockoff generation in great
generality: we provide a sequential characterization of all
possible knockoff distributions, which leads to a Metropolis--Hastings
formulation of an {\em exact} knockoff sampler. We further show
how to use conditional independence structure to speed up computations.
Combining these two threads, we introduce an explicit set of sequential 
algorithms
  % {As I understand it, MCMC means performing monte carlo simulations using markov chains, i.e., the term is not just about a method but also what it is used for. So should we be a bit careful in using it this way? Perhaps we can be more clear that we are using tools or ideas from MCMC, but I don't think it's correct to say we have a `MCMC algorithm'} 
and empirically demonstrate their effectiveness. Our theoretical analysis 
proves that our algorithms achieve near-optimal computational complexity in certain cases. 
The techniques we develop are
sufficiently rich to enable knockoff sampling in challenging models
including cases where the covariates are continuous and heavy-tailed, and follow a graphical model such as the Ising model.
\end{abstract}


\section{Introduction}

In modern science, researchers often have access to large data sets
featuring comprehensive measurements about some phenomenon of
interest.  The question is then to discover meaningful relationships
between an outcome and all the measured covariates. While it is
often expected that only a small fraction of the covariates may be
associated with the outcome, the relevance of any particular variable
is unknown a priori. % A question of fundamental interest is then {\em
  % which variables are \underline{relevant}?}  
For instance, a researcher may be interested in understanding which
of the thousands of gene-expression profiles may help determine the
severity of a tumor. In such circumstances, the researcher often
relies on statistical algorithms to sift through large data sets and find those 
promising candidates,  
% \sbmargin{`blue chips'}{I do not understand what this phrase means in this context.}, 
making variable selection a topic of central importance in contemporary
statistical research.

The knockoff filter
\citep{barber2015controlling,candes2018panning} has recently emerged as a
useful framework for performing controlled variable selection,
allowing the user to convert any black-box feature importance measure
into a variable selection procedure while rigorously controlling 
the expected fraction of false positives. This means that the statistician 
can use essentially any black-box importance measure  to return a list of variables 
with the guarantee that, on the average, the ratio between the number of false positives---loosely 
speaking, a false positive is a variable that does not influence the response, 
see \citet{candes2018panning}---and the total number of reported variables is below a user-specified threshold. 
The strength of this method is that the guarantees hold in finite samples and 
in situations where nothing can be assumed about the dependence between the response 
and the explanatory variables. 
Instead, the statistician must have knowledge of the distribution of
the explanatory variables. When this happens to be the case, a remaining challenge is the ability to generate the {\em knockoffs}, 
a set of synthetic variables, which can essentially be used as negative controls; these fake variables  
must mimic the original variables in a
particular way without having any additional predictive power. 
In sum, constructing valid knockoff distributions and sampling mechanisms
across a wide range of covariate models is critical to deploying the knockoff filter in a number of applications.

\subsection{Our contribution}

This paper describes a theory for sampling knockoff variables and introduces a
general and efficient sampler inspired by ideas
from Markov chain Monte Carlo (MCMC). Before moving on, we pause to explicitly mention the two main considerations one should keep in mind when constructing knockoffs:  
\begin{description}
\item {\bf Computation.} How can we {\em efficiently} sample nontrivial knockoffs?
\item {\bf Statistical power.} How can we generate knockoffs that will
  ultimately lead to {\em powerful} variable selection procedures? On this
  note, it has been observed that knockoffs that are less
  correlated with the original variables lead to higher power
  \citep{barber2015controlling,candes2018panning} and, therefore, low correlation must be a
  design objective.
\end{description}
Having said that, our work makes several specific contributions.
\begin{enumerate}
\item {\bf Characterization of all knockoff distributions}. We
  provide a sequential characterization of \emph{every} valid
  knockoff distribution. Furthermore, we introduce a
  connection linking pairwise exchangeability between original and
  knockoff variables to reversible Markov chains, enabling the use of
  {powerful} %{this word seems a bit self-congratulatory here, and I'm not sure MCMC researcher would agree the tools we're using are that sophisticated...} 
  sampling tools from computational statistics.

\item {\bf Complexity of knockoff sampling procedures}. We introduce
  a class of algorithms which use conditional independence information
  to efficiently generate knockoffs. The computational complexity of
  such procedures is shown to be determined by the complexity of the
  dependence structure in a precise way. Furthermore, we present a
  lower bound on complexity showing that structural assumptions are
  necessary for efficient computation, and that our procedure achieves
  the lower bound in certain cases.

\item {\bf Practical sampling algorithms}. We develop a concrete
  knockoff sampler for a large number of
  distributions. This is achieved by constructing a family of MCMC
  tools---designed to have good performance---which only require the
  numerical evaluation of an \emph{unnormalized} density. We identify
  a default parameter setting for the sampler that performs well
  across a variety of situations, producing a general and easy-to-use
  tool for practitioners.
\end{enumerate}
We shall see that our ideas enable knockoff sampling in challenging models including situations where the
covariates are continuous and heavy-tailed and where they follow an
Ising model.

\subsection{Related literature}

This paper draws most heavily on \citet{candes2018panning},  which
builds on \citet{barber2015controlling} to introduce the model-X
knockoff framework. In particular, the former reference proposes the
\emph{Sequential Conditional Independent Pairs} (SCIP) procedure for
knockoff generation; this is the only known generic knockoff sampler to date, which shall serve 
as our starting point. The SCIP procedure,
however, is only abstractly specified and prior to this paper,
implementations were only available for Gaussian distributions and
discrete Markov chains. Briefly, \citet{sesia2018gene}
developed a concrete SCIP algorithm for discrete Markov chains, and then 
leveraged this construction to sample knockoffs for covariates
following hidden Markov models widely used in genome-wide association
studies.  Similarly relevant is the work of
\citet{gimenez2018knockoffs}, which developed a sampling strategy for
a restricted class of Bayesian networks, most notably Gaussian mixture
models. In contrast, we address here knockoff sampling for a much larger
class of distributions, namely, arbitrary graphical models. We also describe 
the form of all valid knockoff sampling strategies, thereby providing a framework possibly 
enabling the construction of future  knockoff sampling algorithms. Hence, our work may be of value to the increasing number of researchers 
deploying the knockoff
framework for feature selection in a variety of applications including
neural networks \citep{deeppink2018}, time-series modeling
\citep{ipad2018}, Gaussian graphical model structure learning
\citep{ZHENG2018201}, and biology \citep{Xiao2017,Gao2018}. 
Lastly, we close by emphasizing that our contribution is very different from a new strand of research introducing
approximate knockoffs generated with techniques from deep learning
\citep{romano2018deep,jordon2018knockoffgan,liu2018auto}. While these
approaches are tantalizing and demonstrate promising empirical
performance in low-dimensional situations, they currently lack formal
guarantees about their validity.

\section{Characterizing knockoff distributions}  \label{sec:scep}
%{\color{blue}Stephen: the reader doesn't know what Sequential conditional exchangeable pairs is yet, so the section title isn't providing useful singposting. Consider instead ``A universal knockoff sampler'' or ``Characterizing knockoff distributions''.}

%\subsection{The SCEP algorithm}
%It is shown in \citet{candes2018panning} that the SCIP algorithm generates valid knockoffs. Mathematically, this means if $\tilde X_j \mid  \Xno{j},\tilde{ X}_{1:j-1}\eqd X_j \mid  \Xno{j},\tilde{ X}_{1:j-1}$, and $\tilde X_j\ci X_j \mid  \Xno{j},\tilde{ X}_{1:j-1}$ for all $j$, we will have pairwise exchangebility of $(X_1, \dots, X_p, \tilde X_1, \dots, \tilde X_p)$. However, since our goal is pairwise exchangeability, it's natural to ask if the conditional independence requirement of $X_j$ and $\tilde X_j$ in each step can be relaxed. We answer this question in this section and introduce the \emph{Sequential Conditional Exchangeable Pairs} (SCEP) algorithm.

%more direct

\subsection{Knockoff variables}
Consider random covariates $X = (X_1, X_2, \ldots, X_p)$. We say
that the random variables $\tilde{X} = (\Xk_1, \Xk_2, \ldots, \Xk_p)$ are {\em knockoffs} for $X$ if
for each $j = 1, \ldots, p$,
\begin{align} \label{eq:knockoff-swap}
({X},\bXk)_{\swap(j)} \eqd ({X},\bXk).
\end{align}
Here, the notation $\swap(j)$ means permuting $X_j$ and $\Xk_j$; for
instance, $(X_1,X_2,X_3, \Xk_1, \Xk_2, \Xk_3)_{\swap(2)}$ is the vector 
$(X_1,\Xk_2,X_3, \Xk_1, X_2, \Xk_3)$.\footnote{In the presence of a
  response $Y$, we also require $\bXk \ci Y \mid  X$, which is
  easily satisfied by procedures that generate $\bXk$ from $ X$
  without looking at $ Y$.} Property (\ref{eq:knockoff-swap}) is
known as the {\em pairwise exchangeability} property, and it is in
general challenging to define joint distributions $(X, \tilde{X})$ satisfying this condition. Before continuing, we briefly pause to understand the
meaning of pairwise exchangeability.  A consequence of
\eqref{eq:knockoff-swap} is that for all sets
$A \subseteq \{1,\dots,p\}$,
\begin{align*} \label{eq:knockoff-multi-swap}
({X},\bXk)_{\swap(A)} \eqd ({X},\bXk),
\end{align*}
where $({X},\bXk)_{\swap(A)}$ denotes the swapping of $X_j$ and
$\Xk_j$ for all $j \in A$. Taking $A = \{1,\dots,p\}$ and marginalizing, we immediately
see that $\bXk \eqd X$; that is, $X$ and $\Xk$ are distributed in the
same way. Also changing any subset of entries of $ X$ with their
knockoff counterparts does not change the distribution either.
%covariance matrix
Another consequence of the exchangeability property
\eqref{eq:knockoff-swap} is that the mixed second moments of
$({X},\bXk)$ must match. Assume the second moments of $X$ exist and
write $\bm\Sigma = \Cov(X)$. Then the covariance of the vector
$(X,\Xk)$ must take the form
\begin{equation}
\Cov( X,\bXk)=\bm\Gamma(s)  := \left[
  \begin{array}{cc}
\bm\Sigma & \bm\Sigma-\diag(s)\\
\bm\Sigma-\diag(s) & \bm\Sigma\\
  \end{array}
\right], 
\label{eq:knockoff-cov}
\end{equation}
where $s \in \mathbb{R}^p$ is any vector such that the right-hand side is
positive semi-definite. In other words, for each pair $(i,j)$ with
$i \neq j$, we have $\Cov(X_i, \Xk_j) = \Cov(X_i, X_j)$.

We are interested in constructing knockoff variables and below we call a {\em knockoff sampler} a procedure that takes as inputs a distribution $\p$ and a sample $X \sim \p$ and returns $\Xk$ such that \eqref{eq:knockoff-swap} holds. Nontrivial samplers have been demonstrated in a few cases, for instance, when 
$X \sim \mathcal N( 0, \bm\Sigma)$ is multivariate Gaussian. In this case, 
\citet{candes2018panning} show that if $({X},\bXk)$ is jointly
Gaussian with mean zero and covariance $\bm\Gamma(s)$, then the entries of $\bXk$ are knockoffs for $X$. One can say that appropriately matching the first two moments
is sufficient to generate knockoffs in the special case of the
multivariate normal distribution. However, this does not extend and 
matching the first two moments is in general not sufficient; to be
sure, \eqref{eq:knockoff-swap} requires that all moments match
appropriately.

\paragraph{Gibbs measures.}
As a motivating example, consider the Ising model, a frequently discussed family of Gibbs measures first introduced in the statistical physics
literature \citep{ising1925beitrag}. In this model, the random vector $X \in \{-1,1\}^{d_1 \times d_2}$ defined over a $d_1 \times d_2$ 
$ X \in \{-1,1\}^{d_1 \times d_2}$ grid has 
a probability mass function (PMF) of the form 
\begin{equation} \label{eq:ising-pmf}
\p( X) = \frac{1}{Z(\beta, \alpha)} \exp\left(\sum_{\substack{s, t \in \mathcal I \\\|s-t\|_1=1}} \beta_{st} X_{s} X_{t} + \sum_{s\in \mathcal I}\alpha_{s} X_{s} \right);
\end{equation}
here, $\mathcal I = \{(i_1,i_2) : 1\le i_1 \le d_1,1\le i_2 \le d_2\}$
is the grid and $\alpha$ and $\beta$ are
parameters. As we have seen, knockoffs $\bXk$ for $ X$ must
marginally follow the Ising distribution
\eqref{eq:ising-pmf}. Furthermore, $\bXk$ must be dependent on $ X$ in
such a way that any vector of the form 
$\{(Z_1,\dots,Z_p) : Z_j = X_j\text{ or } Z_j = \Xk_j, 1\le j\le p\}$
has PMF given by \eqref{eq:ising-pmf}. It is tempting to na{\"i}vely define a joint PMF for  $(X,\Xk)$ as
\begin{equation*}
\p(X,\Xk) \propto \exp\left(\sum_{\substack{s, t \in \mathcal I \\\|s-t\|_1=1}} \beta_{st} (X_{s} X_{t}+\Xk_s\Xk_t+X_s\Xk_t+\Xk_sX_t) + \sum_{s\in \mathcal I}\alpha_{s} (X_{s}+\Xk_s) \right). 
\end{equation*}
Although the joint distribution is symmetric in $X_s$ and $\Xk_s$, the marginal 
distribution of $X$ is not an Ising model!
Hence, this is not a valid joint distribution. Other than the trivial
construction $\bXk =  X$, it is a priori unclear how one would construct knockoffs. Any distribution continuous or discrete factoring over a grid poses a similar challenge.

\subsection{SCIP and its limitations}
%\tealcomment{EC: Move this before SCEP}
\label{subsec:knockoffconstruction}
The only generic knockoff sampler one can find in
the literature is SCIP from \citet{candes2018panning},
given in Procedure \ref{alg:sequential}. While this procedure provably
generates valid knockoffs for any input distribution, there are two
substantial limitations. The first is that SCIP is only given
abstractly; it is challenging to specify $\law(X_j \mid \Xno{j}, \, \bXk_{1:(j-1)})$,\footnote{We use
  $\law(W_1 \mid W_2)$ to denote the conditional distribution of $W_1$
  given $W_2$. We use the subscript $1:0$ to mean an empty
  vector.} let alone to sample from it.  As a result, it is only known how to implement SCIP for
very special models such as discrete Markov chains and Gaussian
distributions. The second limitation is that SCIP is not able to
generate all valid knockoff distributions. Recall that we want 
knockoffs to have low correlations with the
original variables so that a feature importance statistic will
correctly detect true effects. To achieve this goal, we might need a wider range of sampling mechanisms. 

\begin{procedure}[h]
\caption{Sequential Conditional Independent Pairs \space (SCIP)\label{alg:sequential}}
	\SetAlgoLined\DontPrintSemicolon
	{%$j = 1$ \;
		\For{$j=1$ \KwTo $p$} {
			Sample $\Xk_j$ from $\law(X_j \, \mid  \, \Xno{j}, \, \bXk_{1:(j-1)})$, conditionally independently from $X_j$\; 
			%$j = j + 1$\;
		}
	}
\end{procedure}

% Admittedly, SCEP is currently just as abstract, so much of the remainder work is devoted to understanding the computational limits of knockoff generation and explicitly developing efficient algorithms. 

\subsection{Sequential formulation of knockoff distributions}
We begin by introducing a sequential characterization of \emph{all} valid knockoff distributions, which will later lead to a new class of knockoff samplers.

%\subsubsection*{Sequential Characterization}
\begin{thm}[Sequential characterization of knockoff distributions]
\label{theorem:seq}
Let $( X, \bXk)\in\rr^{2p}$ be a random vector. Then pairwise exchangeability \eqref{eq:knockoff-swap} holds if and only if both of the following conditions hold:
% \begin{enumerate}
% \item (Pairwise exchangeability) for all $1\le j\le p$,
% \begin{equation}
% (X_1, \dots, X_p, \tilde X_1, \dots, \tilde X_p)\eqd(X_1, \dots, X_p, \tilde X_1, \dots, \tilde X_p)_{\emph{swap}(j)}.
% \label{eq:cond1}
% \end{equation}
% \item For all $1\le j\le p$,
\begin{description}
\item[Conditional exchangeability] For each $j \in \{1, \ldots, p\}$, 
\begin{equation}
\label{eq:cond-exch}
(X_j, \tilde X_j)\mid \Xno{j}, \tilde{ X}_{1:(j-1)}\eqd(\tilde X_j, X_j)\mid \Xno{j}, \tilde{ X}_{1:(j-1)}.
\end{equation}
\item[Knockoff symmetry] For each $j \in \{1, \ldots, p\}$, 
\begin{equation} \label{eq:knockoff-symmetry-condition}
\p((X_j, \tilde X_j)\in A\mid \Xno{j}, \tilde{ X}_{1:(j-1)})
\end{equation}
is
$\sigma(X_{(j+1):p},\{X_1,\tilde X_1\},\dots,\{X_{j-1},\tilde
X_{j-1}\})$-measurable for any Borel set $A$, where $\{\cdot,\cdot\}$
denotes the unordered pair. That is, the conditional distribution does not change if we swap previously sampled knockoffs with the original features.
\end{description}
\end{thm}

Theorem \ref{theorem:seq} implies that a sequential knockoff sampling
algorithm faithful to these two conditions is as general as it
gets. The challenge now becomes creating exchangeable random variables
at each step (with a little caution on the dependence on the previous
pairs of variables). In turn, this task happens to be equivalent to
designing a time-reversible Markov chain, as formalized below.
\begin{prop}
A pair of random variables $(Z, \tilde Z)$ is exchangeable, i.e., $(Z, \tilde Z)\eqd(\tilde Z, Z)$, with marginal distribution $\pi$ for $Z$---and, therefore, for $\tilde Z$ as well---if and only if there exists a time-reversible Markov chain $\{Z_n\}_{n=1}^\infty$ such that $Z_1\sim\pi$ is a stationary distribution of the chain, and $(Z_1, Z_2)\eqd(Z,\tilde Z)$.
\label{prop:trMarkov}
%\tealcomment{EC: Don't use $X$, $\Xk$ as scalars here. They are vectors everywhere else in the text.}
\end{prop}

Combining these two results gives SCEP (Procedure \ref{alg:scep} below), which is a
completely general strategy for generating knockoffs: at each step
$j$, we design a time-reversible Markov chain with stationary
distribution $\mathcal L(X_j\mid \Xno{j},\tilde{ X}_{1:(j-1)})$, and
draw a sample by taking one step of this chain starting from $X_j$.
Proposition \ref{prop:trMarkov} implies that the conditional
exchangeability \eqref{eq:cond-exch} holds. Furthermore, the symmetry requirement
on the transition kernel implies that SCEP
does not break the exchangeability from previous steps; that
is, the knockoff symmetry \eqref{eq:knockoff-symmetry-condition} also
holds. Theorem \ref{theorem:seq} then implies that such a procedure
produces valid knockoffs.

\begin{procedure}[H]
\caption{Sequential Conditional Exchangeable Pairs \space (SCEP)\label{alg:scep}}
\SetAlgoLined\DontPrintSemicolon {%$j = 1$ \;
\For{$j=1$ \KwTo $p$} {
    Sample $\Xk_j$ by taking one step of a time-reversible Markov
    chain starting from $X_j$.\\
    The transition kernel must be such that it depends only on
    $X_{(j+1):p}$ and the unordered pairs
    $\{X_1,\tilde X_1\},\dots,\{X_{j-1},\tilde X_{j-1}\}$, and admits
    $\mathcal L(X_j\mid \Xno{j},\tilde{ X}_{1:(j-1)})$ as a
    stationary distribution.  \;} }
\end{procedure}

To rehearse the universality of SCEP, consider an arbitrary knockoff sampler producing $\Xk_1, \ldots, \Xk_p$. Then from Theorem \ref{theorem:seq} we know that $X_1$ and $\tilde{X}_1$ must be exchangeable conditional on $\Xno{1}$. Therefore, $\Xk_1$ may be sampled by taking one step of a reversible Markov chain starting at $X_1$. Moving on to $X_2$, Theorem 2 informs us that $X_2$ and $\Xk_2$ are exchangeable conditional on $\{X_1, \Xk_1\}, X_3, \dots, X_p$, so $\Xk_2$ can again be viewed as taking one step of a reversible Markov chain starting at $X_2$. Continuing in this fashion for $j = 3, \ldots, p$  establishes our claim. 

SCEP as stated remains too abstract to be considered an implementable algorithm, so we will next develop a concrete version of this procedure. Although this may not yet be clear, we would like to stress that formulating a knockoff sampler in terms of
reversible Markov chains is an important step forward because it will
ultimately enable the use of flexible MCMC tools.

\section{The Metropolized knockoff sampler}

We now demonstrate how one can generate knockoffs in a sequential manner by making proposals which are either accepted or rejected in a Metropolis--Hastings-like fashion as to ensure pairwise exchangeability.

\subsection{Algorithm description}
\label{subsec:mhscep}

% We begin by introducing the MH idea, and for pedagogical purposes,
% proceed by showing why a na\"ive implementation would not work.
% Proposed in \citet{metropolis1953equation} and \citet{hastings1970monte}, 

The celebrated  Metropolis--Hastings (MH) algorithm \citep{metropolis1953equation,hastings1970monte}
provides a general
time-reversible Markov transition kernel whose stationary distribution
is an arbitrary density function $\pi$. To 
construct a transition from $x$ to $y$, MH operates as follows: generate a proposal
$x^*$ from a distribution $q(\cdot \mid x)$ (any distribution
depending on $x$) and set\footnote{More generally, we take as
  acceptance probability $\gamma \, \alpha$ with $\gamma\in(0,1]$. In
  this work, $\gamma$ is set to $1$ as default, except in Section
  \ref{subsec:MTM} and Appendix \ref{subapp:effect-gamma}, which
  are cases where tuning $\gamma$ is recommended.}
\[
  y = \begin{cases} x^* & \text{with prob. } \alpha,\\
    x & \text{with prob. } 1- \alpha,
  \end{cases}
  \qquad \alpha = 
    \min\left(1, \frac{\pi(x^*)q(x \mid x^*)}{\pi(x)q(x^* \mid
  x)}\right).
\]
%Equivalently, the kernel
%\begin{equation}
%\begin{aligned}
%\kappa(x,\di y)&=q(y \mid  x)\min\left(1,\frac{\pi(y)q(x \mid  y)}{\pi(x)q(y \mid  x)}\right)\di y\\
%&\quad+\delta(y-x)\di y\int q(x^* \mid  x)\left(1-\min(1,\frac{\pi(x^*)q(x \mid  x^*)}{\pi(x)q(x^* \mid  x)})\right)\di{x^*}\\
%\label{eq:MetHas}
%&\tealcomment{EC: this should be written more algorithmically.}
%\end{aligned}
%\end{equation}
%has $\pi$ as a stationary distribution and is time-reversible.
This can be implemented even when the density $\pi$ is unnormalized,
as the normalizing constants cancel. In
our setting, we shall make sure that the choice of the proposal
distribution depends on the previously sampled pairs in a symmetric
fashion, thereby remaining faithful to the knockoff symmetry condition
\eqref{eq:knockoff-symmetry-condition} in Theorem \ref{theorem:seq}. As such, we call such proposals \emph{faithful}.

Consider now running SCEP (Procedure \ref{alg:scep}) with the MH
kernel, where at the $j$th step, the target distribution $\pi$ is
taken to be $\law(X_j \mid \Xno{j}, \bXk_{1:j-1})$. The issue with such a na\"ive implementation 
is that the target $\pi$ cannot be readily evaluated. To understand why this is the case, set $j = 2$ and
consider $\law(X_2 \mid \Xno{2}, \bXk_{1})$. This distribution has density
proportional to
$\p( X= x) \p(\tilde X_1=\tilde x_1\mid  X= x)$, which is
equal to
\begin{multline}
        \p( X= x)\bigg[q(\xk_1\mid x_1)\min\left(1,\frac{q(x_1\mid \xk_1)\p(X_1=\xk_1,\Xno{1}=\xno{1})}{q(\xk_1\mid x_1)\p(X_1=x_1,\Xno{1}=\xno{1})}\right) \\
    +\delta(\xk_1-x_1)\int q(x^* \mid  x_1)\left(1-\min\left(1,\frac{q(x_1\mid x^*)\p(X_1=x^*,\Xno{1}=\xno{1})}{q(x^*\mid x_1)\p(X_1=x_1,\Xno{1}=\xno{1})}\right)\right)\di{x^*}\bigg]. \label{eq:MH-explicit}
  \end{multline}
  The first term in the summation within the brackets corresponds
  to the acceptance case while the second corresponds to the rejection
  case. This latter term cannot be evaluated because of the
  integral over $x^*$. Hence, the target density cannot be evaluated either.

% In order to compute this density,
% the rejection probability must be known, but this is not accessible
% because the rejection probability requires marginalizing out the
% proposal, i.e., it is the intractable integral
% \begin{equation*}
% \int q(x^* \mid  x)\left(1-\min\left(1,\frac{\pi(x^*)q(x \mid  x^*)}{\pi(x)q(x^* \mid  x)}\right)\right)\di{x^*}.
% \end{equation*} For example,
% \begin{align}
%     &\quad \ \p(X_2 =x_2|\Xno{2}=\xno{2},\tilde X_1=\xk_1) \nonumber \\
%     &\propto \p( X= x) \p(\tilde X_1=\tilde x_1\mid  X= x) \nonumber \\
%     %&=\p(X=x) [\p(\Xk_1 = \xk_1 | X = x, I_{\text{accept}}) \p(I_{\text{accept}} | X = x) + \p(\Xk_1 = \xk_1 | X = x, I_{\text{reject}}) \p(I_{\text{reject}} | X = x)] \\ 
%     &=\p( X= x)\bigg[q(\xk_1\mid x_1)\min\left(1,\frac{q(x_1\mid \xk_1)\p(X_1=\xk_1,\Xno{1}=\xno{1})}{q(\xk_1\mid x_1)\p(X_1=x_1,\Xno{1}=\xno{1})}\right) \nonumber \\
%     &\quad\quad\quad\quad\quad\  +\delta(\xk_1-x_1)\int q(x^* \mid  x_1)\left(1-\min\left(1,\frac{q(x_1\mid x^*)\p(X_1=x^*,\Xno{1}=\xno{1})}{q(x^*\mid x_1)\p(X_1=x_1,\Xno{1}=\xno{1})}\right)\right)\di{x^*}\bigg]. \label{eq:MH-explicit}
% \end{align}
% The first term in the summation in \eqref{eq:MH-explicit} is the acceptance case and the second term is the rejection case, which cannot be evaluated because of the integral.

% %\tealcomment{We need to be more conceptual here. Explain why we have exchangeability at each step and how we retain the exchangeability from previous steps. Move the statement of the algorithm after the conceptual derivation.}

  We propose an effective solution to this problem: {\em condition on the proposals}
  and at step $j$, let the target distribution be
  $\law(X_j \mid \Xno{j}, \bXk_{1:j-1},  X^*_{1:j-1})$ rather than $\law(X_j \mid \Xno{j}, \bXk_{1:j-1})$. This has the effect of removing 
  the integral and makes computing the rejection probability
  tractable. This is best seen by returning to our example where $j = 2$. Here, $\law(X_j \mid \Xno{j}, \bXk_{1:j-1},  X^*_{1:j-1})$ has density now proportional to 
\begin{multline}
\p( X= x)q(x_1^*\mid x_1)\bigg[\delta(\xk_1 - x_1^*)
    \min\left(1,\frac{q(x_1\mid \xk_1)\p(X_1=\xk_1,\Xno{1}=\xno{1})}{q(\xk_1\mid x_1)\p(X_1=x_1,\Xno{1}=\xno{1})}\right)  \\
   + \delta(\xk_1-x_1)\left(1-\min\left(1,\frac{q(x_1\mid x^*_1)\p(X_1=x^*_1,\Xno{1}=\xno{1})}{q(x^*_1\mid x_1)\p(X_1=x_1,\Xno{1}=\xno{1})}\right)\right)\bigg]. \label{eq:MH-SCEP-rejection-prob}
\end{multline}
We will show in Section \ref{sec:graph-structure} how such terms can
be efficiently computed. Leaving aside implementation details for the
moment, this strategy leads to Algorithm \ref{alg:finalscep}. Here and
elsewhere, $\p$ denotes the density of the variables under study, or
formally, the Radon--Nikodym derivative with respect to a common
dominating measure.
\begin{algorithm}[h]
\caption{Metropolized knockoff sampling (Metro).\label{alg:finalscep}}
	\SetAlgoLined\DontPrintSemicolon
	{%Order the variables with Algorithm \ref{alg:jt-scep} \;
	%$j = 1$ \;
		\For{$j=1$ \KwTo $p$} {
			Sample $X^*_j=x^*_j$ from a faithful proposal distribution $q_j$.\;
			Accept the proposal with probability\; $\hspace{1cm}\min\left(1,\frac{q_j(x_j \mid  x^*_j) \p\left(\Xno{j}=\xno{j},X_j=x^*_j,\tilde{ X}_{1:(j-1)}=\tilde{ x}_{1:(j-1)}, X^*_{1:(j-1)}= x^*_{1:(j-1)}\right)}{q_j(x^*_j \mid  x_j)\p\left(\Xno{j}=\xno{j},X_j=x_j,\tilde{ X}_{1:(j-1)}=\tilde{ x}_{1:(j-1)}, X^*_{1:(j-1)}= x^*_{1:(j-1)}\right)}\right)$.\;
			Upon acceptance, set $\tilde x_j=x^*_j$; otherwise, set $\tilde x_j=x_j$.\;
			%$j = j + 1$\;
		}
	}
	Return $\tilde{ X}=(\tilde x_1,\tilde x_2,\dots,\tilde x_p)$
\end{algorithm}
%\tealcomment{EC: Move the formal statement of the algorithm to the end of the exposition.}

At this point, it should be clear that Metropolized knockoff sampling generates
exact knockoffs, a fact we formally record below. 
\begin{coro}
Metropolized knockoff sampling (Metro) produces valid knockoffs.
\label{coro:alg1}
\end{coro}
\begin{proof}
For the sake of the
proof, let $U_j$ be the indicator of acceptance at step $j$, and $Z_j=(1-U_j)X^*_j$. We will prove pairwise exchangeability jointly with the $U_j$'s and $Z_j$'s; marginalizing out these variables will establish the claim. For $1\le j\le p$, let $\density_j(x_j,\xno{j},\tilde{ x}_{1:(j-1)}, u_{1:(j-1)},z_{1:(j-1)})$ be the joint density function of
$(X_j, \Xno{j},\tilde{ X}_{1:(j-1)}, U_{1:(j-1)},Z_{1:(j-1)})$, in this 
order. We will use induction to show that the density of $(X,\Xk_{1:j},U_{1:j},Z_{1:j})$ is symmetric in $X_k$ and $\Xk_k$ for $1\le k\le j$. For $1 \le j\le p$, the inductive hypothesis is that $\density_j$ is symmetric in $x_k$ and $\xk_k$ for $1\le k\le j-1$ (since $\density_j$ is just the density of $(X,\Xk_{1:(j-1)},U_{1:(j-1)},Z_{1:(j-1)})$ after reordering the variables). For $1\le j\le p$,
\begin{equation*}
\begin{aligned}
\text{the } &\text{densi}\text{ty} \text{ of }(X,\Xk_{1:j},U_{1:j},Z_{1:j})\text{ at }(x,\xk_{1:j},u_{1:j},z_{1:j})\\
=\density_j(&x_j,\xno{j},\tilde{ x}_{1:(j-1)}, u_{1:(j-1)},z_{1:(j-1)})\times\\
\Bigg[&\ind_{u_j=1}\delta(z_j-0)q_j(\tilde x_j \mid  x_j)\min\left(1,\frac{\density_j(\tilde x_j,\xno{j},\tilde{ x}_{1:(j-1)},  u_{1:(j-1)},z_{1:(j-1)})q_j(x_j \mid  \tilde x_j)}{\density_j(x_j,\xno{j},\tilde{ x}_{1:(j-1)}, u_{1:(j-1)},z_{1:(j-1)})q_j(\tilde x_j \mid  x_j)}\right)\\
&+\ind_{u_j=0}\delta(\xk_j - x_j)q_j(z_j \mid  x_j)\Bigg(1-\min\left(1,\frac{\density_j(z_j,\xno{j},\tilde{ x}_{1:(j-1)},  u_{1:(j-1)},z_{1:(j-1)})q_j(x_j \mid  z_j)}{\density_j(x_j,\xno{j},\tilde{ x}_{1:(j-1)}, u_{1:(j-1)},z_{1:(j-1)})q_j(z_j \mid  x_j)}\right)\Bigg)\Bigg],
\end{aligned}
\end{equation*}
which is symmetric in the first $j-1$ pairs by the inductive hypothesis. For the symmetry in the $j$th pair, when $u_j=1$, the density simplifies to
\begin{equation*}
\begin{aligned}
\delta(z_j-0)\times\min\big(&\density_j(x_j,\xno{j},\tilde{ x}_{1:(j-1)}, u_{1:(j-1)},z_{1:(j-1)})q_j(\tilde x_j \mid  x_j),\\
&\density_j(\tilde x_j,\xno{j},\tilde{ x}_{1:(j-1)},  u_{1:(j-1)},z_{1:(j-1)})q_j(x_j \mid  \tilde x_j)\big),
\end{aligned}
\end{equation*}
which is invariant to swapping $x_j$ and $\xk_j$; when $u_j=0$, the delta function $\delta(\xk_j-x_j)$ ensures $x_j=\xk_j$, and thus swapping them has no effect.
Hence, when the algorithm terminates, all pairs are exchangeable and therefore remain exchangeable after marginalizing out the $U_j$'s and $Z_j$'s.
\end{proof}

Anticipating possible future applications, we wish to remark that
Metro can be easily adapted to sampling
\emph{group} knockoffs \citep{pmlr-v48-daia16}; see
Appendix~\ref{app:grpko}.

%{\color{blue} Stephen: What do you mean by this?} {\color{red} Wenshuo: I just wanted to mention that if there is no rejection we wouldn't need to condition. Conditioning restricts our freedom. Maybe there's a better way to phrase it?} {\color{blue} I still do not fully understand. Even if the procedure rejects, we still need to condition on the proposal and the rejection. E.g. if we have $p=2$, then $P(X_2 = y \mid   X_1 = x, \Xk_1 = x) \ne P(X_2 = y \mid   X_1 = x)$.} {\color{red} Wenshuo: yes, we need to condition on them when we reject, which limits the power. What't I'm saying is if we accept we don't lose any power. Think of it this way: if we know we will never reject, we don't have to do the extra condition.} {\color{blue}. I see, my mistake. I think it's fine as written.}

%[\textbf{Metropolis--Hastings}]

%\section{Choosing a proposal distribution}
%\label{sec:good-prop}

% This section gives three classes of proposal distributions that work well for knockoff sampling, moving from most specific to most generally applicable.

\subsection{Covariance-guided proposals}
\label{subsesc:cov-guided-prop}
Now that we have available a broad class of knockoff samplers, we turn to the question of finding faithful proposal distributions that will generate statistically powerful knockoffs. The overall challenge is to propose samples that are far away from $X$ to make good knockoffs, but not as far that they are systematically rejected. A rejection at the $j$th step gives $\Xk_j = X_j$, leading to a knockoff with poor contrast. Below, we shall borrow ideas from existing knockoff samplers for Gaussian models to make sensible proposals. 
%In general, we cannot hope to construct a rejection-free proposal distribution, so in this section we introduce a technique for using the covariance of $ X$ to make proposals that will be accepted with high-probability. 

Suppose that $ X$ has mean $\mu$ and covariance $\bm\Sigma$, and consider $s \in \mathbb{R}^p$ with non-negative entries such that $\bm\Gamma(s)$ from \eqref{eq:knockoff-cov} is positive semidefinite. Such a vector $ s$ can be found with techniques from \citet{barber2015controlling} and from \citet{candes2018panning}. We have seen that if $X$ were Gaussian, this covariance matrix would induce a multivariate Gaussian joint distribution over $X$ and $\Xk$ with the correct symmetry. In non-Gaussian settings, our observation is that we can still make proposals as if the variables were Gaussian, but use the MH correction to guarantee exact conditional exchangeability. This can be viewed as a Metropolis-adjustment to the second-order knockoff construction of \citet{candes2018panning}. Concretely, the distribution $q_j$ for a covariance-guided proposal---used to generate a proposal $X_j^*$---is normal with mean 
% \ejc{Shouldn't it be $\Xk$ below?}\tealcomment{SB: either version is valid. Wenshuo's simulations are better with this version.} {\color{blue} [WW: What Stephen said, plus this version is more "Metropolis-adjustment to the second-order knockoff", as it literally proposes according to a multivariate Gaussian and then accept/reject them one by one.]}
\begin{equation*}
\mu_{j} + \left(\Gamma_{12}^{(j)}\right)^\top \left(\bm\Gamma_{11}^{(j)}\right)^\dagger
\left(
 X-\mu, X^*_{1:(j-1)}-\mu_{1:(j-1)}\right)^\top
 \end{equation*}
 and variance
 \begin{equation*}
 \Gamma_{22}^{(j)} - \left(\Gamma_{12}^{(j)}\right)^\top\left(\bm\Gamma_{11}^{(j)}\right)^\dagger  \Gamma_{12}^{(j)};
\end{equation*}
here, $X^*_{1:(j-1)}$ is the sequence of already generated proposals, $\bm\Gamma_{11}^{(j)}=\bm\Gamma_{1:(p+j-1), 1:(p+j-1)}$, $\Gamma^{(j)}_{22}=\Gamma_{p+j,p+j}$, $\Gamma^{(j)}_{12}=\Gamma_{1:(p+j-1),p+j}$, $\mu$ is the mean of $ X$, and $\dagger$ stands for  the pseudoinverse. The parameters of $q_j$ can be efficiently computed using the special structure of $\bm\Gamma$; see Appendix \ref{app:mat-inv}. The faithfulness of the proposal is shown in Appendix \ref{app:cov-guide}.

The covariance-guided proposals are valid even when $\bm\Sigma$ is replaced by any other positive semidefinite matrix---\emph{any faithful proposal distribution will give valid knockoffs}. This allows us to use an empirical estimate of $\Cov(X)$ based on simulated samples from $\law(X)$, or even to apply the covariance-guided proposals to discrete distributions by rounding each proposal $X^*_j$ to the nearest point in the support of $X_j$.  These proposals will be most successful when $ X$ is well-approximated by a Gaussian density, indeed when $ X$ is exactly Gaussian and the true covariance is used, the covariance-guided proposals will never be rejected.  Numerical simulations in a variety of settings can be found in Section \ref{sec:simulations}.
%This, of course, still doesn't solve the time complexity issue. But with Markov chains, it is doable in polynomial time.

\subsection{Multiple-try Metropolis}
\label{subsec:MTM}
%Lastly, we introduce a very general method for generating good proposals. 
A possibility for sampling $\Xk_j$ ``far away'' from $X_j$ is to run multiple MH steps instead of a single one. The issue with this is that this would make the conditional distributions from Metro prohibitively complex at later steps. Longer chains also require conditioning later proposals on a longer sequence of proposals and acceptances or rejections, which will constrain those proposals to be closer to their corresponding true variables and thus reduce power. Instead, we use the multiple-try Metropolis (MTM) technique introduced in \citet{liu2000multiple}. 

The key idea of MTM is to propose a set of several candidate moves in order to increase the probability of acceptance. As in \citet{qin2001multipoint}, we take the candidate set to be $C^{m,t}_x=\{x\pm kt:1\le k\le m\}$, where $m$ is a positive integer and $t$ is a positive number; see Figure \ref{figure:MTM} for an illustration. MTM proceeds by choosing one element $x^*$ from the set $C^{m,t}_x$, with probability proportional to the target density, i.e.,
\begin{equation}
\p(\text{select }x^*\text{ from }C^{m,t}_x)=\frac{\pi(x^*)}{\sum_{u\in C^{m,t}_x}\pi(u)}.
\label{eq:selectionprob}
\end{equation}
This proposal is then accepted with probability
\begin{equation}
\gamma\min\left(1,\frac{\sum_{u\in C^{m,t}_x}\pi(u)}{\sum_{v\in C^{m,t}_{x^*}}\pi(v)}\right),\quad\gamma\in(0,1),
\label{eq:acceptanceprob}
\end{equation}
where $\gamma$ is an additional tuning parameter explained in Appendix \ref{subapp:effect-gamma}. This parameter should be taken to be near $1$ in most settings. If no element of $C^{m,t}_x$ has positive probability, then one automatically rejects. MTM is a special case of MH with the proposal $q(x^* \mid x)$ distribution defined implicitly by the above rules, and furthermore, the proposals are faithful.
% \sout{since the proposal distribution at step $j$ only depends on the conditional density of $X_j\mid \Xno{j},X^*_{1:(j-1)},\Xk_{1:(j-1)}$ which can be easily shown to be symmetric in the first $j-1$ pairs.}} 
Thus, MTM can be used in Metro. 

%Faithfulness and consistency (Definition \ref{def:consistent}) are straightforward since the proposal only depends on the target density.
%Although Theorem \ref{thm:scep-runtime} \ejc{has not been introduced} does not apply in this case because the proposals require evaluating $\Phi$, Algorithm~\ref{alg:finalscep} with MTM proposals can be run with $O(p(3m+1)^w)$ queries of $\Phi$ \ejc{Queries of or queries to?}; see Appendix \ref{subapp:complexity-proofs}.

\begin{figure}[h]
\centering
\includegraphics[width=0.6\linewidth]{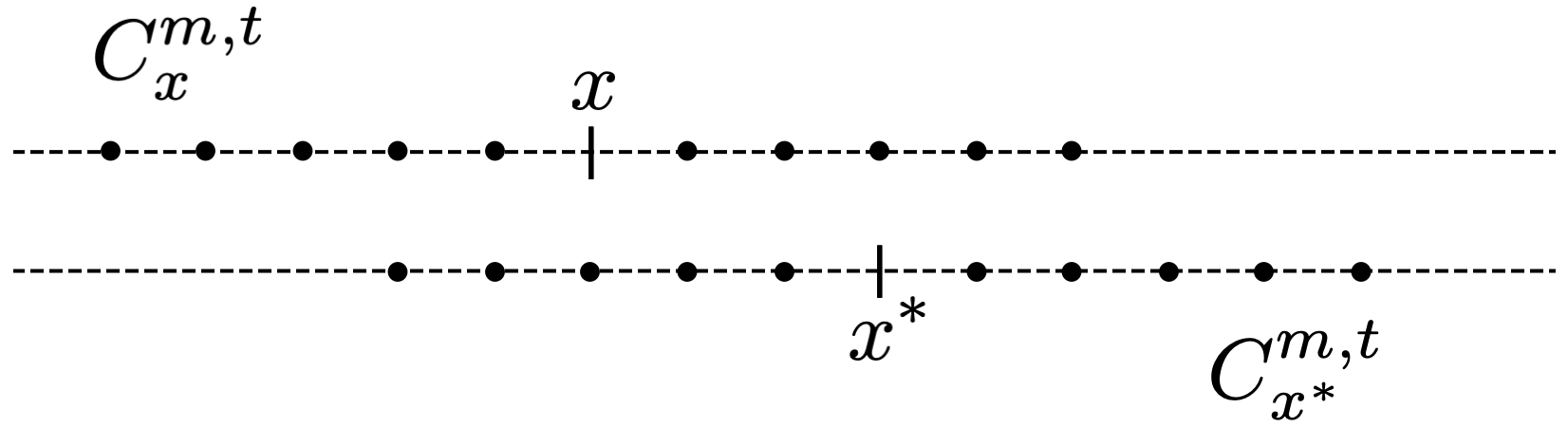}
\caption{Multiple-try Metropolis (adapted from Figure 2 in \citet{qin2001multipoint}).}
\label{figure:MTM} 
\end{figure}

While there is no universally optimal combination of $m$ and $t$, we provide guidance about default values based on our experimental results from Section \ref{sec:simulations}. To understand the choice of parameters, first observe that with a fixed $t$, the knockoff distribution produced by the algorithm should eventually stabilize as $m$ grows to infinity, since (non-pathological) $\pi$ will vanish at positive and negative infinities and equations \eqref{eq:selectionprob} and \eqref{eq:acceptanceprob} will converge. Large values of $m$ require more density evaluations, so we would like to choose the smallest value of $m$ such that the distribution defined by equations \eqref{eq:selectionprob} and \eqref{eq:acceptanceprob} is nearly converged to its limit as $m\rightarrow\infty$. Turning our attention to $t$, smaller values cause higher acceptance rates, and at the same time, encourage $\tilde X_j$ to be close to $X_j$. Clearly, there is a trade-off. Based on our experiments, a sensible default setting is $m=4$ and $t_j = 1.5\sqrt{1/(\bm\Sigma^{-1})_{jj}}$ where $\bm\Sigma = \Cov( X)$. In the Gaussian case, $\Var(X_j |  \Xno{j}) = 1/(\bm\Sigma^{-1})_{jj}$ for any observed value of $\Xno{j}$ \citep{anderson2009introduction}, hence this   choice of scaling is intuitive. In the non-Gaussian case $1 / (\bm\Sigma^{-1})_{jj}$ should be viewed as an approximation to the conditional variance. We have found that this parameter setting achieves nearly the best performance in most of our experiments.

\section{Graphical models and conditional independence}
\label{sec:graph-structure}

% We have proposed a concrete, general algorithm for knockoff sampling, and 
One outstanding issue is whether the Metropolized knockoff sampler can be run in reasonable time for cases of interest. We begin by showing why sequential knockoff sampling is prohibitively expensive without additional structure, and then turn our attention to a common type of structure that enables efficient sampling: graphical models. The central contribution of this section will be a complexity bound on Metro showing how the graphical structure affects the difficulty of sampling. To complete this line of investigation, we give a complexity lower bound for all knockoff samplers which shows that Metro is optimal in some cases.

\subsection{Why do we need structure?}
\label{subsec:why-structure}
Consider running Metro for some input distribution $\p$ and sample $X = x$. In view of
\eqref{eq:MH-SCEP-rejection-prob}, at step $j$ we need to evaluate $\p(\Xno{j}, X_j=z_j, \Xk_{1:(j-1)}, X^*_{1:(j-1)})$ for $z_j \in \{x_j, x_j^*\}$ up to a constant.\footnote{In this section, when not explicitly specified, a variable is set to its observed value, e.g., $\p(X_1 \mid X_2 = z_2, X_3, \Xk_1, X^*_1)$ is shorthand for $\p(X_1 = x_1  \mid X_2 = z_2, X_3 = x_3, \Xk_1 = \xk_1, X^*_1 = x^*_1)$.}  Metro defines a joint distribution on $(X, \Xk_{1:(j-1)}, X^*_{1:(j-1)})$ implicitly, so the only way to evaluate this density is to compute it step by step, from $1$ to $j-1$, i.e., through the sequential decomposition
\begin{multline}
\p(\Xno{j}, X_j=z_j, \tilde{X}_{1:(j-1)}, X^*_{1:(j-1)}) = 
\p(\Xno{j}, X_j = z_j) \times \\
\prod_{k=1}^{j-1}\left[\p(\Xk_k \mid \Xno{j}, X_j = z_j, \Xk_{1:(k-1)}, X^*_{1:k}) \p(X^*_k \mid \Xno{j}, X_j = z_j, \Xk_{1:(k-1)}, X^*_{1:(k-1)}) \right]. \label{eq:seq-decomp}
\end{multline}
Consider the term $\p(\Xk_k \mid \Xno{j}, X_j = z_j, \Xk_{1:(k-1)}, X^*_{1:k})$.
By the definition of Metro, computing this term will require evaluating an acceptance probability of the form
\begin{align}
\label{eq:accept-prob-graphical}
    \min\left(1, \frac{q_k(x_k \mid x_k^*) \p(\Xno{(j,k)}, X_k = x^*_k, X_j = z_j,  \Xk_{1:(k-1)}, X^*_{1:(k-1)})}
    {q_k(x_k^* \mid x_k)\p(\Xno{(j,k)}, X_k = x_k, X_j = z_j, \Xk_{1:(k-1)}, X^*_{1:(k-1)})}\right).
\end{align}
Now, to compute the terms in the acceptance probability, we must use the same sequential decomposition \eqref{eq:seq-decomp} for the terms $\p(\Xno{(j,k)}, X_k = z_k, X_j = z_j,  \Xk_{1:k-1}, X^*_{1:k-1})$ for $z_k \in \{x_k, x_k^*\}$. Considering $k = j-1$, we see that step $j$ is making two calls to the probability at step $j-1$, each of which is in turn making two calls to the probability function at step $j-2$ and so on. Thus, each evaluation of \eqref{eq:seq-decomp} will require $\Omega(2^j)$ function calls. This behavior is not due to a shortcoming of Metro; any genuine knockoff sampler with access only to an unnormalized density will require time exponential in $p$. We will present the formal statement of this lower bound later in Theorem \ref{theorem:timecomp}.

Although knockoff sampling with no restriction on the distribution is prohibitively slow, we will show how to avoid the exponential complexity when there is additional known structure. Consider a Markov chain, i.e., a density that factors as $\p(x) = \prod_{j=1}^{p-1} \phi_j(x_j, x_{j+1})$. In this case, the joint density \eqref{eq:seq-decomp} can be evaluated efficiently provided we proceed along the chain in the natural order. Assume for simplicity that the proposal distribution is fixed in advance so that the second term within the  square brackets in \eqref{eq:seq-decomp} does not depend on any variables and can be ignored. Due to the Markovian structure, only the $k=j-1$ term in the product depends on $z_j$, so it suffices to compute the acceptance probability \eqref{eq:accept-prob-graphical} for $k=j-1$. 
Again using the Markovian structure, this simplifies to
\begin{multline*}
\label{eq:MC-simple-acceptance}
   \min\left(1, {a_{j-1}} \frac{q_{j-1}(x_{j-1}\mid x^*_{j-1})\p(\Xno{(j,j-1)}, X_{j-1} = x_{j-1}^*, X_j = z_j)}
   {q_{j-1}(x_{j-1}^*\mid x_{j-1})\p(\Xno{(j,j-1)}, X_{j-1} = x_{j-1}, X_j = z_j)}\right) \\ = 
    \min\left(1, {a_{j-1}} \frac{q_{j-1}(x_{j-1}\mid x^*_{j-1})\phi_{j-2}(x_{j-2}, x_{j-1}^*) \phi_{j-1}(x_{j-1}^*,z_j)}{q_{j-1}(x_{j-1}^*\mid x_{j-1}) \phi_{j-2}(x_{j-2}, x_{j-1}) \phi_{j-1}(x_{j-1},z_j)}\right)
\end{multline*} 
where $a_{j-1}$ is the ratio 
\begin{align*}
a_{j-1} = \frac{\p(X^*_{1:(j-2)}, \Xk_{1:(j-2)} \mid \Xno{(j,j-1)}, X_{j-1} = x_{j-1}^*, X_j = z_j)}
{\p(X^*_{1:(j-2)}, \Xk_{1:(j-2)} \mid \Xno{(j,j-1)}, X_{j-1} = x_{j-1}, X_j = z_j)}
\end{align*}
which does not depend on $z_j$ by the Markov structure. The key here is that  $a_{j-1}$ was previously computed with $z_j = x_j$ when sampling $\Xk_{j-1}$. Thus, the acceptance probability can be computed in constant time. Putting this all together, for a Markov chain, each of the necessary joint probabilities \eqref{eq:seq-decomp} can be computed in constant time, and the time to sample the entire vector $\Xk$ is linear in the dimension $p$. Markov chains are not the only case where computing the acceptance probability can be done quickly; for other distributions with conditional independence structure, we next develop a systematic way of computing \eqref{eq:seq-decomp}, using the graphical structure to control the depth of the recursion and hence control the running time.

\subsection{Time complexity of Metro for graphical models}
\label{subsec:junction-tree}
We have seen that we must restrict our attention to a subset of distributions in order to efficiently sample knockoffs, so in this section we show how to implement Metropolized knockoff sampling for a very broad class of distributions: graphical models. Let $ X \in \rr^p$ be a random vector whose density factors over a
graph $G$: 
\begin{equation} \label{eq:density-factorization}
\p(x) \propto \Phi( x) = \prod_{c \in C} \phi_{c}( x_c);
\end{equation}
here, $C$ is the set of maximal cliques of the graph $G$ and $\Phi$ is
an unnormalized version of $\p$.  The variables in $X$ can be either
discrete or continuous. All graphical models with positive density or
mass take this form \citep{HammersleyClifford:1971}, and such
distributions are known to have particular conditional independence
properties. We refer the reader to \citet{koller2009probabilistic} for
a general treatment.

In order to take advantage of the conditional independence structure of $X$, we use a graph-theoretic object known as a {\em junction tree} \citep{Bertele:1972:NDP:578817} which encodes properties of the graph $G$.

\begin{defi}[Junction tree]
\label{def:junction-tree}
Let $T$ be a tree with vertices that
are subsets of the nodes $\{1,\dots,p\}$ of a graph $G$. T is a {\em
  junction tree} for $G$ if the following hold:
\begin{enumerate}
\item Each $j \in \{1,\dots,p\}$ appears in some vertex $V$ of $T$.
\item For every edge $(j,k)$ in $G$, $j \in V$ and $k \in V$ for some vertex $V$.
\item {\em (Running intersection property)} If the vertices $V$ and $V^\prime$ both contain a node of $G$, then every vertex in the unique path from $V$ to $V^\prime$ also contains this node.
\end{enumerate}
\end{defi}
Figure \ref{fig:grid-junction-tree} gives an example of a junction
tree over a $2\times3$ grid. The size of the largest vertex of $T$
minus one is known as the {\em width} of the junction tree $T$, and
the smallest width of a junction tree over $G$ is called the {\em
  treewidth} of $G$, a measure of graph complexity.  Finding the
junction tree of lowest width for a graph $G$ is known to be NP-hard
\citep{arnborg1987complexity}, but there exist efficient heuristic
algorithms for finding a junction tree with small width
\citep{Kjaerulff90triangulationof, koller2009probabilistic}.

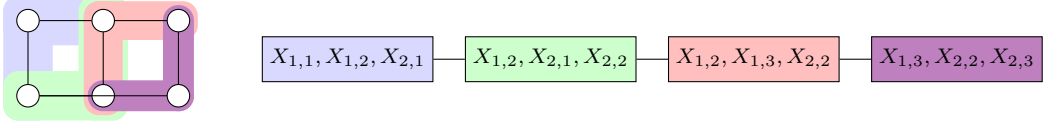
\begin{figure}
\centering
\tikzstyle{int}=[circle, draw, fill=white!20, minimum size=2em, scale = .5]
\def\gwidth{2cm} 
\def\opacity{1}

\begin{tikzpicture}[node distance=2cm,auto]
	\scriptsize
    \node [int] (11) {};
    \node [int] (12) [right of=11, node distance=\gwidth] {};
    \node [int] (13) [right of=12, node distance=\gwidth] {};
    
    \node [int] (21) [below of=11, node distance=\gwidth] {};
    \node [int] (22) [right of=21, node distance=\gwidth] {};
    \node [int] (23) [right of=22, node distance=\gwidth] {};
    
    \path[-] (11) edge node {} (12);
    \path[-] (11) edge node {} (21);
    \path[-] (12) edge node {} (13);
    \path[-] (12) edge node {} (22);
    \path[-] (21) edge node {} (22);
    \path[-] (21) edge node {} (23);
    \path[-] (13) edge node {} (23);
    
\begin{scope}[on background layer]
    \node[fill=blue!15, inner sep=5pt, rounded corners = .2cm, fit=(11) (12), opacity=\opacity](FIt1) {};
    \node[fill=blue!15, inner sep=5pt, rounded corners = .2cm, fit=(11) (21), opacity=\opacity](FIt1) {};
    
    \node[fill=green!20, inner sep=5pt, rounded corners = .2cm, fit=(21) (22), opacity=\opacity](FIt1) {};
    \node[fill=green!20, inner sep=5pt, rounded corners = .2cm, fit=(12) (22), opacity=\opacity](FIt1) {};
    
    \node[fill=red!25, inner sep=3pt, rounded corners = .2cm, fit=(12) (22), opacity=\opacity](FIt1) {};
    \node[fill=red!25, inner sep=3pt, rounded corners = .2cm, fit=(12) (13), opacity=\opacity](FIt1) {};
    
    \node[fill=violet!50, inner sep=1.5pt, rounded corners = .2cm, fit= (22)(23), opacity=\opacity](FIt1) {};
    \node[fill=violet!50, inner sep=1.5pt, rounded corners = .2cm, fit= (13)(23), opacity=\opacity](FIt1) {};
\end{scope}

\tikzstyle{int}=[draw, fill=blue!20, minimum size=2em]
\node [int, fill=blue!15] (1) [below right = .1cm and 3cm of 11] {$X_{1,1}, X_{1,2}, X_{2,1}$};
\node [int, fill=green!20] (2) [right of = 1, node distance=2.7cm] {$X_{1,2}, X_{2,1}, X_{2,2}$};
\node [int, fill=red!25] (3) [right of = 2, node distance=2.7cm] {$X_{1,2}, X_{1,3}, X_{2,2}$};
\node [int, fill=violet!50] (4) [right of = 3, node distance=2.7cm] {$X_{1,3}, X_{2,2}, X_{2,3}$};

\path[-] (1) edge node {} (2);
\path[-] (2) edge node {} (3);
\path[-] (3) edge node {} (4);

\end{tikzpicture}
\caption{A junction tree of treewidth $2$ for the $2 \times 3$ grid, which happens to be a chain.}
\label{fig:grid-junction-tree}
\end{figure}

Given a junction tree $T$ for the graph $G$, we will soon prove that
Metro can be run with $O(p 2^w)$ queries of the unnormalized density
$\Phi$, where $w$ is the width of $T$. In view of
\eqref{eq:MH-SCEP-rejection-prob}, at step $j$ of Metro we need to evaluate
$\p(\Xno{j}, X_j=z_j,\Xk_{1:(j-1)}, X^*_{1:(j-1)})$ for $z_j \in \{x_j, x^*_j\}$ up to a constant
as well as sample from and evaluate the proposal distribution
$q_j( \cdot \mid x_j)$. We can use the graphical model structure to
make these operations tractable by both (1) sampling the variables in
a specific order, and (2) choosing proposal distributions that are not
unnecessarily complex. We formalize these two requirements below.
 
We first consider the order in which we sample the variables. Recalling \eqref{eq:seq-decomp}, the complexity of the computations of $\p(\Xno{j}, X_j = z_j, \Xk_{1:(j-1)}, X^*_{1:(j-1)})$ depends on the number of function calls implied by the recursion \eqref{eq:seq-decomp}. For simplicity, assume that the proposal terms in the product, $\p(X^*_k \mid \Xno{j}, X_j = z_j, \Xk_{1:(k-1)}, X^*_{1:(k-1)})$, never depends on $z_j$; this will be relaxed soon. In that case, we need only consider the terms in \eqref{eq:seq-decomp} of the form $\p(\Xk_k \mid \Xno{j}, X_j = z_j, \Xk_{1:(k-1)}, X^*_{1:k})$ for $k < j$. When there is graphical structure, not all such terms will depend on $z_j$, and the number of terms that do depend on $z_j$ determines the recursion depth. In particular, if at step $j$ only $r_j$ terms depend on $z_j$, then there will be $O(2^{r_j})$ function calls in the recursion. A desirable ordering of the variables is then one that minimizes the largest $r_j$, and such an ordering can be extracted from a junction tree $T$ using Algorithm~\ref{alg:jt-order}.
 
\begin{algorithm}[H]%\label{alg:jt-order}
\caption{Junction tree variable ordering for Metro\label{alg:jt-order}}
	\SetAlgoLined\DontPrintSemicolon
	{Initialize tree $T_\text{active} = T$ and list $J= \{\}$. \;
	%{Initialize list $J= \{\}$} \;
		\While{$T_\textnormal{active} \ne \emptyset$} {
			Select a leaf node $V$ of $T_\text{active}$.  $V$ is connected to at most one other node $V^\prime$ of $T_\text{active}$ because it is a tree.\;
			In any order, append each $j \in V \setminus V^\prime$ to the end of the list $J$. If no $V^\prime$ exists, append all $j \in V$ to $J$ in any order. \;
			Remove $V$ from the active tree $T_\text{active}$. \;
		}
	}
	Return $J$
\end{algorithm}

Algorithm~\ref{alg:jt-order} is valid in that when a node is removed, no $j\in J$ remains in any node in $T_\text{active}$.\footnote{This simple fact follows from the running intersection property; we refer the reader to Lemma~\ref{lemma:var-ordering} in Appendix~\ref{app:lemmas}.}
From now on we assume that the variables are
numbered according to this ordering. Our second consideration is to create proposals that do not add
unnecessary complexity. No matter which proposal
distribution we choose, $\p( X_j = z_j \mid \Xno{j}, \Xk_{1:(j-1)}, X^*_{1:(j-1)})$ will still depend on some $X_\ell$ for $\ell > j$ due to dependencies among coordinates of $X$; we however restrict ourselves to proposal distributions that do not add any additional dependencies.
\begin{defi}[Compatible proposal distributions]
Let $V_j$ be the node of the junction tree when $j$ is appended to $J$ from Algorithm \ref{alg:jt-order}. Set $\bar V_{j} = \{1, \dots, j-1\} \cup V_j$. We say that proposal distributions $q_j$  are {\em compatible} with a junction tree $T$ if they depend only on $X_{\bar V_{j}}, \bXk_{1:(j-1)},$ and  $ X^*_{1:(j-1)}$.
\label{def:consistent}
\end{defi}
This definition is motivated by the property
\begin{equation*}
 X_{1:j} \ci X_{\bar V_{j}^c}  \mid   X_{\bar V_{j} \setminus \{1,\dots,j\}} ,
\end{equation*}
since $\bar V_{j} \setminus \{1,\dots,j\}$ separates $\{1,\dots,j\}$ from $\bar V_{j}^c$ in the graph $G$. Thus, a proposal distribution at step $j$ that violates the compatibility property and relies on $X_\ell$ for some $\ell \notin \bar{V}_j$ will result in additional non-one terms in the product in $\eqref{eq:seq-decomp}$ at step $\ell$, so $\bar{V}_j$ is the largest set that the proposal can be allowed to depend on without  increasing the number of function calls/runtime. Although not all proposals are compatible, it is a rich enough class to handle a broad range of knockoff distributions, including the distribution induced by SCIP.

With these two conditions in place, we now state our main result about the efficiency of knockoff sampling, giving an upper bound on the number of evaluations of the unnormalized density function $\Phi$ that is required by Metro when the graphical structure is known. Assuming the variable ordering from Algorithm \ref{alg:jt-order} and faithful proposal distributions compatible for $T$ such that sampling from and evaluating the proposal distributions does not require evaluating $\Phi$, we reach the following result:
\begin{thm}[Computational efficiency of Metro]\label{thm:scep-runtime}
  Let $ X$ be a random vector with a density which factors over a
  graph $G$ as in \eqref{eq:density-factorization}. Let $T$ be a
  junction tree of width $w$ for the graph $G$.  Under the conditions
  above, Metro uses $O(p 2^w)$ queries of 
  $\Phi$.
\end{thm}

This result means that we can efficiently implement Metropolized knockoff sampling for many interesting distributions, and it shows precisely how
the complexity of the conditional independence structure of $ X$
affects the complexity of the sampling algorithm. Furthermore, in the next section we will prove that this is the optimal complexity in some cases. 

\subsection{Time complexity of general knockoff sampling}
\label{subsec:time-comp}
In the previous section we analyzed the runtime of Metro and showed that it will be tractable for graphs of sufficiently low treewidth. Now, we investigate the computational complexity of knockoff sampling in general. To formalize our investigation, we discuss a model of computation in which we have no information about the distribution of $X$ beyond its graphical structure and the ability to query its (possibly unnormalized) density at any given point.

% Intuitively, since we must construct the conditional density of $\Xk\mid X$ to balance the joint density of $(X,\Xk)$ so that it is symmetric under any pairwise permutation, we need knowledge of the original density at $\{(z_1,\dots,z_p):z_j\in\{x_j,\xk_j\}\}$. To formally address this question, we discuss a model of computation in which we have no information about the distribution of $X$ beyond the ability to query its (possibly unnormalized) density at any given point.

\paragraph{Oracle model.} In this model, we are given as inputs (a) a
$p$-dimensional vector $X$ drawn from a density $\lambda \Phi$, where
$\lambda$ is a (possibly unknown) positive scalar so that we can think of $\Phi$ as an
unnormalized density, (b) the support of $\Phi$, and (c) a black box
capable of evaluating $\Phi$ at arbitrary query points, and (d) a graph $G$ for which the density is known to have the form \eqref{eq:density-factorization}. No other
information about $\Phi$ is available.

We show that in the oracle model with the complete graph, i.e., when there is no graphical structure, knockoff sampling requires exponential time in the number of covariates, $p$. Please note that any complexity bound must take into
account the quality of the generated knockoffs since $\bXk = X$ is
a trivial knockoff that can be sampled in no time.

\begin{thm}[Complexity lower bound for knockoff sampling]
  Consider a procedure operating in the oracle model which makes a
  finite number of calls to the black box $\Phi$ and returns $\bXk$, thereby
  inducing a joint distribution $( X,\bXk)$ obeying the pairwise
  exchangeability \eqref{eq:knockoff-swap} for all $\Phi$. If $G$ is the complete graph so that the procedure generates valid knockoffs for any input density, then the
  total number $N$ of queries of $\Phi$ must obey
  $N \ge 2^{\# \{j : X_j \ne \Xk_j\}} - 1$ a.s..
\label{theorem:timecomp}
\end{thm}

This result means that for any knockoff sampler, we cannot have both full generality and time efficiency. Put differently, in order to efficiently generate nontrivial knockoffs, we will need to restrict our attention to a subset of distributions for which we have structure. This fact justifies our decision to focus on distributions with graphical structure. We also derive a lower bound for the complexity of knockoff sampling for graphical models, stated next.

\begin{coro}[Complexity lower bound for graphical models]
\label{coro:chordal-lb}
Consider the setting of Theorem \ref{theorem:timecomp}. Fix a graph $G$ with maximal cliques $C$. Suppose that for all $\Phi$ of the form $\Phi(x) = \prod_{c \in C} \phi_c(x_c)$, the procedure induces a joint distribution $( X,\bXk)$ obeying pairwise exchangeability \eqref{eq:knockoff-swap}. Then $N \ge \max_{c \in C} 2^{\# \{j \in c: X_j \ne \Xk_j\}} - 1$ a.s..
\end{coro}

This proposition shows that even after making some useful structural
assumptions, there is still a trade-off between knockoff quality and
computation. We next derive a byproduct, which proves that Metro is achieving a good
runtime.

\begin{prop}[Optimality of Metro for chordal Gaussian graphical models]
  Consider continuous distributions of the form
  $\Phi( x) = \prod_{c \in C} \phi_c( x_c)$ over a chordal graph
  $G$.\footnote{A chordal graph is a graph such that any cycle of
    length 4 or larger has a chord.}  On the one hand, for any input, Metro can be
  run with $O(p^2 + p2^w)$ queries of $\Phi$. Furthermore, in the case
  where the distribution is Gaussian with zero mean and positive
  definite covariance (i.e., $\Phi(x)\propto\exp\left(-x\bm\Sigma^{-1} x^\top/2\right)$), Metro can produce
  knockoffs with $X_j \ne \Xk_j$ for all $j$ with probability 1. On the other hand, any general procedure that samples 
  knockoffs such that $X_j \ne \Xk_j$ for all $j$ with probability $\epsilon>0$ will require at
  least $2^w - 1$ queries of $\Phi$ with probability at least
  $\epsilon$.
\label{prop:gaussian-optimality}
\end{prop}

Proposition \ref{prop:gaussian-optimality} means that for chordal
graphs, any general knockoff sampling algorithm such that
$\p(X_j\ne\Xk_j\text{ for all }j)$ is bounded away from zero needs, in
expectation, the same exponential order of queries as Metro (with the proviso that $p$ is negligible compared to $2^w$).

\subsection{Divide-and-conquer to reduce treewidth} \label{subsec:high-width}

Theorem \ref{thm:scep-runtime} shows that Metro enables efficient
computations for random vectors whose densities factor over a graph
$G$ of low treewidth. Not all graphs corresponding to random vectors
of interest have low treewidth, however. A $d_1 \times d_2$ grid, for
example, has treewidth $\min(d_1, d_2)$
\citep{Diestel:2018:GT:3243389}. This section develops a mechanism for
simplifying the graphical structure of a random vector $ X$,
allowing for faster computation of exact knockoffs at the cost
of reduced knockoff quality.

To simplify graphical structure, we fix a set of variables $C$ that separates the graph $G$ into two subgraphs $A$ and $B$. After fixing the variables in $C$, knockoffs can be constructed for the variables in $A$ and $B$ independently.
\begin{prop}[Validity of divide-and-conquer knockoffs]
\label{prop:divide-and-conquer}
Suppose the sets $A, B, C$ form a partition of $\{1,\dots,p\}$ such
that $C$ separates $A$ and $B$ in the graph $G$, i.e., there is no path
from some $j \in A$ to some $k \in B$ in $G$ that does not contain
some $\ell \in C$. Suppose $\bXk$ is a random vector such that
$ X_{C} = \bXk_{C}$ a.s. and for all $j_A \in A$ and $j_B \in B$,
\begin{align*}
( X_D, \bXk_D) &\eqd ( X_D, \bXk_D)_{\swap(j_D)}  \mid   X_C  \qquad \text{ for } \quad D = A, B.
\end{align*}
Furthermore, assume we construct the knockoffs for $A$ and $B$ separately, i.e. $(X_A, \Xk_A) \independent  (X_B, \Xk_B) \mid X_C$.  Then $\bXk$ is a valid knockoff.
\end{prop}

The divide-and-conquer technique can be applied recursively to
split the graph into components of low treewidth until the
junction-tree algorithm for constructing knockoffs can be used on each
component. For example, for an arbitrary planar graph with $p$ nodes, the
planar separator theorem gives the existence of a subset of nodes $C$ of size
$O(\sqrt{p})$ that separates the graph into components $A$ and $B$
with $\max(\abs{A}, \abs{B}) \le 2p/3$ \citep{planar-separator},
suggesting that this technique will apply to many cases of
interest. Figure \ref{fig:grid-seperator} illustrates this technique
for a $d_1 \times d_1$ grid. We split the grid into rectangular
ribbons of size $d_1 \times d_2$ for small $d_2$; each resulting
ribbon has treewidth $d_2$.

The drawback of this approach is that for $j \in C$, we shall have
$X_j = \Xk_j$. When we think of deploying the knockoff framework in statistical
applications, one should remember that we will work with multiple
copies of $X$ corresponding to distinct observations.  We can then
choose different separator sets for each observation so that in the end, 
$X_j \neq \Xk_j$ for most of the observations. For example, in the
setting of Figure \ref{fig:grid-seperator}, one would randomly choose
between the two choices of $C$ for each observation. This technique is
explored numerically in Section \ref{subsubsec:ising-sims}.

\begin{figure}
\def\gwidth{1cm} 
\hspace{.1\textwidth}
\begin{subfigure}[t]{0.38\textwidth}
\centering
\tikzstyle{int}=[circle, draw, fill=white!20, minimum size=2em, scale = .5]
\tikzstyle{init} = [pin edge={to-,thin,black}]
\begin{tikzpicture}[node distance=4cm,auto]
	\tiny
    \node [int] (11) {};
    \node [int] (12) [right of=11, node distance=\gwidth] {};
    \node [int, fill = gray!50] (13) [right of=12, node distance=\gwidth] {};
    \node [int] (14) [right of=13, node distance=\gwidth] {};
    \node [int] (15) [right of=14, node distance=\gwidth] {};
    \node [int] (16) [right of=15, node distance=\gwidth] {};

    \node [int] (21) [below of=11, node distance = \gwidth]{};
    \node [int] (22) [right of=21, node distance=\gwidth] {};
    \node [int, fill = gray!50] (23) [right of=22, node distance=\gwidth] {};
    \node [int] (24) [right of=23, node distance=\gwidth] {};
    \node [int] (25) [right of=24, node distance=\gwidth] {};
    \node [int] (26) [right of=25, node distance=\gwidth] {};

    \node [int] (31) [below of=21, node distance = \gwidth] {};
    \node [int] (32) [right of=31, node distance= \gwidth] {};
    \node [int, fill = gray!50] (33) [right of=32, node distance= \gwidth] {};
    \node [int] (34) [right of=33, node distance=\gwidth] {};
    \node [int] (35) [right of=34, node distance=\gwidth] {};
    \node [int] (36) [right of=35, node distance=\gwidth] {};

    \node [int] (41) [below of=31, node distance = \gwidth] {};
    \node [int] (42) [right of=41, node distance=\gwidth] {};
    \node [int, fill = gray!50] (43) [right of=42, node distance=\gwidth] {};
    \node [int] (44) [right of=43, node distance=\gwidth] {};
    \node [int] (45) [right of=44, node distance=\gwidth] {};
    \node [int] (46) [right of=45, node distance=\gwidth] {};

    \node [int] (51) [below of=41, node distance = \gwidth] {};
    \node [int] (52) [right of=51, node distance=\gwidth] {};
    \node [int, fill = gray!50] (53) [right of=52, node distance=\gwidth] {};
    \node [int] (54) [right of=53, node distance=\gwidth] {};
    \node [int] (55) [right of=54, node distance=\gwidth] {};
    \node [int] (56) [right of=55, node distance=\gwidth] {};

    \node [int] (61) [below of=51, node distance = \gwidth] {};
    \node [int] (62) [right of=61, node distance=\gwidth] {};
    \node [int, fill = gray!50] (63) [right of=62, node distance=\gwidth] {};
    \node [int] (64) [right of=63, node distance=\gwidth] {};
    \node [int] (65) [right of=64, node distance=\gwidth] {};
    \node [int] (66) [right of=65, node distance=\gwidth] {};

    \node [] (A) [above of = 11, node distance = .5\gwidth] {\small A};
    \node [] (C) [above of = 13, node distance = .5\gwidth] {\small C};
    \node [] (B) [above of = 15, node distance = .5\gwidth] {\small B};

\begin{scope}[on background layer]
    \node[fill=blue!15, inner sep=1pt, rounded corners = .2cm, fit=(11) (A) (62)](FIt1) {};
    \node[fill=gray!10, inner sep=1pt, draw = gray, rounded corners = .2cm, fit=(C) (63)](FIt2) {};
    \node[fill=blue!15, inner sep=1pt, rounded corners = .2cm, fit= (14) (B) (66)](FIt3) {};
\end{scope}

    \path[-] (11) edge node {} (12);
    \path[-] (12) edge node {} (13);
    \path[-] (13) edge node {} (14);
    \path[-] (14) edge node {} (15);
    \path[-] (15) edge node {} (16);

    \path[-] (21) edge node {} (22);
    \path[-] (22) edge node {} (23);
    \path[-] (23) edge node {} (24);
    \path[-] (24) edge node {} (25);
    \path[-] (25) edge node {} (26);

    \path[-] (31) edge node {} (32);
    \path[-] (32) edge node {} (33);
    \path[-] (33) edge node {} (34);
    \path[-] (34) edge node {} (35);
    \path[-] (35) edge node {} (36);

    \path[-] (41) edge node {} (42);
    \path[-] (42) edge node {} (43);
    \path[-] (43) edge node {} (44);
    \path[-] (44) edge node {} (45);
    \path[-] (45) edge node {} (46);

    \path[-] (51) edge node {} (52);
    \path[-] (52) edge node {} (53);
    \path[-] (53) edge node {} (54);
    \path[-] (54) edge node {} (55);
    \path[-] (55) edge node {} (56);

    \path[-] (61) edge node {} (62);
    \path[-] (62) edge node {} (63);
    \path[-] (63) edge node {} (64);
    \path[-] (64) edge node {} (65);
    \path[-] (65) edge node {} (66);

    \path[-] (11) edge node {} (21);
    \path[-] (12) edge node {} (22);
    \path[-] (13) edge node {} (23);
    \path[-] (14) edge node {} (24);
    \path[-] (15) edge node {} (25);
    \path[-] (16) edge node {} (26);

    \path[-] (31) edge node {} (21);
    \path[-] (32) edge node {} (22);
    \path[-] (33) edge node {} (23);
    \path[-] (34) edge node {} (24);
    \path[-] (35) edge node {} (25);
    \path[-] (36) edge node {} (26);

    \path[-] (31) edge node {} (41);
    \path[-] (32) edge node {} (42);
    \path[-] (33) edge node {} (43);
    \path[-] (34) edge node {} (44);
    \path[-] (35) edge node {} (45);
    \path[-] (36) edge node {} (46);

    \path[-] (51) edge node {} (41);
    \path[-] (52) edge node {} (42);
    \path[-] (53) edge node {} (43);
    \path[-] (54) edge node {} (44);
    \path[-] (55) edge node {} (45);
    \path[-] (56) edge node {} (46);

    \path[-] (51) edge node {} (61);
    \path[-] (52) edge node {} (62);
    \path[-] (53) edge node {} (63);
    \path[-] (54) edge node {} (64);
    \path[-] (55) edge node {} (65);
    \path[-] (56) edge node {} (66);
\end{tikzpicture}
\end{subfigure}
~
\begin{subfigure}[t]{0.38\textwidth}
\centering
\tikzstyle{int}=[circle, draw, fill=white!20, minimum size=2em, scale = .5]
\tikzstyle{init} = [pin edge={to-,thin,black}]
\begin{tikzpicture}[node distance=4cm,auto]
	\tiny
    \node [int] (11) {};
    \node [int] (12) [right of=11, node distance=\gwidth] {};
    \node [int] (13) [right of=12, node distance=\gwidth] {};
    \node [int, fill = gray!50] (14) [right of=13, node distance=\gwidth] {};
    \node [int] (15) [right of=14, node distance=\gwidth] {};
    \node [int] (16) [right of=15, node distance=\gwidth] {};

    \node [int] (21) [below of=11, node distance = \gwidth]{};
    \node [int] (22) [right of=21, node distance=\gwidth] {};
    \node [int] (23) [right of=22, node distance=\gwidth] {};
    \node [int, fill = gray!50] (24) [right of=23, node distance=\gwidth] {};
    \node [int] (25) [right of=24, node distance=\gwidth] {};
    \node [int] (26) [right of=25, node distance=\gwidth] {};

    \node [int] (31) [below of=21, node distance = \gwidth] {};
    \node [int] (32) [right of=31, node distance= \gwidth] {};
    \node [int] (33) [right of=32, node distance= \gwidth] {};
    \node [int, fill = gray!50] (34) [right of=33, node distance=\gwidth] {};
    \node [int] (35) [right of=34, node distance=\gwidth] {};
    \node [int] (36) [right of=35, node distance=\gwidth] {};

    \node [int] (41) [below of=31, node distance = \gwidth] {};
    \node [int] (42) [right of=41, node distance=\gwidth] {};
    \node [int] (43) [right of=42, node distance=\gwidth] {};
    \node [int, fill = gray!50] (44) [right of=43, node distance=\gwidth] {};
    \node [int] (45) [right of=44, node distance=\gwidth] {};
    \node [int] (46) [right of=45, node distance=\gwidth] {};

    \node [int] (51) [below of=41, node distance = \gwidth] {};
    \node [int] (52) [right of=51, node distance=\gwidth] {};
    \node [int] (53) [right of=52, node distance=\gwidth] {};
    \node [int, fill = gray!50] (54) [right of=53, node distance=\gwidth] {};
    \node [int] (55) [right of=54, node distance=\gwidth] {};
    \node [int] (56) [right of=55, node distance=\gwidth] {};

    \node [int] (61) [below of=51, node distance = \gwidth] {};
    \node [int] (62) [right of=61, node distance=\gwidth] {};
    \node [int] (63) [right of=62, node distance=\gwidth] {};
    \node [int, fill = gray!50] (64) [right of=63, node distance=\gwidth] {};
    \node [int] (65) [right of=64, node distance=\gwidth] {};
    \node [int] (66) [right of=65, node distance=\gwidth] {};

    \node [] (A) [above of = 12, node distance = .5\gwidth] {\small A};
    \node [] (C) [above of = 14, node distance = .5\gwidth] {\small C};
    \node [] (B) [above of = 16, node distance = .5\gwidth] {\small B};

\begin{scope}[on background layer]
    \node[fill=blue!15, inner sep=1pt, rounded corners = .2cm, fit=(11) (A) (63)](FIt1) {};
    \node[fill=gray!10, inner sep=1pt, draw = gray, rounded corners = .2cm, fit=(C) (64)](FIt2) {};
    \node[fill=blue!15, inner sep=1pt, rounded corners = .2cm, fit= (15) (B) (66)](FIt3) {};
\end{scope}

    \path[-] (11) edge node {} (12);
    \path[-] (12) edge node {} (13);
    \path[-] (13) edge node {} (14);
    \path[-] (14) edge node {} (15);
    \path[-] (15) edge node {} (16);

    \path[-] (21) edge node {} (22);
    \path[-] (22) edge node {} (23);
    \path[-] (23) edge node {} (24);
    \path[-] (24) edge node {} (25);
    \path[-] (25) edge node {} (26);

    \path[-] (31) edge node {} (32);
    \path[-] (32) edge node {} (33);
    \path[-] (33) edge node {} (34);
    \path[-] (34) edge node {} (35);
    \path[-] (35) edge node {} (36);

    \path[-] (41) edge node {} (42);
    \path[-] (42) edge node {} (43);
    \path[-] (43) edge node {} (44);
    \path[-] (44) edge node {} (45);
    \path[-] (45) edge node {} (46);

    \path[-] (51) edge node {} (52);
    \path[-] (52) edge node {} (53);
    \path[-] (53) edge node {} (54);
    \path[-] (54) edge node {} (55);
    \path[-] (55) edge node {} (56);

    \path[-] (61) edge node {} (62);
    \path[-] (62) edge node {} (63);
    \path[-] (63) edge node {} (64);
    \path[-] (64) edge node {} (65);
    \path[-] (65) edge node {} (66);

    \path[-] (11) edge node {} (21);
    \path[-] (12) edge node {} (22);
    \path[-] (13) edge node {} (23);
    \path[-] (14) edge node {} (24);
    \path[-] (15) edge node {} (25);
    \path[-] (16) edge node {} (26);

    \path[-] (31) edge node {} (21);
    \path[-] (32) edge node {} (22);
    \path[-] (33) edge node {} (23);
    \path[-] (34) edge node {} (24);
    \path[-] (35) edge node {} (25);
    \path[-] (36) edge node {} (26);

    \path[-] (31) edge node {} (41);
    \path[-] (32) edge node {} (42);
    \path[-] (33) edge node {} (43);
    \path[-] (34) edge node {} (44);
    \path[-] (35) edge node {} (45);
    \path[-] (36) edge node {} (46);

    \path[-] (51) edge node {} (41);
    \path[-] (52) edge node {} (42);
    \path[-] (53) edge node {} (43);
    \path[-] (54) edge node {} (44);
    \path[-] (55) edge node {} (45);
    \path[-] (56) edge node {} (46);

    \path[-] (51) edge node {} (61);
    \path[-] (52) edge node {} (62);
    \path[-] (53) edge node {} (63);
    \path[-] (54) edge node {} (64);
    \path[-] (55) edge node {} (65);
    \path[-] (56) edge node {} (66);
\end{tikzpicture}
\end{subfigure}
\begin{subfigure}[t]{0.1666\textwidth}
\end{subfigure}

\caption{Two examples of conditioning to reduce the treewidth of a $
6 \times 6$ grid from $6$ to $3$.} \label{fig:grid-seperator}
\end{figure}
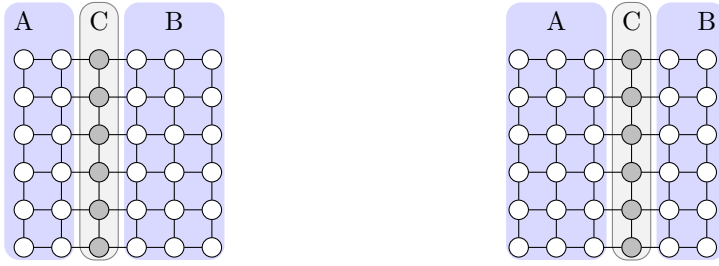

\subsection{Discrete distributions}
\label{subsec:rejection-free}
For discrete distributions with a small number of states for each coordinate $X_j$, the junction tree techniques from Section \ref{subsec:junction-tree} can be directly applied without using  Metropolized knockoff sampling. When each variable $X_j$ can take on at most $K$ values, the probability mass function $\p(X_j \mid \Xno{j}, \Xk_{1:(j-1)})$ can be represented as a vector in $\mathbb{R}^K$, so at step $j$ of the algorithm we simply need to evaluate $\p(X_j = z_j, \Xno{j}, \Xk_{1:(j-1)})$ for $z_j$ in the support of $X_j$. This is the same quantity we computed in Section \ref{subsec:junction-tree}; see, e.g., \eqref{eq:seq-decomp}. Once these probabilities have been computed, sampling from the resulting multinomial probability gives the SCIP procedure. In principle, this can be viewed as a special case of Metro, but for a practical implementation it is simpler to work directly with the probability vectors. A similar analysis to the proof of Theorem \ref{thm:scep-runtime} then shows that the procedure requires $O(pK^w)$ queries of the density $\Phi$; see Appendix \ref{subapp:complexity-proofs} for details. For discrete distributions with infinite or large $K$, this is not tractable. However, Metro still applies and is much faster.  

\subsection{Knockoffs for the Ising model} 
\label{subsec:ising-grid}

The tools from this section have the power to generate knockoffs for the Ising model on a grid \eqref{eq:ising-pmf}. To construct an efficient knockoff sampler for this distribution, we need to find a junction tree of minimal width for the $d_1 \times d_2$ grid so that we can apply the technique from Section \ref{subsec:rejection-free}. A junction tree for the $2 \times 3$ grid of width $2$ is shown in Figure \ref{fig:grid-junction-tree}, and the construction immediately generalizes to a junction tree of width $\min(d_1, d_2)$ for the $d_1 \times d_2$ grid, which is the optimal width. When $d_1 \ge d_2$, this leads to a knockoff sampler that proceeds from left to right, top to bottom; when variable $X_{i,j}$ is sampled, the other variables in the active node of the junction tree are $X_{i,j+1:p}$ and $ X_{i+1,1:j}$; see Figure \ref{fig:grid-active-frontier}. Per our upper bound, this knockoff sampler will have runtime $O(d_1 d_2 2^{\min(d_1,d_2)})$. If $\min(d_1, d_2)$ is large, this runtime may still be prohibitively long, but the divide-and-conquer technique from Section \ref{subsec:high-width} greatly increases speed at the cost of slightly worse knockoffs than the impractical full procedure. We conduct a simulation experiment of both the small-grid and large-grid setting in Section \ref{subsubsec:ising-sims}.

\begin{figure}
\begin{center}
\tikzstyle{int}=[circle, draw, fill=white!20, minimum size=2em, scale = .5]
\tikzstyle{init} = [pin edge={to-,thin,black}]
\def\gwidth{1.0cm}
\begin{tikzpicture}[node distance=4cm,auto]
	\tiny
    \node [int, fill = gray!50] (11) {};
    \node [int, fill = gray!50] (12) [right of=11, node distance=\gwidth] {};
    \node [int, fill = gray!50] (13) [right of=12, node distance=\gwidth] {};
    \node [int, fill = gray!50] (14) [right of=13, node distance=\gwidth] {};
    \node [int, fill = gray!50] (15) [right of=14, node distance=\gwidth] {};
    \node [int, fill = gray!50] (16) [right of=15, node distance=\gwidth] {};

    \node [int, fill = gray!50] (21) [below of=11, node distance = \gwidth]{};
    \node [int, fill = gray!50] (22) [right of=21, node distance=\gwidth] {};
    \node [int, fill = gray!50] (23) [right of=22, node distance=\gwidth] {};
    \node [int, fill = gray!50] (24) [right of=23, node distance=\gwidth] {};
    \node [int, fill = gray!50] (25) [right of=24, node distance=\gwidth] {};
    \node [int, fill = gray!50] (26) [right of=25, node distance=\gwidth] {};

    \node [int, fill = gray!50] (31) [below of=21, node distance = \gwidth] {};
    \node [int, fill = gray!50] (32) [right of=31, node distance= \gwidth] {};
    \node [int, fill = green!30] (33) [right of=32, node distance= \gwidth] {};
    \node [int, fill = blue!20, style=dashed] (34) [right of=33, node distance=\gwidth] {};
    \node [int, fill = blue!20, style=dashed] (35) [right of=34, node distance=\gwidth] {};
    \node [int, fill = blue!20, style=dashed] (36) [right of=35, node distance=\gwidth] {};

    \node [int, fill = blue!20, style=dashed] (41) [below of=31, node distance = \gwidth] {};
    \node [int, fill = blue!20, style=dashed] (42) [right of=41, node distance=\gwidth] {};
    \node [int, fill = blue!20, style=dashed] (43) [right of=42, node distance=\gwidth] {};
    \node [int] (44) [right of=43, node distance=\gwidth] {};
    \node [int] (45) [right of=44, node distance=\gwidth] {};
    \node [int] (46) [right of=45, node distance=\gwidth] {};

    \node [int] (51) [below of=41, node distance = \gwidth] {};
    \node [int] (52) [right of=51, node distance=\gwidth] {};
    \node [int] (53) [right of=52, node distance=\gwidth] {};
    \node [int] (54) [right of=53, node distance=\gwidth] {};
    \node [int] (55) [right of=54, node distance=\gwidth] {};
    \node [int] (56) [right of=55, node distance=\gwidth] {};

    \node [int] (61) [below of=51, node distance = \gwidth] {};
    \node [int] (62) [right of=61, node distance=\gwidth] {};
    \node [int] (63) [right of=62, node distance=\gwidth] {};
    \node [int] (64) [right of=63, node distance=\gwidth] {};
    \node [int] (65) [right of=64, node distance=\gwidth] {};
    \node [int] (66) [right of=65, node distance=\gwidth] {};

    \path[-] (11) edge node {} (12);
    \path[-] (12) edge node {} (13);
    \path[-] (13) edge node {} (14);
    \path[-] (14) edge node {} (15);
    \path[-] (15) edge node {} (16);

    \path[-] (21) edge node {} (22);
    \path[-] (22) edge node {} (23);
    \path[-] (23) edge node {} (24);
    \path[-] (24) edge node {} (25);
    \path[-] (25) edge node {} (26);

    \path[-] (31) edge node {} (32);
    \path[-] (32) edge node {} (33);
    \path[-] (33) edge node {} (34);
    \path[-] (34) edge node {} (35);
    \path[-] (35) edge node {} (36);

    \path[-] (41) edge node {} (42);
    \path[-] (42) edge node {} (43);
    \path[-] (43) edge node {} (44);
    \path[-] (44) edge node {} (45);
    \path[-] (45) edge node {} (46);

    \path[-] (51) edge node {} (52);
    \path[-] (52) edge node {} (53);
    \path[-] (53) edge node {} (54);
    \path[-] (54) edge node {} (55);
    \path[-] (55) edge node {} (56);

    \path[-] (61) edge node {} (62);
    \path[-] (62) edge node {} (63);
    \path[-] (63) edge node {} (64);
    \path[-] (64) edge node {} (65);
    \path[-] (65) edge node {} (66);

    \path[-] (11) edge node {} (21);
    \path[-] (12) edge node {} (22);
    \path[-] (13) edge node {} (23);
    \path[-] (14) edge node {} (24);
    \path[-] (15) edge node {} (25);
    \path[-] (16) edge node {} (26);

    \path[-] (31) edge node {} (21);
    \path[-] (32) edge node {} (22);
    \path[-] (33) edge node {} (23);
    \path[-] (34) edge node {} (24);
    \path[-] (35) edge node {} (25);
    \path[-] (36) edge node {} (26);

    \path[-] (31) edge node {} (41);
    \path[-] (32) edge node {} (42);
    \path[-] (33) edge node {} (43);
    \path[-] (34) edge node {} (44);
    \path[-] (35) edge node {} (45);
    \path[-] (36) edge node {} (46);

    \path[-] (51) edge node {} (41);
    \path[-] (52) edge node {} (42);
    \path[-] (53) edge node {} (43);
    \path[-] (54) edge node {} (44);
    \path[-] (55) edge node {} (45);
    \path[-] (56) edge node {} (46);

    \path[-] (51) edge node {} (61);
    \path[-] (52) edge node {} (62);
    \path[-] (53) edge node {} (63);
    \path[-] (54) edge node {} (64);
    \path[-] (55) edge node {} (65);
    \path[-] (56) edge node {} (66);
\end{tikzpicture}
\end{center}
\caption{An illustration of sampling knockoffs for an Ising model on a grid from Section \ref{subsec:ising-grid}. The blue  dashed nodes represent the active variables of the junction tree when variable $X_{3,3}$ (shown in green) is being sampled. Gray nodes indicate variables that have already been sampled, and white nodes indicate variables that have not been sampled yet and are not in the active node of the junction tree.}
\label{fig:grid-active-frontier}
\end{figure}
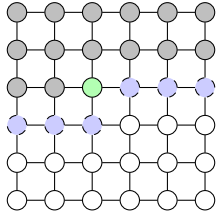

%\ejc{We should probably report timing results.}

%\input{sections/numerical}

\section{Numerical experiments} \label{sec:simulations}
 We now empirically examine the Metropolized knockoff sampler, beginning with the few models where previously known samplers are available as a baseline, and then continuing on to cases with no previously known samplers. Condensed plots are presented in the main text, while more comprehensive versions can be found in Appendix \ref{app:sim-details}. We provide approximate runtimes with a single-core\footnote{The hardware varies across simulations, but each CPU is between 2.5Ghz and 3.3Ghz.} implementation in either R or Python. All source code is available from {\url{https://github.com/wenshuow/metro}} with interactive computing notebooks at {\url{http://web.stanford.edu/group/candes/metro}} demonstrating the usage of the code and presenting further experimental results. 

\subsubsection*{Measuring knockoff quality}
The {\em mean absolute correlation} (MAC) is a useful measure of knockoff quality for a joint distribution of $( X, \bXk)$:
\begin{equation}\label{eq:mac-metric}
\text{MAC}(\law( X, \bXk)) := \frac1p\sum_{j=1}^p  \abs{\text{cor}(X_j, \Xk_j)}.
\end{equation}
We will use this as our measure of knockoff quality in our simulation experiments. Lower values of $\text{MAC}$ are preferred. Let $\bm\Gamma$ be the correlation matrix of $( X,\bXk)$; pairwise exchangeability implies $\bm\Gamma$ is of the form \eqref{eq:knockoff-cov}. The $\MAC$ is then $\frac1p\sum_{j=1}^p|1-s_j|$. Since $\bm\Gamma = \bm\Gamma( s)$ has to be positive semidefinite, a lower bound on the $\MAC$ achievable by any knockoff-generation algorithm for a given distribution is the optimal value of the program 
\begin{equation}
\label{eq:sdp}
\min_{ s}\sum_{j=1}^p|1-s_j|, \text{ subject to }\bm\Gamma( s)\succeq 0.
\end{equation}
%\citet{candes2018panning} building on \citet{barber2015controlling} show how to optimize this metric by solving an semi-definite program (SDP) in the case where $X$ is multivariate Gaussian. We develop a novel lower bound for this criterion as the solution to the same SDP and include it in the plots of our experimental results.
This minimization problem can be solved efficiently with semidefinite programming \citep{barber2015controlling}; we call the solution the \emph{SDP lower bound} for the MAC. 
%This lower bound corresponds to the best MAC for a variable $\bXk$ that correctly matches the first two moments. 
This lower bound can be achieved for Gaussian distributions \citep{candes2018panning}. Valid knockoffs, however, must match {\em all} moments, not just the second moments, so this lower bound is not expected to be achievable in general; still it provides a useful goalpost in our simulations.

\subsection{Models with previously known knockoff samplers}

\subsubsection{Gaussian Markov chains}
\label{subsubsec:gmc}
We first apply our algorithm to Gaussian Markov chains and compare with the SDP Gaussian knockoffs, whose MAC achieves the SDP lower bound exactly, and SCIP knockoffs, both from \citet{candes2018panning}. 
We take $p=500$ features such that $X_1\sim\mathcal N(0,1)$ and $X_{j+1} \mid  X_{1:j}\sim\mathcal N(\rho_jX_j,1-\rho_j^2)$. First, since the model is multivariate Gaussian, the covariance-guided proposal with $s$ computed by the SDP method \eqref{eq:sdp} will be identical to the SDP Gaussian knockoffs, so already a clever implementation of Metro is as good as a method specifically designed for Gaussian distributions, and since both achieve the SDP lower bound, one cannot do better in terms of MAC. % Although we could also match the SDP lower bound with the rejection-free proposals of Section \ref{subsec:rejection-free}, they do not apply to general continuous distributions and thus are not a fair representation of the general form of MH-SCEP. So in our simulations,
Thus, we only investigate the MTM-proposals for implementing Metro. Note that the Gaussian knockoffs from \citet{candes2018panning} do not use the Markovian structure of this problem, but instead rely on operations on $2p \times 2p$ matrices, whereas the MTM knockoffs from this work utilize the Markovian structure to achieve time complexity linear in $p$.

The results are presented in Figure \ref{figure:GMC}. Following
Section \ref{subsec:MTM}, we vary the number of proposals and the step
size. We find that choosing the step size for $X_j$ to be proportional
to $\sqrt{1/(\bm\Sigma^{-1})_{jj}}$ gives consistent results across
different sets of $\rho_j$'s. The MTM consistently outperforms the SCIP procedure,
and is reasonably close to the SDP procedure. It is observed that the defaults from Section~\ref{subsec:MTM} of eight
proposals ($m=4$) and $t_j=1.5\sqrt{1/(\bm\Sigma^{-1})_{jj}}$ performs
nearly the best in all settings. Confirming our reasoning in Section
\ref{subsec:MTM}, we find that the performance stabilizes as $m$ grows
and the step size should not be too large or too small, although for
sufficiently large $m$ the MAC is fairly stable to the choice of $t$. In this setting, it takes around $1$ second for MTM to sample one knockoff vector with $m=4$ and $t_j=1.5\sqrt{1/(\bm\Sigma^{-1})_{jj}}$.

\begin{figure}
    \centering
\includegraphics[width = \textwidth]{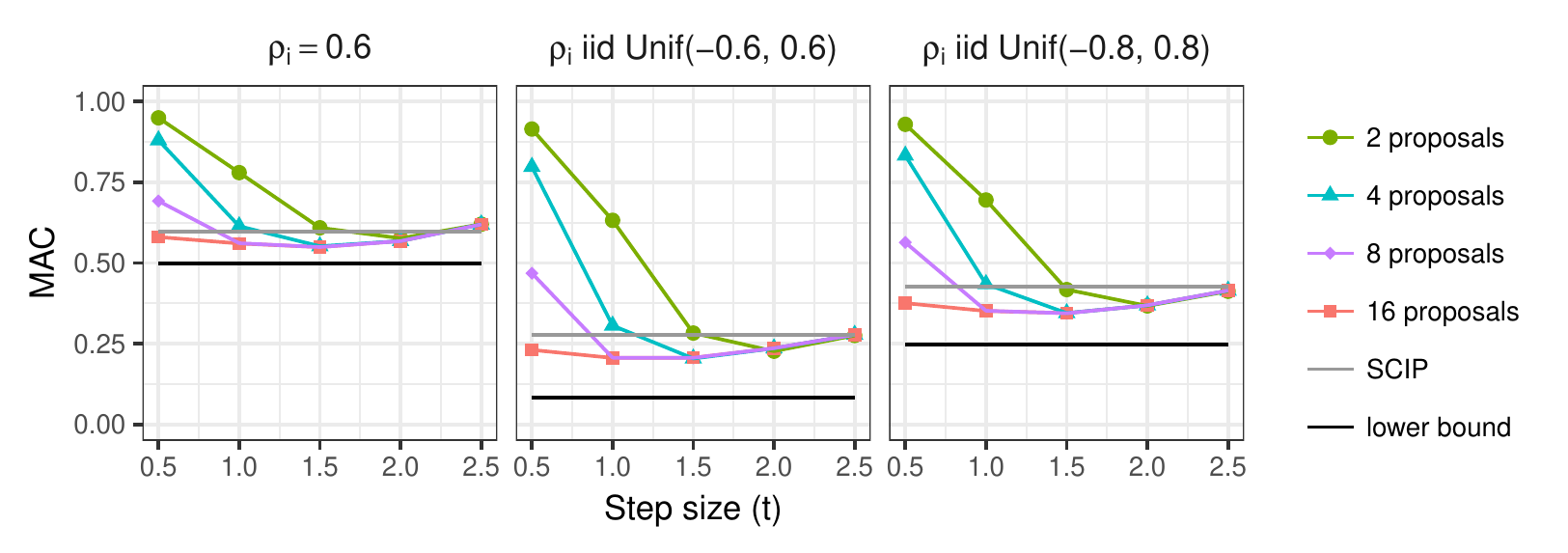}
    \caption{Simulation results for Gaussian Markov chains. The unit of step sizes is $\sqrt{1/(\bm\Sigma^{-1})_{jj}}$. All standard errors are below $0.001$. In this case, the lower bound is achieved by the SDP Gaussian knockoffs, or equivalently, the covariance-guided proposal with an $s$ given by the SDP \eqref{eq:sdp}.} 
    \label{figure:GMC}
\end{figure}

\subsubsection{Discrete Markov chains}
For discrete Markov Chains there is one previously-known knockoff sampler, which is an implementation of the SCIP procedure \citep{sesia2018gene}. We consider here Metro with MTM proposals. (The covariance-guided proposals would require ad-hoc rounding so we do not consider this here.) We take a simple Markov Chain with $K \in \{5, 10\}$ states with uniform initial distribution and transition probabilities $Q(j,j^\prime)$ defined as
\begin{equation}
Q(j,j^\prime)=\frac{(1 - \alpha)^{\abs{j - j^\prime}}}{\sum_{\ell=1}^K(1 - \alpha)^{ \abs{j - \ell}  }}.
\end{equation}
We examine $\alpha$ from 0 (independent coordinates) to $0.5$ (strong dependence between adjacent coordinates), with $p=500$ features.

%The bottom panels of figure \ref{fig:dmc-scep} shows how the optimal tuning for procedure varies with the level of dependence. For independent coordinates ($\alpha = 0)$, the optimal rho is $\rho = 0$, which corresponds to SCIP. As the dependence between adjacent coordinates increases, lower values of $\rho$ lead to optimal knockoffs. For $K = 5, \alpha = .2$, for example, we find that $\rho \approx -.15$ is the best point, leading to knockoffs with $\text{MAC}$ nearly 0.

%The results are summarized in Figures \ref{figure:DMC5} and \ref{figure:DMC10} in the Appendix.

We examine the MTM methods across a range of values of the tuning parameters, and the results are presented in Figure \ref{fig:dmc-compare}. Full simulation results are given in Appendix \ref{app:sim-details}. Note that the cases with $K=5$ and $\alpha \le 0.15$ are tuned with the additional parameter $\gamma$ from \eqref{eq:acceptanceprob}, as detailed in Appendix \ref{subapp:effect-gamma}. We find that the best-tuned MTM method outperforms the SCIP method and achieves MAC near the lower bound for all dependence levels $\alpha$. It takes around $0.5$ seconds and $0.7$ seconds respectively, to run MTM ($m=4$ and $t=1$) for $K=5$ and $K=10$.%We believe tuning $\gamma$ is helpful only incases where the distribution is discrete with small support, where we recommend the rejection-free method for simplicity. 
%{\color{blue} It is interesting to note the relatively bad performance of MTM when $K=5$ and $\alpha$ is small; we explain why this is expected when the distribution is discrete and the support is small and show how to fix it in . [SB: we should put a fixed version in the plot to provide evidence for this claim.]}

% The relatively bad performance of MTM when $K=5$ and $\beta$ is small is explained in Appendix \ref{subapp:effect-gamma}. As suggested there, this issue can be fixed within the MTM framework by tuning $\beta$. However, we expect this to happen only when the distribution is discrete and has a small support, in which case the rejection-free method is recommended, as it performs just as well and involves fewer parameters to tune.

\begin{figure} %\label{fig:dmc_mtm_rf_comparison}
\begin{center}
\includegraphics[width = \textwidth]{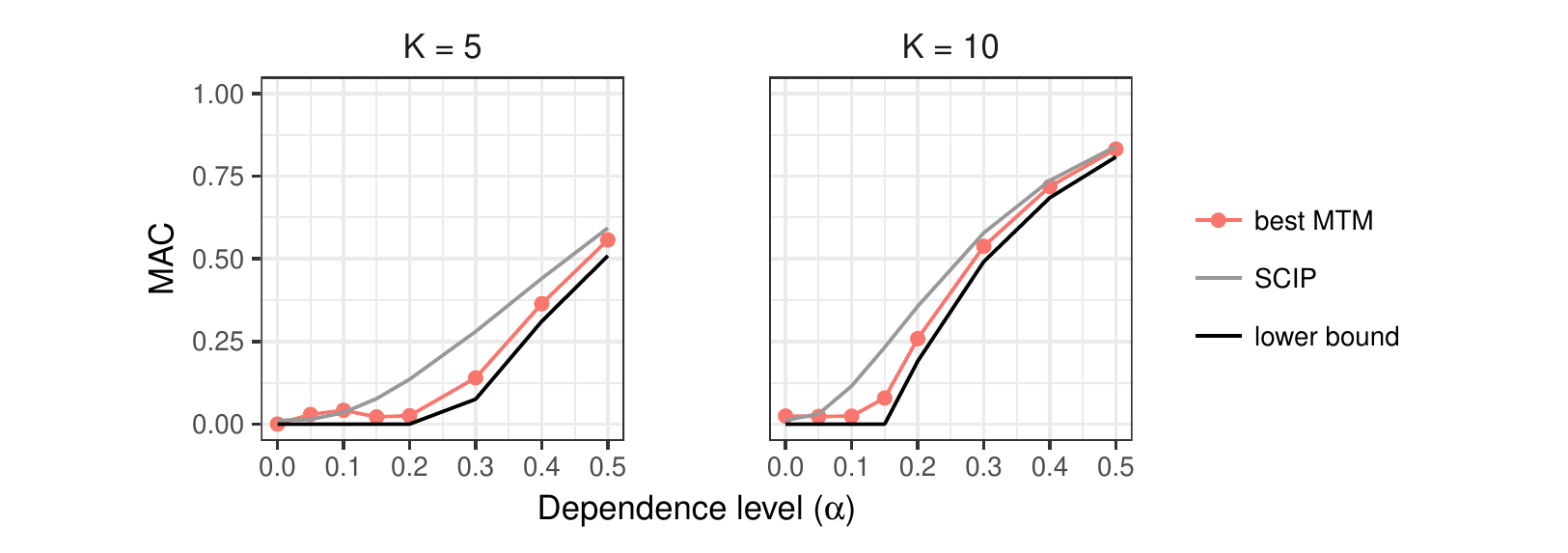}
\end{center}
\caption{A comparison of the MTM procedure for discrete Markov chains with SCIP and the SDP lower bound. All standard errors are below $0.002$.}  \label{fig:dmc-compare}
\end{figure}

\subsection{Models with no previously-known knockoff sampler}

\subsubsection{Heavy-tailed Markov chains}
\label{subsubsec:htmc}
As an example of a heavy-tailed distribution, we consider a Markov chain with $t$-distributed tails.
\begin{equation}
X_1=\sqrt{\frac{\nu-2}\nu}Z_1, \quad X_{j+1}=\rho_j X_j + \sqrt{1-\rho_j^2}\sqrt{\frac{\nu-2}\nu}Z_{j+1}, \quad Z_j\simiid t_\nu,
\label{eq:t-dist}
\end{equation}
for $j=1,\dots,p=500$ where $t_\nu$ represents the Student's $t$-distribution with $\nu>2$ degrees of freedom (note this is not a multivariate $t$-distribution). We try both the covariance-guided proposal with $ s$ provided by the SDP method \eqref{eq:sdp} and the MTM proposals. We set $\nu=5$ and use the same $\rho_j$'s as in the Gaussian setting.  As in Section \ref{subsubsec:gmc}, a step size of $1.5 \sqrt{1/(\bm\Sigma^{-1})_{jj}}$ again performs well. The covariance-guided proposals also perform well, although unlike the Gaussian case, there is now a gap between the lower bound and the performance of the covariance-guided proposals. In this setting, it takes around $1.6$ seconds for MTM to sample one knockoff vector with $m=4$ (eight proposals) and $t_j=1.5\sqrt{1/(\bm\Sigma^{-1})_{jj}}$. For the covariance-guided proposals, it takes around $12.5$ seconds for the one-time computation of the parameters (excluding time used for computing $s$, which varies depending on the method) and then  $0.3$ seconds to sample each knockoff vector.

\begin{figure}[h]
\centering
\includegraphics[width = \textwidth]{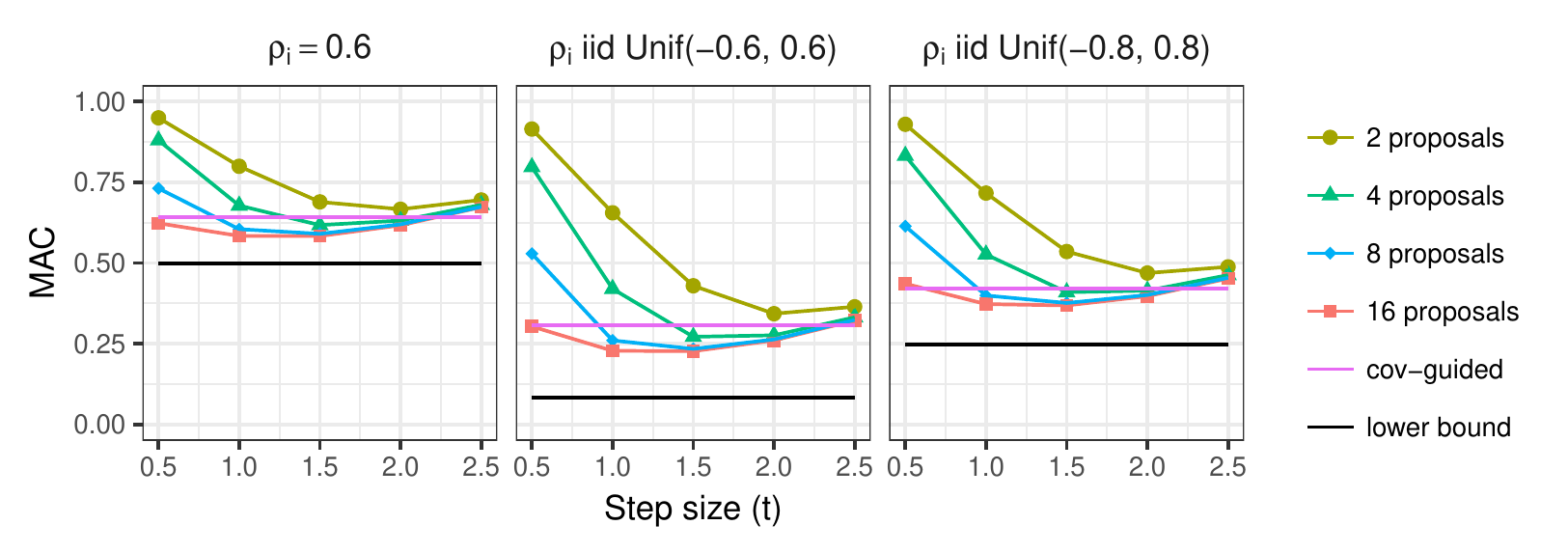}
    \caption{Simulation results for the $t$-distributed Markov chains. The unit of step sizes is $\sqrt{1/(\bm\Sigma^{-1})_{jj}}$. All standard errors are below $0.001$.}
    \label{figure:NGMC}
\end{figure}

\subsubsection{Asymmetric Markov chains}
\label{subsubsec:amc}
As an example of asymmetric, continuous distributions, we take a standardized equal mixture of Gaussian and exponential random variables and then form a Markov chain. Explicitly,
\begin{equation*}
Z_j\simiid \frac{I\cdot \mid  Y_\text{G} \mid  -(1-I)\cdot Y_\text{E}- \mu}{\sigma} \ \text{ for } j=1,\dots,p=500,
\end{equation*}
where $Y_\text{G}\sim\mathcal N(0,1)$, $Y_\text{E}\sim\Expo(1)$ and $I\sim\Bern(1/2)$ are independent. The parameters $\mu$ and $\sigma$ are chosen so that $Z_j$ has mean $0$ and variance $1$. We then take
\begin{equation*}
X_1=Z_1, \quad X_{j+1}=\rho_j X_j + \sqrt{1-\rho_j^2}Z_{j+1} \ \text{ for } j=2,\dots,p.
\end{equation*}
We examine both the covariance-guided proposal with $s$ provided by the SDP \eqref{eq:sdp} and the multiple-try proposals. We use the same $\rho_j$'s as in the Gaussian setting. As in the previous case, $m=4$ (eight proposals) and $t_j=1.5\sqrt{1/(\bm\Sigma^{-1})_{jj}}$ performs essentially as well as any other MTM parameter choices, and significantly outperforms the covariance-guided proposals. The timing results are the same as in the heavy-tailed Markov chains.

\begin{figure}
\centering
\includegraphics[width = \textwidth]{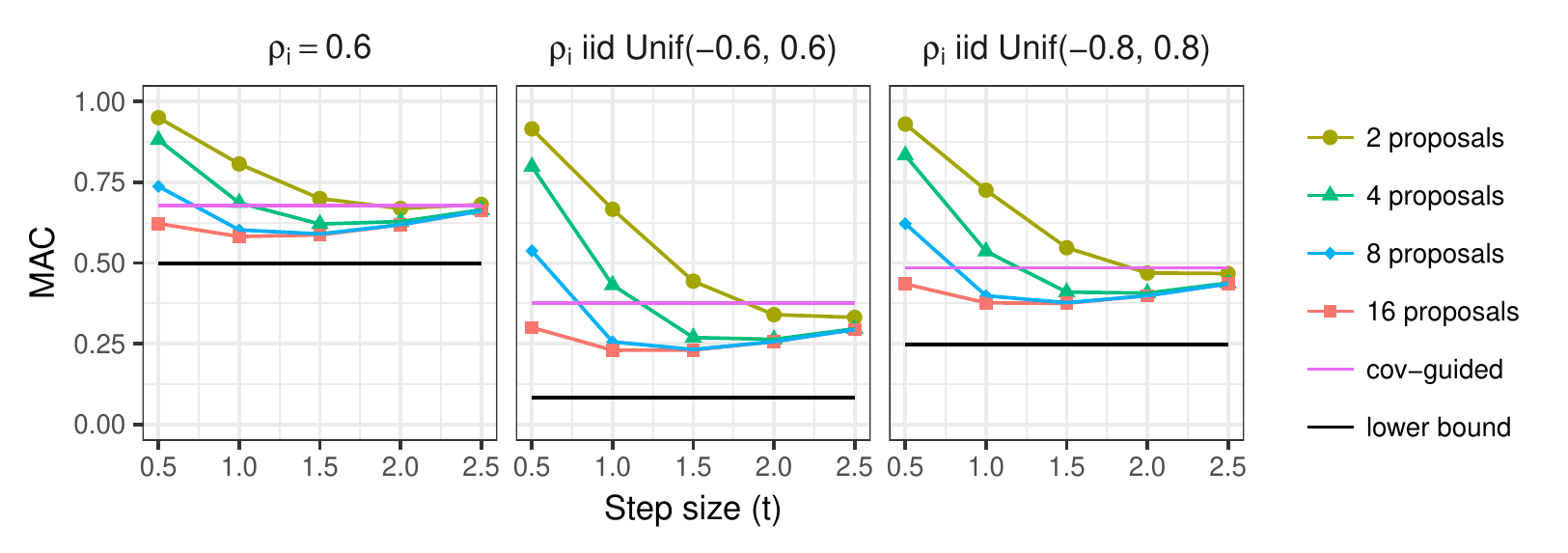}
    \caption{Simulation results for the asymmetric Markov chains. The unit of step sizes is $\sqrt{1/(\bm\Sigma^{-1})_{jj}}$. All standard errors are below $0.001$.}
    \label{figure:NEMC}
\end{figure}

\subsubsection{Ising model}
\label{subsubsec:ising-sims}
In this section, we consider an Ising model over a square grid \eqref{eq:ising-pmf}. We generate knockoffs with the method for discrete random variables from Section \ref{subsec:rejection-free} combined with the divide-and-conquer technique, the combination of which was described for Ising models in Section \ref{subsec:ising-grid}; no other exact knockoff samplers are known for the Ising model. Although our sampling procedures for the Ising model do not explicitly use the Metropolis--Hastings step, as explained in Section \ref{subsec:rejection-free}, we will refer to the sampler as ``Metro'' in this section for simplicity.

First, we take a $10 \times 10$ grid and set all $\beta_{i,j,i^\prime,j^\prime} = \beta_0$ and all $\alpha_{i,j} = 0$. The results are presented in Figure \ref{fig:ising-scep-results}. The left panel shows 
how the MAC increases---or, the quality decreases---
as the dependence between adjacent variables---$\beta_0$---increases.
We see that the procedure is close to the lower bound for large $\beta_0$. In the middle panel, we plot $\text{cor}(X_{j,k}, \Xk_{j,k})$ across different coordinates $(j,k)$. We see that on the edges of the grid, especially on the corners, knockoffs have lower correlation with their original counterparts. These variables are less determined by the values of the rest of the grid, so this is expected. In this setting, it takes about 12 seconds to sample a knockoff.
\begin{figure}[h]
\centering
\includegraphics[width = \textwidth]{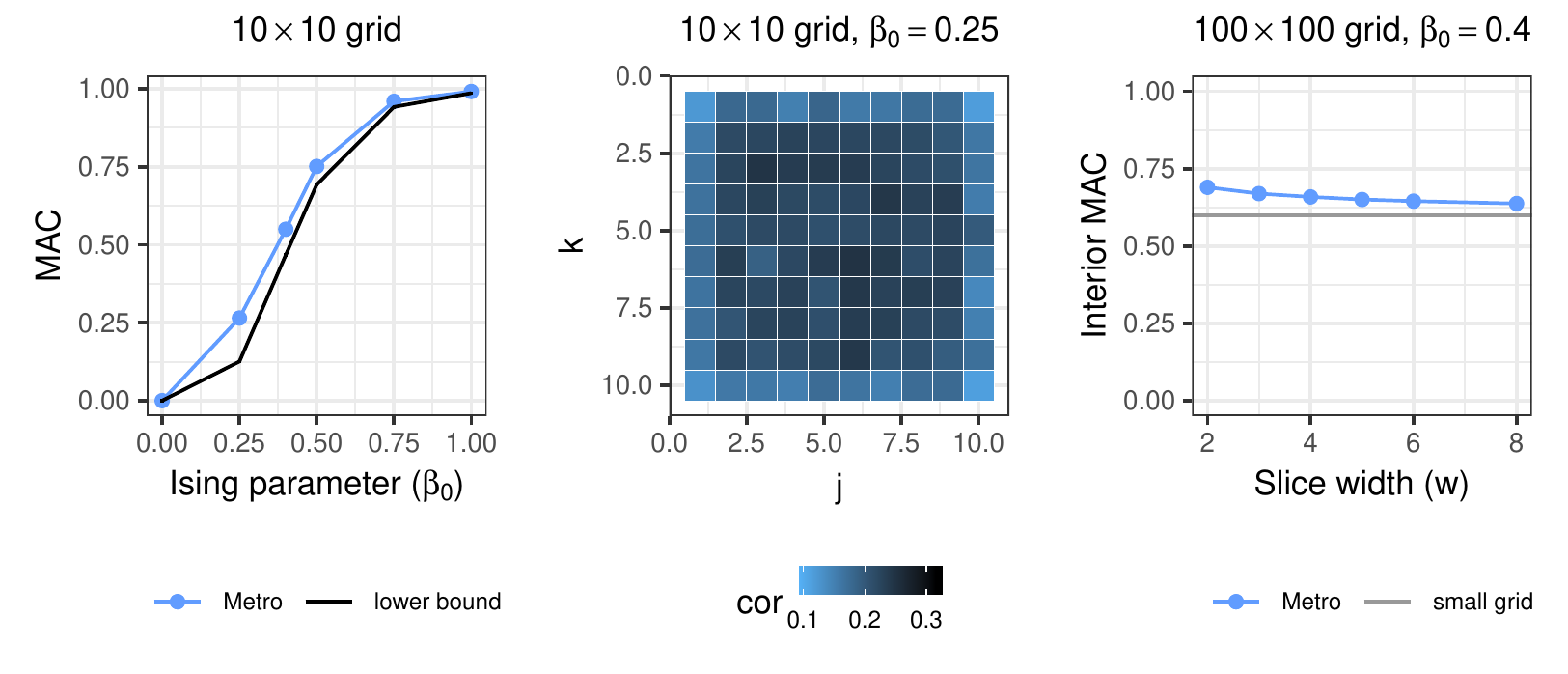}
\caption{Results of the Ising model experiments. All standard errors in the line plots are less than $0.005$.} \label{fig:ising-scep-results}
\end{figure}
%\subsubsection*{Divide-and-conquer for large grids} \label{subsubsec:slicing-experiment}

Next, we demonstrate the divide-and-conquer technique from Section \ref{subsec:high-width}. Here we consider the Ising model from above on a $100 \times 100$ grid, for a total dimension of $10,000$. The $100 \times 100$ grid has treewidth $100$, so Metro would not be tractable without the the divide-and-conquer technique. We divide the graph into subgraphs of width $w$, by fixing entire columns as in Figure \ref{fig:grid-seperator}. To measure the effect of the slicing, we compute the MAC on the interior points and compare this to the MAC of the interior points of a smaller grid for a procedure without slicing, see Appendix \ref{subapp:ising-sims} for details. We find that the quality of the knockoffs increases as we take larger slices, as expected. Furthermore, even modest values of $w$ such as $w=5$ result in a procedure that achieves a MAC close to that of the baseline. Recall that the complexity of Metro scales as $2^w$, so fixing $w=5$ dramatically reduces the computation time compared to $w=100$. With $w=5$, it takes about 2.5 minutes to generate one knockoff for the $100 \times 100$ grid.

\subsubsection{Gibbs measure on a grid}
Lastly, we demonstrate the MTM proposals simultaneously with the junction tree techniques for complex dependence structure. Consider a Gibbs measure on $\{1,\dots,K\}^{d \times d}$, with a probability mass function
\begin{equation*}
\p(X) = \frac{1}{Z(\beta_0)} \exp\left(-\beta_0\sum_{\substack{s, t \in \mathcal I \\\|s-t\|_1=1}}(x_s-x_t)^2 \right),\quad\mathcal I = \{(i_1,i_2) : 1\le i_1,i_2 \le d\},
\end{equation*}
and note that like the Ising model, this density factors over the grid. For our experiment, we take a $10\times10$ grid and examine different dependence levels $\beta_0$ with $K=20$ possible states for each variable. We apply Metro with the MTM proposals and the divide-and-conquer technique on the grid, tuning the procedure across a range of parameters as detailed in Appendix \ref{app:sim-details}. The condensed results are given in Figure \ref{fig:gibbs}. We do not know of another knockoff sampler in this setting. Having said this, we observe that our procedure has MAC close to the lower bound. We also observe that in the case where $w=3$, with as few as two proposals, our procedure performs well and takes about half a second to generate a knockoff copy; when we increase the number of proposals to ten, the compute time is around $2$ minutes. When $w$ is set to $5$, the slowest setting is $m=t=1$, which takes less than $4$ minutes.

\begin{figure}[h]
\begin{center}
\includegraphics[height = 2.3in]{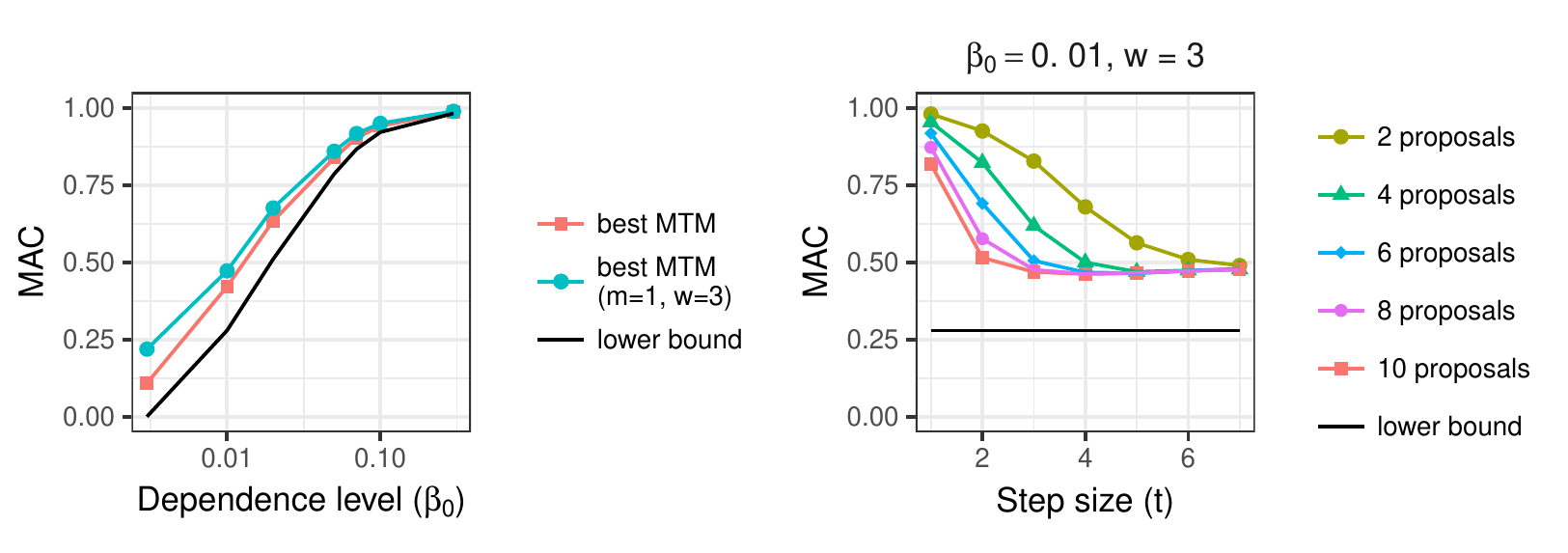}
\end{center}
\caption{Results of the Gibbs measure experiments. All standard errors are below $0.002$. In the left panel, $\beta_0$ is shown in logarithmic scale.}  \label{fig:gibbs}
\end{figure}

% \begin{figure}
% \begin{center}
% \includegraphics[width = 2.3in]{slice_plot.pdf}
% \end{center}
% \caption{The results of the divide-and-conquer experiment from section \ref{subsubsec:slicing-experiment}. We set $\beta_0 = .4$, and each point represents $n=1000$ data points with all standard errors below $.001$. The dashed line shows the small grid interior MAC baseline.} \label{fig:slicing-results}
% \end{figure}

%\input{sections/discussion}
\section{Discussion}
\label{sec:discussion}
This paper introduced a sequential characterization of all valid knockoff-generating procedures and used it along with ideas from MCMC and graphical models to create Metropolized knockoff sampling, an algorithm which generates valid knockoffs in complete generality with access only to $X$'s unnormalized density. Although we proved in Theorem~\ref{theorem:timecomp} that no algorithm (including Metro) can sample exact knockoffs \emph{efficiently} for arbitrary $X$ distributions, we characterized one way out of this impossibility result: conditional independence structure in $X$. An interesting future direction would be to establish other sufficient conditions on a model family that would allow one to sample knockoffs efficiently. Another way out of the lower bound in  Theorem~\ref{theorem:timecomp} is to forgo exact knockoffs and settle for approximations. Although this arguably is a tall order, it would be interesting to establish theoretical guarantees on the approximation quality of these or other approximate knockoff constructions, and better understand the tradeoff between knockoff approximation quality and time complexity.

\section*{Acknowledgements}
E.~C.~was partially supported by the Office of Naval Research under grant N00014-16- 1-2712, by the National Science Foundation via DMS 1712800, and by a generous gift from TwoSigma. S.~B.~was supported by a Ric Weiland Graduate Fellowship.  S.~B.~and E.~C.~would like to thank Yaniv Romano and Matteo Sesia for useful comments on an early version of this work. L.~J.~and W.~W.~would like to thank Jun Liu for fruitful discussions on MCMC.

\bibliography{refs}{}

\appendix

\section{Junction tree lemmas}
\label{app:lemmas}

This section includes several important lemmas which will be used in some proofs in the appendix.
\begin{lemma}
\label{lemma:var-ordering}
If the variables are ordered by Algorithm~\ref{alg:jt-order}, then for each $1\le j\le p$, any node in the junction tree that contains $j$ is an element of the set $\{V_1,V_2,\dots,V_j\}$. In addition, if $j\in V_k$ for some $k>j$, then $V_k=V_j$.
\end{lemma}

\begin{proof}[Proof of Lemma \ref{lemma:var-ordering}]
According to Algorithm~\ref{alg:jt-order}, when a node $V$ is selected, all variables in $V\setminus V'$---here, $V'$ is the unique neighbor of $V$ in the remaining junction tree---are sampled before the next node is selected. Recall that $V_j$ is the selected node when $j$ is sampled; let $k\ge j$ be the last sampled variable when $V_j$ is selected (in this case, we have $V_j=V_{j+1}=\cdots=V_k$ by definition). After $k$ is sampled, $V_j$ is removed, and by Algorithm~\ref{alg:jt-order}, this means $j,j+1,\dots,k$ do not appear in $V_j$'s only remaining neighbor. Now we claim no remaining node contains $j$. Otherwise, if there is a node $V^{*}_j$ which contains $j$ and still remains, by the running intersection property, all nodes on the unique path between $V_j$ and $V^*_j$ contain $j$. This  would imply that $V_j$'s remaining neighbor in the junction tree also contains $j$ since the unique path must pass through $V_j$'s only remaining neighbor. This is a contradiction. Now we know that any node that contains $j$ is some $V_\ell$ with $1\le\ell\le k$. Since $V_j=V_{j+1}=\cdots=V_k$, the lemma follows.
\end{proof}

\begin{lemma}
\label{lemma:nested-V}
If the variables are ordered by Algorithm~\ref{alg:jt-order}, then for any $j>\ell$ such that $j\in\bar V_\ell$, we have $\bar V_\ell\subseteq\bar V_j$.
\end{lemma}

\begin{proof}[Proof of Lemma~\ref{lemma:nested-V}]
Consider any $k\in\bar V_\ell$. Now we show such a $k$ must be in $\bar V_j$.
The case $k\le j$ is trivial, since $k\in\bar V_j$ by definition.
If $k>j$, then $j,k\in V_\ell$. Assume $V_\ell\ne V_j$; otherwise there is nothing to prove. Before $j$---and, therefore, $k$---are sampled, each time a node $V$ containing $j$ and $k$ (e.g., $V_\ell$) is selected, $j$ and $k$ appear in $V$'s neighbor. By a recursive argument, before $V_j$ is selected, each time the node containing $j$ and $k$ is selected, $j$ does not get sampled and neither does $k$ ($k$ is sampled after $j$). Hence, before $V_j$ is selected, there is always at least one remaining node that contains both $j$ and $k$. By Algorithm~\ref{alg:jt-order} and Lemma~\ref{lemma:var-ordering}, $V_j$ is the last selected node $j$ appears in, so it has to contain both $j$ and $k$ (otherwise no node contains both $j$ and $k$ at this point). This means that $k\in V_j\subseteq\bar V_j$.
\end{proof}

\begin{lemma}
\label{lemma:G-V-relationship}
If the variables are ordered by Algorithm~\ref{alg:jt-order}, then for any $j\ne k$, if $j$ and $k$ are connected in $G$, we have $j\in\bar V_k$ and $k\in\bar V_j$.
\end{lemma}

\begin{proof}[Proof of Lemma~\ref{lemma:G-V-relationship}]
Without loss of generality, we assume $k>j$, so $j\in\bar V_k$ by definition. Now we show $k\in\bar V_j$ also holds.
By the second property of the junction tree, $k$ co-appears with $j$ at least once in some node $V_\ell$. But $j$ does not appear in any node after $V_j$ by Lemma~\ref{lemma:var-ordering}, so there is some $\ell\le j$ such that $\{j,k\}\subseteq V_\ell$. If $j=\ell$, then we already have $k\in V_\ell=V_j\subseteq\bar V_j$; otherwise, $j>\ell$, so by Lemma~\ref{lemma:nested-V}, $k\in V_\ell\subseteq\bar V_\ell\subseteq\bar V_j$.
%By the same argument as in the proof of Lemma~\ref{lemma:nested-V}, this implies $k\in V_j\subseteq\bar V_j$ ($k$ and $j$ always co-appear in at least one remaining node before $V_j$ is selected).
\end{proof}

\section{Covariance-guided proposals}
\label{app:cov-guide}

This section includes details on the covariance-guided proposal introduced in Section~\ref{subsesc:cov-guided-prop}, and proofs of its faithfulness and compatibility. See Appendix \ref{app:mat-inv} for details on how to do the necessary linear algebra computations for the covariance-guided proposals efficiently.

We first recall the definition of the covariance-guided proposals.  Let $Z$ be a $2p$-dimensional vector drawn from $\mathcal{N}((\mu,\mu),{\bm \Gamma(s)})$, where $\bm\Gamma(s)$ is as in \eqref{eq:knockoff-cov}. 
Let $q_j$ be the probability density function of $Z_{p+j}$ conditional on $Z_{1:(p+j-1)}$.
% In the Gaussian setting, if we implement the algorithm introduced in Section \ref{subsec:mhscep} with the proposal $q_j$,
% \begin{equation*}
% \begin{aligned}
% &\ \density_{j+1}(x_{j+1},\xno{(j+1)},\tilde{ x}_{1:j}, x^*_{1:j})\\
% =&\ \density_j(x_j,\xno{j},\tilde{ x}_{1:j-1}, x^*_{1:j-1})\Bigg[\delta(\tilde x_j-x_j)q_j(x^*_j \mid  x_j, \xno{j}, \tilde{ x}_{1:(j-1)}, x^*_{1:j-1})\\
% &\times\left(1-\min\left(1,\frac{q_j(x_j \mid  x^*_j, \xno{j}, \tilde{ x}_{1:(j-1)},  x^*_{1:j-1})\density_j(x^*_j,\xno{j},\tilde{ x}_{1:j-1}, x^*_{1:j-1})}{q_j(x^*_j \mid  x_j, \xno{j}, \tilde{ x}_{1:(j-1)},  x^*_{1:j-1})\density_j(x_j,\xno{j},\tilde{ x}_{1:j-1}, x^*_{1:j-1})}\right)\right)\\
% &+q_j(\tilde x_j \mid  x_j, \xno{j}, \tilde{ x}_{1:(j-1)}, x^*_{1:j-1})\\
% &\times\min\left(1,\frac{q_j(x_j \mid  \tilde x_j, \xno{j}, \tilde{ x}_{1:(j-1)}, x^*_{1:j-1})p_j(\tilde x_j,\xno{j},\tilde{ x}_{1:j-1}, x^*_{1:j-1})}{q_j(\tilde x_j \mid  x_j, \xno{j}, \tilde{ x}_{1:(j-1)}, x^*_{1:j-1})p_j(x_j,\xno{j},\tilde{ x}_{1:j-1}, x^*_{1:j-1})}\right)\delta(x^*_j-\text{NA})\Bigg],
% \end{aligned}
% \end{equation*}
The proposal distribution at the $j$th step $q_j(x_j^* \mid x_{1:p},x^*_{1:(j-1)})$, is defined to be the conditional density of $Z_{p+j}$ given $Z_{1:(p+j-1)}=(x_{1:p},x^*_{1:(j-1)})$.\footnote{It seems equally plausible to use $ x_{1:p}$ and $\tilde{ x}_{1:(j-1)}$ (i.e., $(x_1,x_2,\dots,x_p,\tilde x_1,\tilde x_2,\dots,\tilde x_{j-1})$ instead of equation \eqref{eq:cov-guide-condition}). However, we find that its empirical performance is not as good as the version presented in the main text.} 

\begin{prop}
\label{prop:cov-guieded-faithful}
The covariance-guided proposals are faithful (see Section~\ref{subsec:mhscep})
% \sout{. \ejc{I still see a source of confusion. Section 3.1 states: "In our setting, we shall make sure that the choice of the proposal
% distribution depends on the previously sampled pairs in a symmetric
% fashion, thereby remaining faithful to the knockoff symmetry condition
% \eqref{eq:knockoff-symmetry-condition} in Theorem \ref{theorem:seq}. As such, we call such proposals \emph{faithful}." Who are the previously sampled pairs? We are symmetric in what? In $(x, \tilde{x})$, in $(x, x^*)$, in $(\tilde{x}, x^*)$...?}} 
in that the proposal distribution at the $j$th step depends on $(X_k,\Xk_k)$ in a symmetric way for $1\le k<j$.
\end{prop}

\begin{proof}[Proof of Proposition \ref{prop:cov-guieded-faithful}]
Consider the proposal at step $j$. To see how the proposal distribution depends on $(X_k,\Xk_k)$ for $k<j$, note that we are using the distribution of $Z_{j+p}$ conditional on
\begin{equation}
Z_{1:(p+j-1)} = (x_1,x_2,\dots,x_p,\ind_{x_1=\tilde x_1}x^*_1+\ind_{x_1\ne\tilde x_1}\tilde x_1,\dots,\ind_{x_{j-1}=\tilde x_{j-1}}x^*_{j-1}+\ind_{x_{j-1}\ne\tilde x_{j-1}}\tilde x_{j-1}).
\label{eq:cov-guide-condition}
\end{equation}
We only need to check if the proposal density changes when swapping $x_k$ and $\tilde x_k$ for $k\le j-1$. Note that if we rejected at the $k$th step, $x_k=\tilde x_k$, so there is no effect of swapping the two; if we accepted at the $k$th step, the dependence is symmetric in $(x_k,\tilde x_k)$ because of the structure of the covariance matrix.
\end{proof}

\begin{prop}
\label{prop:cov-guided-compatible}
If $(\Sigma^{-1})_{ij} = 0$ whenever $i \ne j$ and $(i,j)$ is not an edge in the graph $G$, then the covariance-guided proposals are compatible with $G$ (see Definition \ref{def:consistent}).
\end{prop}

\begin{proof}[Proof of Proposition \ref{prop:cov-guided-compatible}]
We wish to show, %The joint proposal distribution is
%\begin{equation*}
% X^*\mid X\sim\mathcal N(\mu^\top+(\bm\Sigma-\diag( s))\bm\Sigma^{-1}( X-\mu)^\top,\bm\Sigma-(\bm\Sigma-\diag( s))\bm\Sigma^{-1}(\bm\Sigma-\diag( s))).
%\end{equation*}
$\mathcal L(X_j^*\mid X, X_{1:j-1}^*)$ only depends on $X_k$ if $k\in\bar V_{j}$. %Dependence on $X_k$ through $X_k^*=\ind_{X_{k}=\tilde X_{k}}X^*_{k}+\ind_{X_{k}\ne\tilde X_{k}}\tilde X_{k}$ could only happen if $k\in\{1,2,\dots,j-1\}\subseteq\bar V_j$. 
To do this, we use induction over $j$. For $j=1$, since $X_1^*\mid X$ is a Gaussian distribution whose conditional variance does not depend on $X$, it suffices to show that $\e[X_1^*\mid X]$ depends on $X_k$ only if $k\in\bar V_1$. Note that
\begin{equation*}
(X_1,X_1^*)\mid \Xno{1}\sim\mathcal N\left((\mu^\text{cond},\mu^\text{cond}),\bm\Sigma^\text{cond}\right),\quad\mu^\text{cond}=\e[X_1\mid \Xno{1}],\ \  \bm\Sigma^\text{cond}=\Cov\left((X_1,X^*_1)\mid\Xno{1}\right),
\end{equation*}
Since $\bm\Sigma^\text{cond}$ does not depend on $X$, $\e[X_1^*\mid X_1,\Xno{1}]$ is a linear function of $X_1$ and $\mu^\text{cond}$. It is easy to see that $\mu^\text{cond}=\e[X_1\mid \Xno{1}]$ depends on $X_k$ only if $k$ and $1$ co-appear in some node of the junction tree. By Lemma~\ref{lemma:var-ordering}, this node can only be $V_1$ and, therefore, $k\in V_1$. Thus the base case $j=1$ holds.

Suppose the claim is true up to $j-1$. By the same argument on the conditional distribution $(X_j,X_j^*)\mid \Xno{j},X^*_{1:(j-1)}$, it suffices to show that $\e[X_j\mid\Xno{j},X^*_{1:(j-1)}]$ only depends on $X_k$ if $k\in\bar V_j$. We have 
\begin{equation*}
\begin{aligned}
\p(x_j\mid\xno{j},x^*_{1:(j-1)})&=\frac{\p(x_j,\xno{j},x^*_{1:(j-1)})}{\int_{\rr}\p(x'_j,\xno{j},x^*_{1:(j-1)})\di{x'_j}}\\
&=\frac{\p(x_j,\xno{j})\prod_{\ell=1}^{j-1}\p(x^*_{\ell}\mid x_j,\xno{j},x^*_{1:(\ell-1)})}{\int_{\rr}\p(x'_j,\xno{j})\prod_{\ell=1}^{j-1}\p(x^*_{\ell}\mid x'_j,\xno{j},x^*_{1:(\ell-1)})\di{x'_j}}\\
&=\frac{\p(\xno{j})\p(x_j\mid\xno{j})\prod_{\ell=1}^{j-1}\p(x^*_{\ell}\mid x_j,\xno{j},x^*_{1:(\ell-1)})}{\int_{\rr}\p(\xno{j})\p(x'_j\mid\xno{j})\prod_{\ell=1}^{j-1}\p(x^*_{\ell}\mid x'_j,\xno{j},x^*_{1:(\ell-1)})\di{x'_j}}\\
&=\frac{\p(x_j\mid\xno{j})\prod_{\ell=1}^{j-1}\p(x^*_{\ell}\mid x_j,\xno{j},x^*_{1:(\ell-1)})}{\int_{\rr}\p(x'_j\mid\xno{j})\prod_{\ell=1}^{j-1}\p(x^*_{\ell}\mid x'_j,\xno{j},x^*_{1:(\ell-1)})\di{x'_j}}.
\end{aligned}
\end{equation*}
By the induction hypothesis, for $\ell < j$, $\p(x^*_{\ell}\mid x_j,\xno{j},x^*_{1:(\ell-1)})$ does not depend on $x_j$ unless $j\in\bar V_\ell$ (which implies $j \in V_\ell$ since $j > \ell$), so the $\ell$th term in the product can be removed from the numerator and the denominator if $j\notin V_\ell$. Thus, we now have
\begin{equation*}
\begin{aligned}
\p(x_j\mid\xno{j},x^*_{1:(j-1)})&=\frac{\p(x_j\mid\xno{j})\prod_{\ell: \ell<j,\ j\in V_\ell}\p(x^*_{\ell}\mid x_j,\xno{j},x^*_{1:(\ell-1)})}{\int_{\rr}\p(x'_j\mid\xno{j})\prod_{\ell: \ell<j,\ j\in V_\ell}\p(x^*_{\ell}\mid x'_j,\xno{j},x^*_{1:(\ell-1)})\di{x'_j}}.
\end{aligned}
\end{equation*}
Now we will prove that $\p(x_j\mid\xno{j})\prod_{\ell: \ell<j,\ j\in V_\ell}\p(x^*_{\ell}\mid x_j,\xno{j},x^*_{1:(\ell-1)})$ depends on $x_k$ only if $k\in\bar V_j$, which will conclude case $j$. Consider first the terms $\p(x^*_{\ell}\mid x_j,\xno{j},x^*_{1:(\ell-1)})$ in the product: such a term depends on $x_k$ only if $k \in \bar{V}_{\ell}$ by the induction hypothesis, and by Lemma~\ref{lemma:nested-V}, $\bar V_\ell\subseteq\bar V_j$.
% To see that such a $k$ must be in $\bar V_j$, consider the following two cases: 
% \begin{itemize}
%     \item If $k\le j$, then $k\in\bar V_j$ by definition.
%     \item If $k>j$, then $j,k\in V_\ell$. Before $j$---and, therefore, $k$---is sampled, each time a node containing $j$ and $k$ (e.g., $V_\ell$) is selected, \rev{$j$ and $k$ appear in its neighbor}. \ejc{Whose neighbor? The neighbor of $V_\ell$?} By a recursive argument, before $V_j$ is selected, each time the node containing $j$ and $k$ is selected, $j$ does not get sampled and neither does $k$ ($k$ is sampled after $j$). Hence, before $V_j$ is selected, there is always at least one remaining node that contains both $j$ and $k$. By Algorithm~\ref{alg:jt-order} and Lemma~\ref{lemma:var-ordering}, $V_j$ is the last selected node $j$ appears in, so it has to contain both $j$ and $k$ (otherwise no node contains both $j$ and $k$ at this point). This means that $k\in V_j\subseteq\bar V_j$.
% \end{itemize}
Next, the term $\p(x_j\mid\xno{j})$ only depends on $x_k$ if $k=j$ or $k$ is connected to $j$ in $G$. If $k=j$, $k\in\bar V_j$ by definition; otherwise, $k\in\bar V_j$ follows from Lemma~\ref{lemma:G-V-relationship}.
% \begin{itemize}
%     \item If $k\le j$, then $k\in\bar V_j$ by definition;
%     \item Consider $k > j$.
%     By the property of the junction tree, $k$ co-appears with $j$ at least once in some node $V_\ell$. But $j$ does not appear in any node after $V_j$ by Lemma~\ref{lemma:var-ordering}, so there is some $\ell\le j$ such that $\{j,k\}\subseteq V_\ell$. By the same argument as above, this implies $k\in V_j\subseteq\bar V_j$ ($k$ and $j$ always co-appear in at least one remaining node before $V_j$ is selected).
% \end{itemize}
\end{proof}

We have established that as long as $\bm\Sigma$ reflects the structure of $G$, i.e., for $i\ne j$, $\left(\bm\Sigma^{-1}\right)_{ij}\ne0$ only if $i$ and $j$ are connected in $G$, the covariance-guided proposals are compatible. For these proposals, sampling and evaluating the proposal density can be done without without querying the density $\Phi$, so Theorem \ref{thm:scep-runtime} implies that Metro with covariance-guided proposals requires $O(p2^w)$ queries of $\Phi$.

It is easy to see that if $X$ is Gaussian and $\gamma=1$, we always accept because the acceptance ratio is always $1$. %\ejc{The remainder does not belong here.} We can change the proposal distribution to non-Gaussian while retaining the conditional mean and variance, so the proposal would still be compatible while we might address issues like discreteness or heavy tails.

\section{Proofs}
% \subsection{Theorem \ref{theorem:timecomp}}
% \label{subapp:time-comp}
% % \addtocounter{thm}{-6}

% % Note that the procedure $\mathcal K$ is an algorithm for generating knockoffs. This theorem implies that in the most desired case when $x_i\ne\tilde x_i$ for all $i$, the time complexity is at least $O(2^p)$.
% \begin{proof}
% See Appendix \ref{subapp:proof-time-comp}. 
% \renewcommand{\qedsymbol}{} 
% A continuous version is include in Appendix \ref{subapp:cont-time-comp}.
% \end{proof}{}

\subsection{Necessity of the knockoff symmetry condition}
At first glance, it might not be directly clear why we need the symmetry condition \eqref{eq:knockoff-symmetry-condition} in Theorem \ref{theorem:seq}. To illustrate why this is necessary, consider the following example: let $p$ be the density function
\begin{equation*}
p(x,\bxk)= 1+\sin\left(2\pi\Bigl(x_p+\tilde x_p+\sum_{j=1}^{p-1}(x_j-\tilde x_j)\Bigr)\right), \qquad  (x,\bxk)\in[0,1]^{2p}.
\end{equation*}
Each pair $(X_j, \Xk_j)$ is unexchangeable unless $j = p$. However, marginalizing out any coordinate would yield the uniform distribution. In other words, in any sequential construction, we would have that $X_j$ and $\Xk_j$ are conditionally independent and, therefore, exchangeable up until the last step. 
Since $X_p$ and $\tilde X_p$ are exchangeable conditionally on everything else, conditional exchangeability would hold. This example shows that if we require conditional exchangeability only, we would not necessarily end up with valid knockoffs. 
To press this point further, imagine that in the SCIP algorithm, we change the last step: instead of conditional independence we simply require conditional exchangeability. Then we are not guaranteed to get valid knockoffs. Violation of the symmetry condition in just one step is, in general, not allowed.

\subsection{Section \ref{sec:scep} proofs}
\begin{proof}[Proof of Theorem \ref{theorem:seq}]
When condition 1 is met, we have
\begin{equation}
(X_1, \dots, X_p, \tilde X_1, \dots, \tilde X_j)\eqd(X_1, \dots, X_p, \tilde X_1, \dots, \tilde X_j)_{\text{swap}(k)}, \quad 1\le k\le j,
\label{eq:seq}
\end{equation}
for each $j = 1, \ldots, p$ by marginalizing out $\tilde { X}_{(j+1):p}$ in \eqref{eq:knockoff-swap}. This implies both \eqref{eq:cond-exch} and \eqref{eq:knockoff-symmetry-condition}.

Assume now that \eqref{eq:cond-exch} and \eqref{eq:knockoff-symmetry-condition} hold. We prove by induction that \eqref{eq:seq} holds for $j=1,2,\dots,p$; when $j = p$, we achieve pairwise exchangeability \eqref{eq:knockoff-swap}.  When $j = 0$, there is nothing to prove.  Assume \eqref{eq:seq} holds up until $j-1$. The distribution of $(X_1, \dots, X_p, \tilde X_1, \dots, \tilde X_j)$ can be decomposed into the marginal distribution of $(\Xno{j},\tilde{ X}_{1:(j-1)})$ and the conditional distribution $(X_j, \tilde X_j) \mid  \Xno{j}, \tilde{ X}_{1:(j-1)}$. The former is symmetric in $X_k$ and $\tilde X_k$, $1\le k\le j-1$ as seen by taking the induction hypothesis and marginalizing out $X_j$. The latter is symmetric in $X_j$ and $\tilde X_j$ because of \eqref{eq:cond-exch}, and symmetric in $X_k$ and $\tilde X_k$ for $1\le k\le j-1$ because of  \eqref{eq:knockoff-symmetry-condition}.
\end{proof}

\begin{proof}[Proof of Proposition \ref{prop:trMarkov}]
Let $Z_1\sim\pi$, and the Markov kernel be given by the law $\mathcal L(\tilde Z \mid Z)$. Then the chain has  $\pi$ as a stationary distribution. Also, $(Z_1, Z_2)\eqd(Z,\tilde Z)$. Time reversibility also holds since
\begin{equation*}
(Z_t, Z_{t+1})\eqd(Z, \tilde Z)\eqd(\tilde Z, Z)\eqd(Z_{t+1},Z_t).
\end{equation*}
The converse is a direct consequence of time reversibility.
\end{proof}

\subsection{Proof of Theorem \ref{theorem:timecomp}}
\label{subapp:proof-time-comp}

\begin{proof}
Suppose we are given a procedure $\mathcal K$ that always generates valid knockoffs for $X$ given an unnormalized density function $\Phi$ for $X$ and (implicitly) the support of $\Phi$ (or the dominating measure). Below, the symbols $\p_\Phi$, $\mathcal L_\Phi$, etc., indicate that we are working in the probability space defined by $\Phi$ (together with its support) and, implicitly, the procedure $\mathcal K$; $(X,\Xk,N)$ are jointly defined on this space. Let $\pi$ be the normalized density, which is defined by $\Phi$ through
\begin{equation*}
\pi(x)=\lambda_\Phi\Phi(x),\quad x\in\rr^p.
\end{equation*}
We abuse notation slightly and for a Borel set $M$ write $\Phi(X\in M)$ for $\p_\Phi(X\in M)/\lambda_\Phi$. 

We will consider the conditional probability $\p_\Phi(N\ge2^{\#\{j:X_j\ne\tilde X_j\}}-1\mid { X},\bXk)$ and prove it is almost surely one. Conditioning on $(X,\Xk)$ makes the probability easier to analyze because it fixes the exponent $\#\{j:X_j\ne\tilde X_j\}$. We will basically identify which $2^{\#\{j:X_j\ne\tilde X_j\}}-1$ points have to be queried: any point obtained by changing $x_j$ to $\xk_j$ for $j$ in any non-empty subset of $\{j:x_j\ne\xk_j\}$. We will prove the theorem for both discrete and continuous distributions. Analysis of the conditional probability can be done directly in the discrete case, while in the continuous case we cover the possible values of $(X,\Xk)$ by a countable union of sets, and then prove the conditional probability of interest is one on each of the sets. Although more technical, the proof for the continuous case shares the same structure as that for the discrete case.

\paragraph{Discrete case.} Let $(x,\xk)$ be some pair of input and output, respectively, of $\mathcal K$. For any $S\subseteq\{1,2,\dots,p\}$, define $x_{\ch(S)}$ as $x$ except with $x_j$ changed to $\xk_j$ for all $j\in S$, and vice versa for $\xk_{\ch(S)}$, so that 
\begin{equation*}
(x,\xk)_{\swap(S)}=(x_{\ch(S)},\xk_{\ch(S)}).
\end{equation*}
We will prove that as long as $x_{\ch(S)}\neq x$, then $\mathcal K$ must have queried the oracle at $x_{\ch(S)}$; now assume $x_{\ch(S)}\neq x$. We also assume $\p_\Phi(X=x,\Xk=\xk)>0$ (otherwise $(x,\xk)$ is not a possible pair of input and output), which implies $\pi(x_{\ch(S)})>0$ for any $S\subseteq\{1,2,\dots,p\}$ by pairwise exchangeability. Let $q_\Phi(\bxk\mid x)=\p_\Phi(\bXk=\bxk\mid { X}= x)$. Also by pairwise exchangeability,
\begin{equation}
\pi(x)q_{\Phi}(\xk \mid x)=\pi(x_{\ch(S)})q_{\Phi}(\xk_{\ch(S)} \mid  x_{\ch(S)})\le\pi(x_{\ch(S)}).
\label{eq:discrete-bound}
\end{equation}
Let $A_x$ be the event that the input vector is $x$, so
\begin{equation*}
\p_\Phi(A_x)=\pi(x)=\lambda_\Phi\Phi(x).
\end{equation*}
Let $B_{\xk}$ be the event that the output of $\mathcal{K}$ is $\xk$, so
\begin{equation*}
\p_\Phi(B_{\xk} \mid  A_x)=q_{\Phi}(\xk \mid  x).
\end{equation*}
Let $C_S$ be the event that $\mathcal K$ does not query the oracle at $x_{\ch(S)}$.
Dividing \eqref{eq:discrete-bound} by $\lambda_\Phi$ gives 
\begin{equation}
\Phi(x_{\ch(S)})\ge\Phi(x)q_{\Phi}(\xk \mid  x)=\Phi(x)\p_\Phi(B_{\xk} \mid  A_x)\ge\Phi(x)\p_\Phi(B_{\xk}\cap C_S \mid  A_x).
\label{eq:K-ineq}
\end{equation}
Equation~\eqref{eq:K-ineq} holds for any $\Phi$ such that $\Phi(x),\Phi(\xk)>0$.
Consider a new unnormalized density
\begin{equation*}
    \Phi^S_{\eta}( y) = 
    \begin{cases}
    \eta \Phi( y),  &  y=x_{\ch(S)},\\
    \Phi( y), & \text{otherwise},
    \end{cases}
\end{equation*}
where $\eta\in(0,1]$. This new density has the same support as $\Phi$, so $\Phi_\eta^S(x_{\ch(S')})>0$ for any $S'\subseteq\{1,2,\dots,p\}$, and thus \eqref{eq:K-ineq} also holds for $\Phi_\eta^S$. Now consider ${\p_{\Phi^S_\eta}(B_{\xk}\cap C_S \mid  A_x)}=\p_{\Phi^S_\eta}(B_{\xk} \mid C_S, A_x)\p_{\Phi^S_\eta}(C_S \mid  A_x)$. The first probability does not depend on $\eta$ because the conditioning on $C_S$ means changing the oracle only at $x_{\ch(S)}$ does not affect the procedure $\mathcal K$ in any way. The second probability does not depend on $\eta$ either, for  the points that $\mathcal K$ queries can only depend on $\eta$ after $\mathcal K$ queries the oracle at $x_{\ch(S)}$. Therefore, $\p_{\Phi^S_\eta}(B_{\xk}\cap C_S \mid  A_x)=\p_{\Phi}(B_{\xk}\cap C_S \mid  A_x)$ for any $\eta\in(0,1]$.
Thus, we get from equation \eqref{eq:K-ineq} that (recall we are assuming $x\ne x_{\ch(S)}$)
\begin{equation*}
\begin{aligned}
\eta\Phi(x_{\ch(S)})&=\Phi^S_\eta(x_{\ch(S)})\ge\Phi^S_\eta(x)\p_{\Phi^S_\eta}(B_{\xk}\cap C_S \mid  A_x)=\Phi(x)\p_\Phi(B_{\xk}\cap C_S \mid  A_x).
\end{aligned}
\end{equation*}
Since $\Phi(x)>0$ (i.e., $x$ is a possible input), we conclude by letting $\eta\rightarrow 0$ that $\p_\Phi(B_{\xk}\cap C_S \mid  A_x)=0$. Combining this with $\p_\Phi(B_{\xk} \mid  A_x)>0$ (i.e., given $x$ as an input, $\xk$ is a possible output), we conclude $\p_\Phi(C_S\mid A_x,B_{\xk})=0$. That is, if $\mathcal{K}$ generates $\xk$ from input $x$, it must have queried $x_{\ch(S)}$.
%Note that there is nothing special about $(\tilde x_1,x_2,\dots,x_p)$, \rev{which changes $x_1$ to $\tilde x_1$ in $x$}, and it is easy to extend the proof for any other point\rev{, which we obtain by changing $x_j$ to $\xk_j$ for all $j$ in a non-empty subset of $\{1,2,\dots,p\}$, and claim that the oracle is also queried at that point with probability one, provided that it is a different point than $x$. The proof is exactly the same, with only the definitions of $C$ and $\Phi_\eta$ slightly changed}.
Thus, combining the results for all the $S$'s that make $x\ne x_{\ch(S)}$, we can claim that given $x$ as input, $\mathcal K$ outputs $\xk$ only if it has queried the oracle at least at the set of points
\begin{equation*}
H_{(x,\tilde x)}=\{x_{\ch(S)}: S\subseteq {\{j:x_j\ne\xk_j\}}, S\ne\emptyset\}.
\end{equation*}
Mathematically, when $\p_\Phi(A_x,B_{\xk})>0$,
\begin{equation}
\p_\Phi(\mathcal K\text{ queried }\Phi\text{ at }z\text{ for all }z\in H_{(x,\tilde x)} \mid  A_x,B_{\xk})=1.
\label{eq:query-all-points}
\end{equation}
Since there are $2^{\#\{j:x_j\ne\tilde x_j\}}-1$ points in $H_{(x,\tilde x)}$, we have
\begin{equation*}
\p_\Phi(N\ge2^{\#\{j:x_j\ne\tilde x_j\}}-1 \mid  A_x,B_{\xk})=1.
\end{equation*}
After marginalizing out $x$ and $\tilde x$, this leads to the claimed a.s. inequality:
\begin{equation*}
N\ge2^{\#\{j:X_j\ne\tilde X_j\}}-1.
\end{equation*}

\paragraph{Continuous case.}
The proof works similarly for the continuous case. Loosely speaking, we will construct non-overlapping hypercubes around the points in $H_{(x,\xk)}$ defined previously, and show they all contain points of query with probability one. Concretely, consider a hypercube around $z$, defined as
\begin{equation*}
F_{( z,\tilde{ z})}=\{x:|x_k-z_k|\le g(z_k,\tilde z_k),1\le k\le p\},
\end{equation*}
where
\begin{equation*}
g(z_k,\tilde z_k)=\left\{
\begin{aligned}
&\frac{|z_k-\tilde z_k|}{3}, &z_k\ne\tilde z_k,\\
&1, &z_k=\tilde z_k.
\end{aligned}
\right.
\end{equation*}
The denominator $3$ in the definition of $g$ is not essential, as long as it is large enough so $F_{(z,\tilde z)}$ does not overlap with $F_{(z,\tilde z)_{\swap(j)}}$ if $z_j\ne\tilde z_j$. Let $E_{(z,\tilde z)}$ be a joint hypercube around $(z,\tilde z)$, which is defined through $F_{(z,\tilde z)}$ and $F_{(\tilde z,z)}$ as
\begin{equation*}
E_{( z,\tilde{ z})}=\{( x,\bxk): x\in F_{( z,\tilde{ z})}, \bxk\in F_{(\tilde{ z},  z)}\}=\{( x,\bxk):|x_k-z_k|,|\xk_k-\tilde z_k|<g(z_k,\tilde z_k),1\le k\le p\}.
\end{equation*}
We use the fact that the rational points are dense in $\rr^{2p}$ to cover the entire space using a countable collection of sets. Hence, if we can prove the conditional probability $\p_\Phi(N\ge2^{\#\{j:X_j\ne\tilde X_j\}}-1\mid { X},\bXk)=1$ on every set in this collection, we can claim the corresponding equality holds unconditionally. Formally, we first consider $E_{( r,\tilde{ r})}$, where $( r,\tilde{ r})\in\mathbb Q^{2p}$ and $r_k\ne\tilde r_k,1\le k\le q$ for some positive integer $q\le p$. We will show $\p_\Phi(N\ge2^q-1\mid(X,\Xk)\in E_{(r,\tilde r)})=1$ as long as $\p_\Phi((X,\Xk)\in E_{(r,\tilde r)})>0$. Now suppose $\p_\Phi((X,\Xk)\in E_{(r,\tilde r)})>0$, which implies $\p_\Phi(X\in F_{(r,\tilde r)_{\swap(S)}})>0$ for any $S\subseteq\{1,2,\dots,p\}$. Define
\begin{equation*}
A_{(r,\tilde r)}=\{X\in F_{(r,\tilde r)}\},\quad B_{(\tilde r,r)}=\{\Xk\in F_{(\tilde r,r)}\}.
\end{equation*}
Let $S$ be any non-empty subset of $\{1,2,\dots,q\}$, so $(r,\tilde r)_{\swap(S)}\ne(r,\tilde r)$. Now we have
\begin{equation*}
\p_\Phi(({ X},\bXk)\in E_{( r,\tilde{ r})})=\p_\Phi(A_{( r,\tilde{ r})}\cap B_{(\tilde{ r}, r)})=\p_\Phi(A_{( r,\tilde{ r})})\p_\Phi(B_{(\tilde{ r},r)}\mid A_{( r,\tilde{ r})}),
\end{equation*}
and
\begin{equation*}
\begin{aligned}
\p_\Phi(({ X},\bXk)\in E_{( r,\tilde{ r})})&=\p_\Phi((X,\Xk)_{\swap(S)}\in E_{(r,\tilde r)})\\
&=\p_\Phi((X,\Xk)\in E_{(r,\tilde r)_{\swap(S)}})\\
&=\p_\Phi(X\in F_{(r,\tilde r)_{\swap(S)}},\Xk\in F_{(\tilde r,r)_{\swap(S)}})\\
&=\p_\Phi({ X}\in F_{( r,\tilde{ r})_{\swap(S)}})\p_\Phi(\tilde { X}\in F_{(\tilde{ r}, r)_{\swap(S)}}\mid { X}\in F_{( r,\tilde{ r})_{\swap(S)}})\\
&\le\p_\Phi({ X}\in F_{( r,\tilde{ r})_{\swap(S)}}).
\end{aligned}
\end{equation*}
% Now we have
% \begin{equation*}
% \p_\Phi(({ X},\bXk)\in E_{( r,\tilde{ r})})=\p_\Phi({ X}\in F_{( r,\tilde{ r})},\tilde { X}\in F_{(\tilde{ r}, r)})=\p_\Phi({ X}\in F_{( r,\tilde{ r})})\p_\Phi(\tilde { X}\in F_{(\tilde{ r}, r)}\mid { X}\in F_{( r,\tilde{ r})}),
% \end{equation*}
% \ljmargin{and}{can we somehow incorporate analogues of $A_x$ and $B_{\xk}$ and use similar notation? It would at least be useful to state the analogous events.}\wwmargin{}{Not sure how this would work; let's talk about it on Monday.}
% \begin{equation*}
% \begin{aligned}
% \p_\Phi(({ X},\bXk)\in E_{( r,\tilde{ r})})&=\p_\Phi((X,\Xk)_{\swap(1)}\in E_{(r,\tilde r)})\\
% &=\p_\Phi((X,\Xk)\in E_{(r,\tilde r)_{\swap(1)}})\\
% &=\p_\Phi(X\in F_{(r,\tilde r)_{\swap(1)}},\Xk\in F_{(\tilde r,r)_{\swap(1)}})\\
% &=\p_\Phi({ X}\in F_{( r,\tilde{ r})_{\swap(1)}})\p_\Phi(\tilde { X}\in F_{(\tilde{ r}, r)_{\swap(1)}}\mid { X}\in F_{( r,\tilde{ r})_{\swap(1)}})\\
% &\le\p_\Phi({ X}\in F_{( r,\tilde{ r})_{\swap(1)}}).
% \end{aligned}
% \end{equation*}
Hence, by dividing by the common normalizing constant $\lambda_\Phi$ in the above two equations and combining them, we get
\begin{equation}
\Phi({ X}\in F_{( r,\tilde{ r})_{\swap(S)}})\ge\Phi({ X}\in F_{( r,\tilde{ r})})\p_{\Phi}(B_{(\tilde r,r)}\mid A_{( r,\tilde{ r})}).
\label{eq:cont-ineq}
\end{equation}
Now, similar to the discrete case, we consider a new unnormalized density
\begin{equation*}
    \Phi^S_{\eta}( x) = 
    \begin{cases}
    \eta \Phi( x) &  x \in F_{( r,\tilde{ r})_{\swap(S)}},\\
    \Phi( x) & \text{otherwise},
    \end{cases}
\end{equation*}
for $\eta \in (0,1]$, which has the same support/dominating measure as $\Phi$. 
%Now that $\Phi_\eta(X\in F_{(r,\tilde r)_{\swap(S)}})>0$ for any $S\subseteq\{1,2,\dots,p\}$, 
It is easy to check that
%the terms appearing in the above equations are all well-defined, and thus the math deduction goes through and 
equation \eqref{eq:cont-ineq} also holds for $\Phi^S_{\eta}$. Let $C_S$ be the event that $\mathcal K$ does not query $\Phi$ at any point in $F_{( r,\tilde{ r})_{\swap(S)}}$. To use the same trick as in the discrete case, we next prove $\p_{\Phi^S_\eta}(C_S\cap B_{(\tilde{ r}, r)}\mid A_{( r,\tilde{ r})})$ does not depend on $\eta$.
We have $\Phi^S_\eta({ X}\in F_{( r,\tilde{ r})})=\Phi({ X}\in F_{( r,\tilde{ r})})$ because $F_{( r, \tilde{ r})}\cap F_{( r,\tilde{ r})_{\swap(S)}}$ is empty (so $\Phi^S_\eta=\Phi$ on $F_{(r,\tilde r)}$). Note that by definition,
\begin{equation}
\begin{aligned}
\p_{\Phi}(C_S\cap B_{(\tilde{ r}, r)}\mid A_{( r,\tilde{ r})})&=\e_\Phi[\p_{\Phi}(C_S\cap B_{(\tilde{ r}, r)}\mid { X})\mid A_{( r,\tilde{ r})}],\\
\p_{\Phi^S_\eta}(C_S\cap B_{(\tilde{ r}, r)}\mid A_{( r,\tilde{ r})})&=\e_{\Phi^S_\eta}[\p_{\Phi^S_\eta}(C_S\cap B_{(\tilde{ r}, r)}\mid { X})\mid A_{( r,\tilde{ r})}],
\end{aligned}
\label{eq:cond-prob-eta}
\end{equation}
and $\p_{\Phi^S_\eta}(C_S\cap B_{(\tilde{ r}, r)}\mid { X})$, as a function of ${ X}$, does not depend on $\eta$; this holds since all $\Phi^S_\eta$'s for $\eta\in(0,1]$ have the same support, and we can follow the same argument as in the discrete case. Specifically,
\begin{equation*}
\p_{\Phi^S_\eta}(C_S\cap B_{(\tilde{ r}, r)}\mid { X})=\p_{\Phi}(C_S\cap B_{(\tilde{ r}, r)}\mid { X}).
\end{equation*}
In addition, since the unnormalized density is the same on the set $F_{(r, \tilde r)}$, we have
\begin{equation*}
\mathcal L_\Phi({ X}\mid A_{( r,\tilde{ r})})=\mathcal L_{\Phi^S_\eta}({ X}\mid A_{( r,\tilde{ r})}).
\end{equation*}
The last two equations together with equations \eqref{eq:cond-prob-eta} imply $\p_{\Phi}(C_S\cap B_{(\tilde{ r}, r)}\mid A_{( r,\tilde{ r})})=\p_{\Phi^S_\eta}(C_S\cap B_{(\tilde{ r}, r)}\mid A_{( r,\tilde{ r})})$. By \eqref{eq:cont-ineq},
\begin{equation}
\begin{aligned}
\eta\Phi({ X}\in F_{( r,\tilde{ r})_{\swap(S)}})&=\Phi^S_\eta({ X}\in F_{( r,\tilde{ r})_{\swap(S)}})\\
&\ge\Phi^S_\eta(A_{( r,\tilde{ r})})\p_{\Phi^S_\eta}(B_{(\tilde{ r}, r)}\mid A_{( r,\tilde{ r})})\\
&\ge\Phi^S_\eta(A_{( r,\tilde{ r})})\p_{\Phi^S_\eta}(C_S\cap B_{(\tilde{ r}, r)}\mid A_{( r,\tilde{ r})})\\
&=\Phi(A_{( r,\tilde{ r})})\p_{\Phi}(C_S\cap B_{(\tilde{ r}, r)}\mid A_{( r,\tilde{ r})}).
\end{aligned}
\label{eq:cont-inequ}
\end{equation}
The left hand side of \eqref{eq:cont-inequ} goes to $0$ as $\eta\to0$. Thus, recall that we are assuming $\p_\Phi(({ X},\bXk)\in E_{( r,\tilde{ r})})>0$ and $\Phi({ X}\in F_{( r,\tilde{ r})})>0$, and so we get $\p_{\Phi}(C_S\cap B_{(\tilde{ r}, r)}\mid A_{( r,\tilde{ r})})=0$. By Bayes' rule, since $\p_\Phi(A_{(r,\tilde r)}\cap B_{(\tilde r,r)})>0$, we have
\begin{equation*}
\begin{aligned}
\p_\Phi(C_S\mid A_{(r,\tilde r)}, B_{(\tilde r,r)})&=\frac{\p_\Phi(C_S\cap A_{(r,\tilde r)}\cap B_{(\tilde r,r)})}{\p_\Phi(A_{(r,\tilde r)}\cap B_{(\tilde r,r)})}\\
&=\frac{\p_\Phi(A_{(r,\tilde r)})\p_{\Phi}(C_S\cap B_{(\tilde{ r}, r)}\mid A_{(r,\tilde r)})}{\p_\Phi(A_{(r,\tilde r)}\cap B_{(\tilde r,r)})}=0.
\end{aligned}
\end{equation*}
That is, unless $( X,\bXk)\in E_{( r,\tilde{ r})}$ happens with zero probability, with probability one at least one point in $F_{( r,\tilde{ r})_{\swap(S)}}$ is queried conditional on $( X,\bXk)\in E_{( r,\tilde{ r})}$. Combining the results for all the $S$'s that make $(r,\tilde r)_{\swap(S)}\ne(r,\tilde r)$, we get $2^{q}-1$ disjoint sets, each of which must contain at least one point of query. Hence, now we can claim that for any $E_{( r,\tilde{ r})}$, where $( r,\tilde{ r})\in\mathbb Q^{2p}$ and $r_k\ne\tilde r_k,1\le k\le q$, either
\begin{equation*}
\p_\Phi(( X,\bXk)\in E_{( r,\tilde{ r})})=0
\end{equation*}
or
\begin{equation*}
\p_\Phi(N\ge2^q-1\mid ( X,\bXk)\in E_{( r,\tilde{ r})})=1.
\end{equation*}
Note that this is equivalent to $\p_\Phi(N\ge2^q-1\mid X,\bXk)=1$ almost surely on $E_{( r,\tilde{ r})}$.\footnote{The conditional probability $\p_\Phi(N\ge2^q-1\mid X,\bXk)=1$ is almost surely one on a set $U$ means $\p_\Phi(N\ge2^q-1\mid (X,\Xk)=(x,\xk))=1$ for $(x,\xk)\in U\setminus V$, where $V$ is some set satisfying $\p_\Phi((X,\Xk)\in V)=0$.} These two equations imply that, as a function of $( X,\bXk)$, the conditional probability satisfies
\begin{equation*}
\p_\Phi(N\ge2^q-1\mid X,\bXk)=1,\quad a.s.\text{ on }\bigcup_{\substack{( r,\tilde{ r})\in\mathbb Q^{2p}\\r_k\ne\tilde r_k,1\le k\le q}}E_{( r,\tilde{ r})},
%\label{eq:aslarge}
\end{equation*}
because the union is over a countable index set. There is nothing special about choosing the $r_k \ne \tilde{r}_k$ on first $q$ coordinates, so we have 
\begin{equation}
\p_\Phi(N\ge2^{|D|}-1\mid X,\bXk)=1,\quad a.s.\text{ on }\bigcup_{\substack{( r,\tilde{ r})\in\mathbb Q^{2p}\\r_k\ne\tilde r_k,k\in D}}E_{( r,\tilde{ r})}, \quad D\subseteq\{1,2,\dots,p\}.
\label{eq:aslarge}
\end{equation}
Note the case $D=\emptyset$ does not follow the exact same proof, but no proof is needed in this case, since $N\ge2^{|\emptyset|}-1=0$ holds trivially. We shall thus keep in mind that \eqref{eq:aslarge} holds for any $\Phi$.

Now we go back to the conditional probability we are interested in, mathematically defined as
\begin{equation*}
f_\Phi( X,\bXk)=\p_\Phi(N\ge2^{\#\{j:X_j\ne\tilde X_j\}}-1\mid X,\bXk)=\e_\Phi[\ind_{N\ge2^{\#\{j:X_j\ne\tilde X_j\}}-1}\mid X,\bXk].
\end{equation*}
We want to show $f_\Phi( X,\bXk)=1$, a.s. Let
\begin{equation*}
T_{n,D}=\{( x,\bxk):|x_k-\xk_k|>1/n, k\in D\text{ and }x_k=\xk_k, k\notin D\}; 
\end{equation*}
therefore, $D$ is the set of coordinates where $x$ and $\xk$ could differ, and $1/n$ measures the minimum difference between these original and knockoff coordinates.
Since any point $( x,\bxk)\in\rr^{2p}$ is contained in
\begin{equation*}
T_{\left\lfloor1/\min\limits_{j,x_j\ne\xk_j}|x_j-\xk_j|\right\rfloor+1,\{j:x_j\ne\tilde x_j\}}
\end{equation*}
if $ x\ne\bxk$, and in $T_{1,\emptyset}$ if $ x=\bxk$, we have
\begin{equation*}
\rr^{2p}=\bigcup_{n=1}^\infty\bigcup_{D\subseteq {\{1,2,\dots,p\}}}T_{n,D}.
\end{equation*}
This is also a countable union, so in order to show $f_\Phi( X,\bXk)=1$ a.s., we only have to show that $f_\Phi( X,\bXk)=1$ a.s. for any $T_{n,D}$ that has positive probability of containing $( X,\bXk)$. Note that since there are exactly $|D|$ coordinates that differ for $ x$ and $\bxk$ in the set $T_{n,D}$,
\begin{equation*}
f_\Phi( X,\bXk)=\p_\Phi(N\ge2^{\#\{j:X_j\ne\tilde X_j\}}-1\mid X,\bXk)=\p_\Phi(N\ge2^{|D|}-1\mid X,\bXk)\text{ on }T_{n,D}.
\label{eq:d-delta}
\end{equation*}
% Thus,
% \begin{equation*}
% D=\bigcup_{n=1}^\infty\bigcup_{A\in2^{\{1,2,\dots,p\}}}T_{n,A}\cap D,
% \end{equation*}
% Since this is a countable union, at least one $D_{n,A}=T_{n,A}\cap D$ contains postive mass of $(X,\Xk)$. A closer inspection shows such an $A$ cannot be an empty set, since $T_{n,\emptyset}=\{(x,\xk):x=\xk\}$, on which $2^{\#\{j:X_j\ne\tilde X_j\}}-1=0\le N$ is always the case. Now we have showed if $f_\Phi(X,\tilde X)<1$ happens with positive probability, there exist a positive integer $n$ and a non-empty subset $A$ of $\{1,2,\dots,p\}$ such that $f_\Phi(X,\tilde X)\le1-\varepsilon$ a.s. on $D_{n,A}$, and $D_{n,A}$ contains postive mass of $(X,\Xk)$. As a direct consequence,
%\begin{equation*}
%D_{n,A}\subseteq\{(x,\xk):\min_{j,x_j\ne\xk_j}|x_j-\xk_j|>1/n\}.
%\end{equation*}
%Without loss of generality, we assume $A=\{1,2,\dots,q\}$, i.e.,
% \begin{equation*}
% D_{n,A}\subseteq T_{n,A}=\{(x,\xk):\min_{k\in A}|x_k-\xk_k|>1/n, x_\ell=\xk_\ell,\ell\notin A\},
% \end{equation*}
% where $|A|\ge1$. We now have
% \begin{equation*}
% f_\Phi(X,\Xk)=\p_\Phi(N\ge2^{\#\{j:X_j\ne\tilde X_j\}}-1|X,\Xk)\le1-\varepsilon\text{ on }D_{n,A}.
% \end{equation*}
So now we only need to show
\begin{equation*}
\p_\Phi(N\ge2^{|D|}-1\mid X,\bXk)=1,\quad a.s. \text{ on }T_{n,D},
\end{equation*}
which would be implied by \eqref{eq:aslarge} if we can show that
% \begin{equation*}
% \bigcup_{\substack{(r,\tilde r)\in\mathbb Q^{2p}\\r_k\ne\tilde r_k,k\in A}}E_{(r,\tilde r)}\supseteq D_{n,A}.
% \end{equation*}
% To do this, we just have to show
\begin{equation*}
\bigcup_{\substack{( r,\tilde{ r})\in\mathbb Q^{2p}\\r_k\ne\tilde r_k,k\in D}}E_{( r,\tilde{ r})}\supseteq T_{n,D}.
\end{equation*}
To see this, take any point $( x,\bxk)$ from $T_{n,D}$. Find rational numbers $r_k\in(x_k-1/5n,x_k+1/5n)$ and $\tilde r_k\in(\xk_k-1/5n,\xk_k+1/5n)$, $1\le k\le p$. We have $|r_k-\tilde r_k|>3/5n$ (hence also $r_k\ne\tilde r_k$) for $k\in D$, since $|x_k-\tilde x_k|>1/n$ for $k\in D$. We can now check that $( x,\bxk)\in E_{( r,\tilde{ r})}$. For $k\in D$ (if any), $|x_k-r_k|,|\tilde x_k-\tilde{r}_k|<1/5n<|r_k-\tilde r_k|/3$, and for $k\notin D$ (if any), $|r_k-x_k|,|\tilde r_k-\xk_k|<1/5n<1$.
\end{proof}

\subsection{Divide-and-conquer knockoffs}
\begin{proof}[Proof of Proposition \ref{prop:divide-and-conquer}] 
\label{proof:prop:divide-and-conquer}
    Consider the distribution of $ X$ conditional on $ X_C =  x_C$. Since $C$ separates $A$ and $B$ in the graph $G$, we have
\begin{equation*}
 X_A \independent  X_B  \mid    X_C.
\end{equation*}
The assumptions of the proposition then imply that
\begin{equation*}
( X_A,  X_B, \bXk_A, \bXk_B) \eqd ( X_A,  X_B, \bXk_A, \bXk_B)_{\swap(j)}  \mid    X_C,\quad j\in A\cup B
\end{equation*} 
and since $\bXk_C =  X_C$ a.s.,
\begin{equation*}
( X, \bXk) \eqd ( X, \bXk)_{\swap(j)}  \mid    X_C,\quad j\in A\cup B\cup C.
\end{equation*} 
Lastly, we note that conditional exchangeability implies marginal exchangeability, so
\begin{equation*}
( X, \bXk) \eqd ( X, \bXk)_{\swap(j)},\quad j\in A\cup B\cup C
\end{equation*}
as claimed.
\end{proof}

\subsection{Complexity proofs for Metropolized knockoff sampling}
\label{subapp:complexity-proofs}
% \begin{prop} \label{prop:consistent-proposals}
% Any valid knockoff distribution can be sampled using consistent proposal distributions.
% \end{prop}
% \begin{proof}
% Suppose $\Xk$ is valid knockoff for $X$. Now 
% \begin{equation}
% X_{1:j} \ci X_{V_{j-\textnormal{activated}}^C}  \mid   X_{V_{j-\textnormal{activated}} \setminus \{1,\dots,j\} }
% \end{equation}
% since by the defining property of the junction tree $T$, ${V_{j-\textnormal{activated}}}$ separates nodes $\{1,\dots,j\}$ from $V_{j-\textnormal{activated}}^C$ in the graph $G$. The knockoff swap property (\ref{eq:knockoff-swap}) implies 
% \begin{equation}
% \Xk_{1:j} \ci X_{V_{j-\textnormal{activated}}^C}  \mid   X_{V_{j-\textnormal{activated}} \setminus \{1,\dots,j\} }
% \end{equation}
% This implies {\color{red}[How?]} that the distribution of $\Xk_j$ given $X$ and $\Xk_{1:(j-1)}$ only depends on $X_{V_{j-\textnormal{activated}}}$. Taking $\law(\Xk_j  \mid   X, \Xk_{1:(j-1)})$ as a proposal distribution implies the claim. Note that with this proposal distribution we always accept, so the additional conditioning on $X^*_j$ has no effect.
% \end{proof}

%\subsubsection*{Proof of Theorem \ref{thm:scep-runtime}}

\begin{lemma}\label{lemma:jtree}
When the proposal distributions are compatible for the junction tree $T$, for $1\le j\le p$, $\p(\Xk_j, X^*_j \mid   X, \bXk_{1:(j-1)},  X^*_{1:(j-1)})$, depends on $X_k$ only if $k\in\bar V_j$.
\end{lemma}

\begin{proof}[Proof of Lemma \ref{lemma:jtree}]
We use induction over $j$. For the base case $j=1$, $\p(X_1^*\mid X)$ depends on $X_k$ only if $k \in\bar V_1$ by our assumption of compatible proposals. And $\p(\Xk_1\mid X,X_1^*)$ is a function of the acceptance probability
\begin{equation*}
\frac{\p(X_1^*=x_1^*\mid X_1=x_1,\Xno{1})}{\p(X_1^*=x_1\mid X_1=x_1^*,\Xno{1})}\frac{\p(X_1=x_1^*,\Xno{1})}{\p(X_1=x_1,\Xno{1})}.
\end{equation*}
The first term depends only on $X_k\in\bar V_1$ by assumption of compatible proposals. The second term depends on $X_k$ if $k$ is connected to $1$ in $G$ (or $k=1$). Since the variables are ordered by Algorithm~\ref{alg:jt-order}, $1$ only appears in $V_1$ by Lemma~\ref{lemma:var-ordering}, so $k$ has to appear in $V_1$ if $1$ and $k$ are connected. The base case is thus proved.

Suppose the claim is true for $1,\dots,j-1$. First we have $\p(X^*_j  \mid   X, \bXk_{1:(j-1)},  X^*_{1:(j-1)})$ depends on $X_k$ only if $k \in\bar V_j$ by our assumption of compatible proposals. Now $\p(\Xk_j =\xk_j  \mid X, \bXk_{1:(j-1)},  X^*_{1:j})$ is a function of the acceptance probability, which is computed from the ratio of the proposal densities (which depends only on $X_k$ for $k\in\bar V_j$ by assumption of compatible proposals) and the following ratio
\begin{equation}
\label{eq:lemma-acceptance-ratio}
\frac{\p(X_j = x_j^*, \Xno{j}, \bXk_{1:(j-1)},  X^*_{1:(j-1)})}
{\p(X_j = x_j, \Xno{j}, \bXk_{1:(j-1)},  X^*_{1:(j-1)})} 
= \frac{\p(X_j = x_j^* \mid \Xno{j})}{\p(X_j = x_j \mid \Xno{j})}\frac{\p(\bXk_{1:(j-1)},  X^*_{1:(j-1)}  \mid   X_j = x^*_j, \Xno{j})}
{ \p(\bXk_{1:(j-1)},  X^*_{1:(j-1)}  \mid   X_j = x_j, \Xno{j})}.
\end{equation}
We first consider the second term in the right hand side of the above. Consider the decomposition
\begin{equation*}
\p(\Xk_{1:(j-1)},X^*_{1:(j-1)}\mid X_j = z_j, \Xno{j})=\prod_{\ell=1}^{j-1}\p(\Xk_{\ell},X^*_{\ell}\mid X_j = z_j, \Xno{j} \Xk_{1:(\ell-1)},X^*_{1:(\ell-1)}).
\end{equation*}
The $\ell$th term in this product depends on $X_j$ only if $j \in \bar{V}_{\ell}$, by the induction hypothesis. Lemma~\ref{lemma:nested-V} then implies that for such $\ell$, $\bar{V}_\ell \subseteq \bar{V}_{j}$.
%To see this, consider the following two cases for any $k \in \bar{V}_{\ell}$.
% \begin{itemize}
%     \item If $k\le j$, then $k\in\bar V_j$ by definition.
% %    \item Consider $k > j$. Suppose for contradiction that $k$ is not a member of each node in the unique path from $V_\ell$ to $V_j$. Then by Algorithm~\ref{alg:jt-order} we have that $k$ is sampled when it first disappears from the path, and hence $k < j$. This contradiction shows that $k \in V_j$, so we have $k \in \bar{V}_j$,
%     \item If $k>j$, then $j,k\in V_\ell$. Before $j$ is sampled, each time a node containing $j$ and $k$ (e.g., $V_\ell$) is selected, since $j$ and $k$ are not sampled they also appear in its neighbor. By a recursive argument, before $V_j$ is selected, each time the node containing $j$ and $k$ is selected, $j$ does not get sampled and neither does $k$ ($k$ is sampled after $j$). Hence, before $V_j$ is selected, there is always at least one remaining node that contains both $j$ and $k$. By Algorithm~\ref{alg:jt-order} and Lemma~\ref{lemma:var-ordering}, $V_j$ is the last selected node $j$ appears in, so it has to contain both $j$ and $k$ (otherwise no node contains both $j$ and $k$ at this point). This means that $k\in V_j\subseteq\bar V_j$.
% \end{itemize}
Any term in the product that does not depend on $X_j$ will be identical in the numerator and denominator of \eqref{eq:lemma-acceptance-ratio} and so will cancel. Together, this shows that the second term on the right hand side of \eqref{eq:lemma-acceptance-ratio} depends only on $k$ for $k \in \bar{V}_j$.
% Now we show $\bar V_\ell\subseteq\bar V_j$ for $\ell\le j$. If $k \in\bar V_\ell$ for some $\ell \le j$, then we consider the following two cases.
% \begin{itemize}
%     \item If $k\le j$ then $k\in\bar V_j$ directly by definition;
%     \item \sout{if $k>j$ then then $k>\ell$, so $k\in\bar V_\ell\Rightarrow k\in V_\ell$. But also $k\in V_k$, so $k \in V_j$ by the running intersection property (recall that the variables are ordered according to the junction tree by Algorithm~\ref{alg:jt-order}).}
% \end{itemize}

Next, the numerator and denominator of the first term of the right hand side of \eqref{eq:lemma-acceptance-ratio} only depends on $X_k$ if $k=j$ or $k$ is connected to $j$ in $G$. If $k=j$, then $k\in\bar V_j$ by definition; otherwise $k\in\bar V_j$ by Lemma~\ref{lemma:G-V-relationship}.
% \begin{itemize}
%     \item If $k\le j$, then $k\in\bar V_j$ by definition;
%     \item Consider $k > j$.
%     By property 2 of the junction tree's definition, $k$ co-appears with $j$ at least once in some node $V_\ell$. But $j$ does not appear in any node after $V_j$ by Lemma~\ref{lemma:var-ordering}, so there is some $\ell\le j$ such that $\{j,k\}\subseteq V_\ell$. By the same argument as above, this implies $k\in V_j\subseteq\bar V_j$ ($k$ and $j$ always co-appear in at least one remaining node before $V_j$ is selected).
% \end{itemize}
Now we have showed \eqref{eq:lemma-acceptance-ratio} depends only on $X_k$ for $k\in\bar V_j$ and the proof of the lemma is complete.
\end{proof}

\begin{proof}[Proof of Theorem \ref{thm:scep-runtime}]
Take $\gamma = 1$ for simplicitly, and take a proposal distribution $q_j$ that can be sampled from and evaluated without an evaluation of $\Phi$ (e.g., an independent Gaussian proposal in the continuous setting). We will show how Algorithm~\ref{alg:jt-order} uses the conditional dependence structure encoded in the graph $G$ to make computations of \eqref{eq:seq-decomp} simpler.

Define
\begin{equation*}
F_j( X_{V_j} =  z_{V_j}) := \p(\Xk_j, X^*_j  \mid    X_{V_j} =  z_{V_j},  X_{V_j^c}, \bXk_{1:(j-1)},  X^*_{1:(j-1)}).    
\end{equation*}
By the definition of Metro, we can write $F_j$ as the product of the proposal density and the acceptance/rejection probability: 
\begin{align*}
F_j( X_{V_j} =  z_{V_j})&=\p(\Xk_j, X^*_j   \mid    X_{V_j} =  z_{V_j},  X_{V_j^c}, \bXk_{1:(j-1)},  X^*_{1:(j-1)}) \\
&= \p(X^*_j \mid    X_{V_j} =  z_{V_j},  X_{V_j^c}, \bXk_{1:(j-1)},  X^*_{1:(j-1)}) \p(\Xk_j \mid    X_{V_j} =  z_{V_j},  X_{V_j^c}, \bXk_{1:(j-1)},  X^*_{1:j}) \\
&= q^{(2)}
    \alpha^{\ind_\text{accept}}(1 - \alpha)^{\ind_\text{reject}};
\end{align*}
where
\begin{align*}
 q^{(1)} &= \p(X_j^* = z_j  \mid   X_j = x^*_j,  X_{V_j\setminus j} =  z_{V_j \setminus j},  X_{V_j^c}, \bXk_{1:(j-1)},  X^*_{1:(j-1)}) \\
 q^{(2)} &= \p(X_j^* = x_j^*  \mid  X_j=z_j,  X_{V_j\setminus j} =  z_{V_j\setminus j},  X_{V_j^c}, \bXk_{1:(j-1)},  X^*_{1:(j-1)}). 
\end{align*}
are the proposal terms and 
\begin{equation*}
\alpha = \min \left(1,\frac{q^{(1)} \p(X_j = x^*_j,  X_{V_j \setminus j} =  z_{V_j \setminus j},  X_{V_j^c}, \bXk_{1:(j-1)},  X^*_{1:(j-1)})}{q^{(2)}\p(X_j = z_j,  X_{V_j \setminus j} =  z_{V_j \setminus j},  X_{V_j^c}, \bXk_{1:(j-1)},  X^*_{1:(j-1)})} \right)
\end{equation*}
is the acceptance probability. 
Sequentially decomposing the ratio of probabilities in the term $\alpha$, we get
\begin{equation}
\begin{aligned}
&\frac{\p(X_j = x^*_j,  X_{V_j \setminus j} =  z_{V_j \setminus j},  X_{V_j^c}, \bXk_{1:(j-1)},  X^*_{1:(j-1)})}{\p(X_j = z_j,  X_{V_j \setminus j} =  z_{V_j \setminus j},  X_{V_j^c}, \bXk_{1:(j-1)},  X^*_{1:(j-1)})}\\
&=\frac{\p(X_j = x^*_j,  X_{V_j \setminus j} =  z_{V_j \setminus j},  X_{V_j^c})}{\p(X_j = z_j,  X_{V_j \setminus j} =  z_{V_j \setminus j},  X_{V_j^c})}\times\prod_{k=1}^{j-1}\frac{\p(\Xk_k,X_k^*\mid X_j = x^*_j,  X_{V_j \setminus j} =  z_{V_j \setminus j},  X_{V_j^c},\Xk_{1:(k-1)},X^*_{1:(k-1)})}{\p(\Xk_k,X_k^*\mid X_j = z_j,  X_{V_j \setminus j} =  z_{V_j \setminus j},  X_{V_j^c},\Xk_{1:(k-1)},X^*_{1:(k-1)})}\\
&=\frac{\Phi(X_j = x^*_j,  X_{V_j \setminus j} =  z_{V_j \setminus j},  X_{V_j^c})}{\Phi(X_j = z_j,  X_{V_j \setminus j} =  z_{V_j \setminus j},  X_{V_j^c})}\times\prod_{k=1}^{j-1}\frac{\p(\Xk_k,X_k^*\mid X_j = x^*_j,  X_{V_j \setminus j} =  z_{V_j \setminus j},  X_{V_j^c},\Xk_{1:(k-1)},X^*_{1:(k-1)})}{\p(\Xk_k,X_k^*\mid X_j = z_j,  X_{V_j \setminus j} =  z_{V_j \setminus j},  X_{V_j^c},\Xk_{1:(k-1)},X^*_{1:(k-1)})},
\label{eq:step-decomp}
\end{aligned}
\end{equation}
where all we did in the second equality was cancel the normalizing constants in the first ratio.
In light of Lemma~\ref{lemma:jtree}, the ratio
\begin{equation*}
\frac{\p(\Xk_k,X_k^*\mid X_j = x^*_j,  X_{V_j \setminus j} =  z_{V_j \setminus j},  X_{V_j^c},\Xk_{1:(k-1)},X^*_{1:(k-1)})}{\p(\Xk_k,X_k^*\mid X_j = z_j,  X_{V_j \setminus j} =  z_{V_j \setminus j},  X_{V_j^c},\Xk_{1:(k-1)},X^*_{1:(k-1)})}
\end{equation*}
is one unless $j\in\bar V_k$, since the value of $X_j$ is the only one that differs between the numerator and denominator. Recall that $\bar V_k=V_k\cup\{1,2,\dots,k\}$ and note $j>k$, so $j\in\bar V_k$ is equivalent to $j\in V_k$. Thus,  \eqref{eq:step-decomp} gives 
\begin{multline*}
\frac{\p(X_j = x^*_j,  X_{V_j \setminus j} =  z_{V_j \setminus j},  X_{V_j^c}, \bXk_{1:(j-1)},  X^*_{1:(j-1)})}{\p(X_j = z_j,  X_{V_j \setminus j} =  z_{V_j \setminus j},  X_{V_j^c}, \bXk_{1:(j-1)},  X^*_{1:(j-1)})}\\
=\frac{\Phi(X_j = x^*_j,  X_{V_j \setminus j} =  z_{V_j \setminus j},  X_{V_j^c})}{\Phi(X_j = z_j,  X_{V_j \setminus j} =  z_{V_j \setminus j},  X_{V_j^c})}\times\\
\prod_{k:k<j,\ j\in V_k}\frac{\p(\Xk_k,X_k^*\mid X_j = x^*_j,  X_{V_j \setminus j} =  z_{V_j \setminus j},  X_{V_j^c},\Xk_{1:(k-1)},X^*_{1:(k-1)})}{\p(\Xk_k,X_k^*\mid X_j = z_j,  X_{V_j \setminus j} =  z_{V_j \setminus j},  X_{V_j^c},\Xk_{1:(k-1)},X^*_{1:(k-1)})}.
\end{multline*}
Now we will show that the numerators in the product satisfy
\begin{multline*}
\p(\Xk_k,X_k^*\mid X_j = x^*_j,  X_{V_j \setminus j} =  z_{V_j \setminus j},  X_{V_j^c},\Xk_{1:(k-1)},X^*_{1:(k-1)})\\
=F_k(X_j = x^*_j,  X_{V_k \cap V_j \setminus j}=  z_{V_k \cap V_j\setminus j},  X_{V_k \setminus V_j} =  x_{V_k \setminus V_j}).
\end{multline*}
By the definition of $F_k$, we need to show that
\begin{multline*}
\p(\Xk_k,X_k^*\mid X_j = x^*_j,  X_{V_j \setminus j} =  z_{V_j \setminus j},  X_{V_j^c},\Xk_{1:(k-1)},X^*_{1:(k-1)})\\
=\p(\Xk_k,X_k^*\mid X_j = x^*_j,  X_{V_k \cap V_j \setminus j} =  z_{V_k \cap V_j\setminus j}, X_{V_k \setminus V_j} =  x_{V_k \setminus V_j}, X_{V_k^c},\Xk_{1:(k-1)},X^*_{1:(k-1)}).
\end{multline*}
Inspecting the equation, we note that the only difference between the two quantities is the value of the variables that are being conditioned on, and only the values of $X_{V_j\setminus V_k}$ are are different. Thus, we only need to show this set of values do not affect the conditional density. By Lemma~\ref{lemma:jtree}, we just have to show that $V_j\setminus V_k$ does not overlap with $\bar V_k$. To see this, take any $\ell\in\bar V_k=V_k\cup\{1,2,\dots,k\}$. If $\ell\in V_k$, then certainly $\ell\notin V_j\setminus V_k$. Now we consider the case where $\ell\in\{1,2,\dots,k\}$ and thus less than $j$. If $\ell\in V_j\setminus V_k$, then it must be in $V_j$. By Lemma~\ref{lemma:var-ordering}, we have $V_j=V_\ell$, which means variables $\ell,\dots,j$ are all sampled when $V_\ell$ is selected; specifically, $k$ is sampled when $V_\ell$ is selected, so $V_k=V_\ell=V_j$. But in this case certainly $V_j\setminus V_k=\emptyset$, which is a contradiction.

Similarly, we also have that the corresponding denominators in the product satisfy 
\begin{multline*}
\p(\Xk_k,X_k^*\mid X_j = z_j,  X_{V_j \setminus j} =  z_{V_j \setminus j},  X_{V_j^c},\Xk_{1:(k-1)},X^*_{1:(k-1)})\\
=F_k(X_j = z_j,  X_{V_k \cap V_j \setminus j} =  z_{V_k \cap V_j\setminus j},  X_{V_k \setminus V_j} =  x_{V_k \setminus V_j}).
\end{multline*}
Combining all these together, the acceptance probability becomes
\begin{multline*}
 \alpha = \min \left(1,\frac{q^{(1)} \p(X_j = x^*_j,  X_{V_j \setminus j} =  z_{V_j \setminus j},  X_{V_j^c}, \bXk_{1:(j-1)},  X^*_{1:(j-1)})}{q^{(2)}\p(X_j = z_j,  X_{V_j \setminus j} =  z_{V_j \setminus j},  X_{V_j^c}, \bXk_{1:(j-1)},  X^*_{1:(j-1)})} \right) \\
 = \min \left(1,
\frac{q^{(1)} c^{(1)} \Phi(X_j = x^*_j,  X_{V_j \setminus j} =  z_{V_j \setminus j},  X_{V_j^c})}
{q^{(2)} c^{(2)} \Phi(X_j = z_j,  X_{V_j \setminus j} =  z_{V_j \setminus j},  X_{V_j^c})} \right),
\end{multline*}
where
\begin{align*}
 c^{(1)} &=\prod_{k:k<j,\ j\in V_k} F_k(X_j = x^*_j,  X_{V_k \cap V_j \setminus j} =  z_{V_k \cap V_j\setminus j},  X_{V_k \setminus V_j} =  x_{V_k \setminus V_j}), \\
 c^{(2)} &=\prod_{k:k<j,\ j\in V_k} F_k(X_j = z_j,  X_{V_k \cap V_j \setminus j} =  z_{V_k \cap V_j\setminus j},  X_{V_k \setminus V_j} =  x_{V_k \setminus V_j}).
\end{align*}
Note that the only difference between $c^{(1)}$ and $c^{(2)}$ is changing $x_j^*$ to $z_j$. From this expression, we see that $\left\{F_j(X_{V_j}= z_{V_j}): z_\ell\in\{x_\ell,x_\ell^*\}\text{ for all }\ell\in V_j\right\}$ can be computed in terms of 
\begin{enumerate}
    \item $F_{k}(X_{V_k}= z_{V_k})$ for all $k < j$ with $z_{V_k}$ such that $z_\ell \in \{x_\ell, x^*_\ell\}$ for all $\ell \in V_k$,
    \item $\Phi( X_{V_j} =  z_{V_j},  X_{V_j^c} =  x_{V_j^c})$
        for all $ z_{V_j}$ with $z_\ell \in \{x_\ell, x^*_\ell\}$ for all $\ell \in V_j$.
\end{enumerate}
Thus, it requires $2^{ \abs{V_j}   }$ evaluations of $\Phi$ to compute $\left\{F_j(X_{V_j}= z_{V_j}): z_\ell\in\{x_\ell,x_\ell^*\}\text{ for all }\ell\in V_j\right\}$ from the previously computed values of $F_k(X_{V_k}=z_{V_k})$ for $k < j$. Since $\abs{V_j} \le w+1$, $1\le j\le p$, we have that the total number of queries of $\Phi$ is $O(p2^w)$. Having access to the $F_k(X_{V_k}=z_{V_k})$'s is sufficient to run the algorithm, because at each step of the algorithm, it is clear that the acceptance ratio $\alpha$ can be computed from these $F_k(X_{V_k}=z_{V_k})$'s, so the proof is complete.
\end{proof}

More generally, if an evaluation of $\Phi$ costs $a$ units of computation and a floating point operation requires $1$ unit, then the same proof shows that the algorithm takes $O(p (p + a) 2^w)$ since computing the $F_j(X_{V_j}= z_{V_j})$ requires a total of $O(p 2^w)$ evaluations of $\Phi$ and each  $F_j(X_{V_j}= z_{V_j})$ requires $O(p)$ floating point operations to compute $c^{(1)}$ and $c^{(2)}$.

%\ejc{My sense is that to a very large extent, we are re-writing below the proof from Section A. So there is a lot of redundancy here. Please remove it keeping as much of the text from Section A, which I edited. The best is to have a subsection with preparatory lemmas (e.g. Lemmas 1 and 2 and then use them as needed in A and B. I will not touch the two proofs until I hear from you. } {\color{blue} [WW: Thanks for the suggestion. Now I introduce Lemmas 1, 2, 3 in Appendix A and all the proofs are deferred to this section. Lemma 4 is not used in Appendix A so it's introduced in this section. I could also put all those in a new appendix section if that's preferred (which will be the new Appendix A, and current Appendix A will be B, etc.).]}

\subsubsection*{MTM}
For the MTM method of Section \ref{subsec:MTM}, the proposal distribution requires evaluations of $\Phi$, so the requirements of Theorem \ref{thm:scep-runtime} do not hold. The proof of Theorem \ref{thm:scep-runtime} can easily be adapted to apply to MTM, however. Instead, for the MTM method we see that at step $j$ we need access to 
\begin{equation*}
\p(X_j = z_j, \Xno{j} = \xno{j}, \bXk_{1:(j-1)} = \bXk_{1:(j-1)},  X^*_{1:(j-1)} =  x^*_{1:(j-1)})
\end{equation*}
up to a common constant for $z_j \in C^{m,t}_{x_j} \cup C^{m,t}_{x_j^*}$. By an analysis similar to the proof of Theorem \ref{thm:scep-runtime}, it suffices to have access to
\begin{enumerate}
    \item $F_k(X_{V_k}= z_{V_k})$ for all $k < j$ for $z_{V_k}$  such that $z_\ell \in C^{m,t}_{x_\ell} \cup C^{m,t}_{x_\ell^*}$ for all $k \in V_k$, 
    \item $\Phi(X_j = z_j, \Xno{j} = \xno{j})$ for all $z_j \in C^{m,t}_{x_j} \cup C^{m,t}_{x_j^*}$.
\end{enumerate}
In order to compute $F_j(X_{V_j}= z_{V_j})$ for all $ z_{V_j}$ such that $z_\ell \in C^{m,t}_{x_\ell} \cup C^{m,t}_{x_\ell^*}$ for all $\ell \in V_j$ for use in later steps, we additionally need to compute
\begin{equation*}
    \Phi( X_{V_j} =  z_{V_j},  X_{V_j^c} =  x_{V_j^c}) \text{ for all }  z_{V_j} \text{ such that } z_\ell \in C^{m,t}_{x_\ell} \cup C^{m,t}_{x_\ell^*} \text{ for all } \ell \in V_j.
\end{equation*}
Thus, MTM requires $O(p(3m + 1)^w)$ queries of $\Phi$, where $3m+1$ is an upper bound on $|C^{m,t}_{x_\ell} \cup C^{m,t}_{x_\ell^*}|$ for all $\ell$.

\subsubsection*{Discrete distributions with small support}
Consider the direct methods for discrete distributions with small support in Section \ref{subsec:rejection-free}, which can be viewed as a Metro algorithm that never rejects. The proof of Theorem \ref{thm:scep-runtime} can easily be adapted to apply to the this case. Note that since the procedure never rejects, we have $X^* = \Xk$ and we can omit writing terms of $X^*$ in all of the following discussion.

Let $C_j$ be the support of $X_j$ and suppose that $\abs{C_j} \le K$ for all $j$. Then at step $j$, the method requires access to
\begin{equation*}
\p(X_j = z_j, \Xno{j} = \xno{j}, \bXk_{1:(j-1)} = \bxk_{1:(j-1)})
\end{equation*}
for all $z_j \in C_j$. 
By an analysis similar to the proof of Theorem \ref{thm:scep-runtime}, it suffices to have access to
\begin{enumerate}
    \item $F_k(X_{V_k}= z_{V_k})$ for all $k < j$ for $ z_{V_k}$ such that $z_\ell \in C_k$ for all $\ell \in V_j$.
    \item $\Phi(X_j = z_j, \Xno{j} = \xno{j})$ for all $z_j \in C_j$. 
\end{enumerate}
In order to compute $F_j(X_{V_j}= z_{V_j})$ for all $ z_{V_j}$ such that $z_\ell \in C_\ell$ for all $\ell \in V_j$ for use in later steps, we additionally need to compute
\begin{equation*}
    \Phi( X_{V_j} =  z_{V_j},  X_{V_j^c} =  x_{V_j^c}) \text{ for all }  z_{V_j} \text{ such that } z_\ell \in C_\ell \text{ for all } \ell \in V_j.
\end{equation*}
Thus, the rejection-free procedure requires $O(pK^w)$ queries of $\Phi$.

\subsection{Lower bounds for graphical models}
\label{sec:chordal}

\begin{proof}[Proof of Proposition \ref{prop:gaussian-optimality}]
We only have to prove that we can design a Metro algorithm that meets the requirement in the proposition. The rest of the proposition is implied by Corollary \ref{coro:chordal-lb}, proved next. We focus on the continuous case (and when the support is all of $\rr^p$), because we are interested in the Gaussian case. % For the sake of this proposition, the algorithm can choose any proposal in those cases. \ejc{Which cases?}  In the case where the distribution is supported on $\rr^p$, 
Since Metro can only learn about the distribution by making queries to the oracle, we describe an algorithm which first makes about $p^2/2$ queries to attempt to recover the covariance matrix. This can be done in such a way that if the model is Gaussian as described in Proposition~\ref{prop:gaussian-optimality}, we recover the correct covariance matrix. Thus, covariance-guided proposals will accept at every step and get us a knockoff vector that differs with the input original vector at every coordinate (with probability one).

Let $e_j$ be the $j$th vector of the canonical basis. Assume the model is $\mathcal N( 0,\bm\Sigma)$ let $\Phi$ be the unnormalized density. Then 
\begin{equation*}
\begin{aligned}
\Phi(0)& =W\cdot \exp\left(-\frac12 0\bm\Sigma^{-1} 0^\top\right)=W,\\
\Phi( e_j) & =W\cdot \exp\left(-\frac12 e_j\bm\Sigma^{-1} e_j^\top\right)=W\cdot\exp\left(-\frac12\left(\bm\Sigma^{-1}\right)_{jj}\right),\ 1\le j\le p,\\
\Phi( e_j + e_k )& =W\cdot\exp\left(-\frac12\left(\left(\bm\Sigma^{-1}\right)_{jj}+\left(\bm\Sigma^{-1}\right)_{kk}+2\left(\bm\Sigma^{-1}\right)_{jk}\right)\right),\ j\ne k,
\end{aligned}
\end{equation*}
where $W$ is an unknown positive constant. Hence, if we query the oracle at the above $1+p(p+1)/2$ points, we can always solve for a potential precision matrix. The algorithm's next step depends on the solution to these equations.
\begin{itemize}
    \item If the matrix formed by the solution to this system of equations is positive definite, and reflects the structure of the graph $G$, i.e., the $(i,j)$th entry is non-zero only if $i=j$ or $i$ and $j$ are connected by an edge in $G$, then the algorithm inverts this solution matrix to get $\bm\Sigma$, and then proceeds with covariance-guided proposals with any positive $s$ which makes $\bm\Gamma(s)$ in  \eqref{eq:knockoff-cov} positive definite.
    \item Otherwise, the model must not be a multivariate Gaussian distribution with zero mean, positive definite covariance matrix and have the required conditional independence structure. In that case, the algorithm will just choose any proposal distribution (e.g., independent Gaussian proposals). 
\end{itemize}

Since running Metro with the covariance-guided proposal requires $O(p2^w)$ queries of $\Phi$, this algorithm in total requires $O(p^2+p2^w)$ queries of $\Phi$. If indeed $\Phi( x)\propto\exp\left(-x\bm\Sigma^{-1} x^\top/2\right)$, the algorithm will recover the right covariance matrix $\bm\Sigma$, and therefore will never reject, and produce a knockoff $\tilde{ X}$ such that $X_j\ne\tilde{X}_j$ for all $j$ (with probability one).
\end{proof}

\begin{proof}[Proof of Corollary \ref{coro:chordal-lb}]
Call the procedure $\mathcal K$. We argue by contradiction and show that if the Corollary did not hold, we would be able to exploit $\mathcal K$ and construct an algorithm that contradicts Theorem~\ref{theorem:timecomp}. Loosely speaking, if there is one clique ${c_0}$ for which, with positive probability, the inequality fails to hold, then we can design an algorithm that generates knockoffs for any $|{c_0}|$-dimensional random vector $X_{c_0}$ by inferring the ``missing'' variables $X_{\{1:p\}\setminus{c_0}}$, applying $\mathcal K$ to get $(\Xk_{c_0},\Xk_{\{1:p\}\setminus {c_0}})$, and keeping only $\Xk_{c_0}$, which is a valid knockoff of $X_{c_0}$.

Formally, if there exists a $\Phi_0 = \prod_{c \in C} \phi_c( x_c)$ such that with positive probability
\begin{equation*}N < \max_{c \in C}2^{\#\{j\in c: X_j \ne \Xk_j\}} - 1,
\end{equation*}
then there must exist some clique $c_0\in C$ such that simultaneously $\#\{j\in c_0: X_j \ne \Xk_j\} = \max_{c \in C} \#\{j\in c: X_j \ne \Xk_j\}$ and $N < 2^{\#\{j\in c_0: X_j \ne \Xk_j\}} - 1$ with positive probability. Note that such a $c_0$ must not be empty ($|c_0|\ge1$), since we cannot have $N<0$ with positive probability. Fix this distribution $\Phi_0$ and this clique $c_0$. For each $x_{c_0}$ in the domain, use $\Phi_0(\cdot \mid x_{c_0})$ to denote the normalized density of $X_{\{1:p\} \setminus {c_0}} \mid X_{c_0} = x_{c_0}$ when $X\sim\Phi_0$, and consider any sampler $\mathcal S_{x_{c_0}}$ for this conditional distribution. This sampler takes a $|c_0|$-dimensional vector $x_{c_0}$ and produces a sample from $\Phi_0(\cdot \mid x_{c_0})$. Now consider the following generic procedure for knockoff sampling:
\begin{enumerate}
    \item The user inputs unnormalized density $\Psi_{c_0}$ and vector $ X_{c_0}$, where $ X_{c_0}$ follows the distribution induced by the unnormalized density $\Psi_{c_0}$.
    \item Sample $ X_{\{1:p\} \setminus {c_0}}$ from the conditional distribution $\Phi_{c_0}(\cdot \mid X_{c_0})$ using the  sampler $\mathcal S_{X_{c_0}}$.
    \item Provide $\Phi^\prime( z) := \Psi_{c_0}( z_{c_0})\Phi_{c_0}( z_{\{1:p\} \setminus{c_0}} \mid  z_{c_0})$ as a function of $ z$ and the realization $( X_{c_0},  X_{\{1:p\} \setminus {c_0}})$ as an input to procedure $\mathcal K$, which then returns $(\bXk_{{c_0}}, \bXk_{\{1:p\} \setminus {c_0}})$. Let $N$ bet the number of queries of $\Phi^\prime$ required by $\mathcal K$.
    \item Return $\bXk_{c_0}$.
\end{enumerate}
This procedure queries $\Psi_{c_0}$ exactly $N$ times, since step 2 uses $\mathcal S_{x_{c_0}}$, which does not rely on $\Psi_{c_0}$ and does not query it. Furthermore, step 3 queries $\Phi^\prime$ and hence $\Psi_{c_0}$ exactly $N$ times. We now show that the procedure is also guaranteed to produce valid knockoffs for any $\Psi_{c_0}$. To do this, we only need to show that $\Phi'(z)$ factors over $G$; note that
\begin{equation*}
 \Phi^\prime( z) = \Psi_{c_0}( z_{c_0})\frac{\Phi_0(z)}{ \int\Phi_0( z_{c_0},w_{\{1:p\}\setminus {c_0}})\di{w_{\{1:p\}\setminus{c_0}}} }=\underbrace{\frac{\Psi_{c_0}(z_{c_0})\phi_{c_0}(z_{c_0})}{\int\Phi_0( z_{c_0},w_{\{1:p\}\setminus {c_0}})\di{w_{\{1:p\}\setminus {c_0}}}}}_{\text{only depends on }z_{c_0}}\prod_{\substack{c\in C\\c\ne {c_0}}}\phi_{c}(z_{c}).
\end{equation*}
Since $\Phi'$ has the assumed structure implied by $G$, by the assumption on the validity of $\mathcal K$, $(\bXk_{{c_0}}, \bXk_{\{1:p\} \setminus {c_0}})$ is a valid knockoff for the augmented random vector $( X_{c_0},  X_{\{1:p\} \setminus {c_0}})$. Marginally, $ X_{c_0} \sim \Psi_{c_0}$, so we simply marginalize out $ X_{\{1:p\} \setminus {c_0}}$ and $\bXk_{\{1:p\} \setminus {c_0}}$ which preserves pairwise exchangeability.

Finally, because ${c_0}$ is the complete graph on $|{c_0}|$ coordinates, this is a generic knockoff sampler for random vectors of dimension $|{c_0}|$. Specifically, by our initial choice of $\Phi_0$, letting $\Psi_{c_0}$ correspond to $\Phi_{c_0}$ (the marginal density of $ X_{c_0}$ when $ X \sim \Phi_0$) we have $N < 2^{\#\{j \in {c_0}: X_j \ne \Xk_j\}} - 1$ with positive probability. This contradicts Theorem $\ref{theorem:timecomp}$, which says the inequality must hold with zero probability for any input density, including $\Phi_{c_0}$.
\end{proof}

\section{Efficient matrix inversion for covariance-guided proposals}
\label{app:mat-inv}
Let $\bm\Sigma_j$ be the matrix composed of the first $(p+j)$ rows and columns of $$\bm\Gamma=\left[
  \begin{array}{cc}
\bm\Sigma & \bm\Sigma-\diag( s)\\
\bm\Sigma-\diag( s) & \bm\Sigma\\
  \end{array}
\right].$$
We want to find the inverses of $\bm\Sigma_0,\bm\Sigma_1,\dots,\bm\Sigma_{p-1}$ (assuming $\bm\Sigma_{p-1}$ is invertible). Note that
\begin{equation*}
\bm\Sigma_{j+1}=\left[
  \begin{array}{cc}
\bm\Sigma_j & \gamma_{j+1}\\
\gamma_{j+1}^\top & \sigma_{j+1}^2\\
  \end{array}
\right],
\end{equation*}
where $\sigma_{j+1}^2$ is the $(j+1)$th diagonal element of $\bm\Sigma$, and $\gamma_{j+1}$ is the truncated $(p+j)$th column of $\bm\Gamma$.
We have
\begin{equation*}
\bm\Sigma_{j+1}^{-1}=\left[
  \begin{array}{cc}
\left(\bm\Sigma_j-\frac{1}{\sigma_{j+1}^2}\gamma_{j+1}\gamma_{j+1}^\top\right)^{-1} & -\frac{1}{\sigma_{j+1}^2-\gamma_{j+1}^\top\bm\Sigma_j^{-1}\gamma_{j+1}}\bm\Sigma_j^{-1}\gamma_{j+1}\\
-\frac{1}{\sigma_{j+1}^2-\gamma_{j+1}^\top\bm\Sigma_j^{-1}\gamma_{j+1}}\gamma_{j+1}^\top\bm\Sigma_j^{-1} & \frac{1}{\sigma_{j+1}^2-\gamma_{j+1}^\top\bm\Sigma_j^{-1}\gamma_{j+1}}\\
  \end{array}
\right].
\end{equation*}
And by the Sherman--Morrison formula,
\begin{equation*}
\left(\bm\Sigma_j-\frac{1}{\sigma_{j+1}^2}\gamma_{j+1}\gamma_{j+1}^\top\right)^{-1}=\bm\Sigma_j^{-1}-\frac{\bm\Sigma_j^{-1}\gamma_{j+1}\left(\bm\Sigma_j^{-1}\gamma_{j+1}\right)^{\top}}{-\sigma_{j+1}^2+\gamma_{j+1}^\top\bm\Sigma_j^{-1}\gamma_{j+1}}.
\end{equation*}
With all the elements of recursion in place, we can invert $\bm\Sigma_0$ and recursively calculate the inverse matrices of $\bm\Sigma_1,\dots,\bm\Sigma_{p-1}$. We have made code available that implements this recursion efficiently.

\section{Group knockoffs}
\label{app:grpko}
We can easily generalize our work to the group knockoff filter first presented in \citet{pmlr-v48-daia16} to control the group false discovery rate. As in that work, let $\{I_1,I_2,\dots,I_k\}$ be a partition of $\{1,2,\dots,p\}$, and suppose we want to construct $\bXk$ such that for each $j = 1, \ldots, k$,
\begin{equation*}
({X}_{I_1},{X}_{I_2},\dots,{ X}_{I_k},\tilde { X}_{I_1},\tilde { X}_{I_2},\dots,\tilde { X}_{I_k})\eqd ({ X}_{I_1},{ X}_{I_2},\dots,{ X}_{I_k},\tilde { X}_{I_1},\tilde { X}_{I_2},\dots,\tilde { X}_{I_k})_{\text{swap}(I_j)}.
\end{equation*}
At each step, we can draw a proposal ${ X}^*_{I_j}= x^*_{I_j}$ from a faithful multivariate distribution, and accept it with probability
\begin{equation*}
\min\left(1,\frac{q_j( x_{I_j} \mid   x^*_{I_j})\p(\Xno{I_j}=\xno{I_j},{ X}_{I_j}= x^*_{I_j},\tilde { X}_{I_{1:(j-1)}}=\tilde{ x}_{I_{1:(j-1)}},{ X}^*_{I_{1:(j-1)}}= x^*_{I_{1:(j-1)}})}{q_j( x^*_{I_j} \mid   x_{I_j}) \p (\Xno{I_j}=\xno{I_j},{ X}_{I_j}= x_{I_j},\tilde { X}_{I_{1:(j-1)}}=\tilde{ x}_{I_{1:(j-1)}},{ X}^*_{I_{1:(j-1)}}= x^*_{I_{1:(j-1)}})}\right).
\end{equation*}

\section{Extended simulation results} \label{app:sim-details}

% \subsection{Comparison of rejection-free methods and MTM} \label{subsec:rf-mtm-compare-details}

\subsection{Discrete Markov chains simulation details}
In Figure \ref{fig:dmc-compare}, the best MTM specification is taken from $\{(\gamma,m,t):\gamma=0.999,1\le m\le10,1\le t\le 5\}$ for $\alpha=0.2,0.3,0.4,0.5$, and from $\{(\gamma,m,t):\gamma=0.999,1\le m\le10,1\le t\le 5\}\cup\{(\gamma,m,t):m=4,t=1,\gamma=0.1,0.2,0.3,0.4,0.5,0.6,0.7,0.8,0.9,0.999\}$ for $\alpha=0,0.05,0.1,0.15$. Plots showing the the individual performance of each of these methods are included below.

% \begin{figure}%[h]
% \centering
% \includegraphics[width = \textwidth]{drw_plots.pdf}
%     \caption{Results for the rejection-free procedure in the discrete Markov chain experiment. All standard errors are below $0.001$.}
%     \label{fig:dmc-scep}
% \end{figure}

% We first consider the rejection-free MH-SCEP procedure. Figure \ref{fig:dmc-scep} shows that as the dependence level $\alpha$ varies, the quality of the knockoffs for fixed $\alpha$ varies. Rejection-free MH-SCEP with more local negative dependence (lower $\rho$) performs poorly when the level of dependence is small, but it does significantly better than SCIP ($\rho = 0)$ as the level of dependence increases. With the right choice of $\rho$, the rejection-free MH-SCEP procedure is consistently near the lower bound.

\begin{figure}%[h]
\centering
\includegraphics[width = \textwidth]{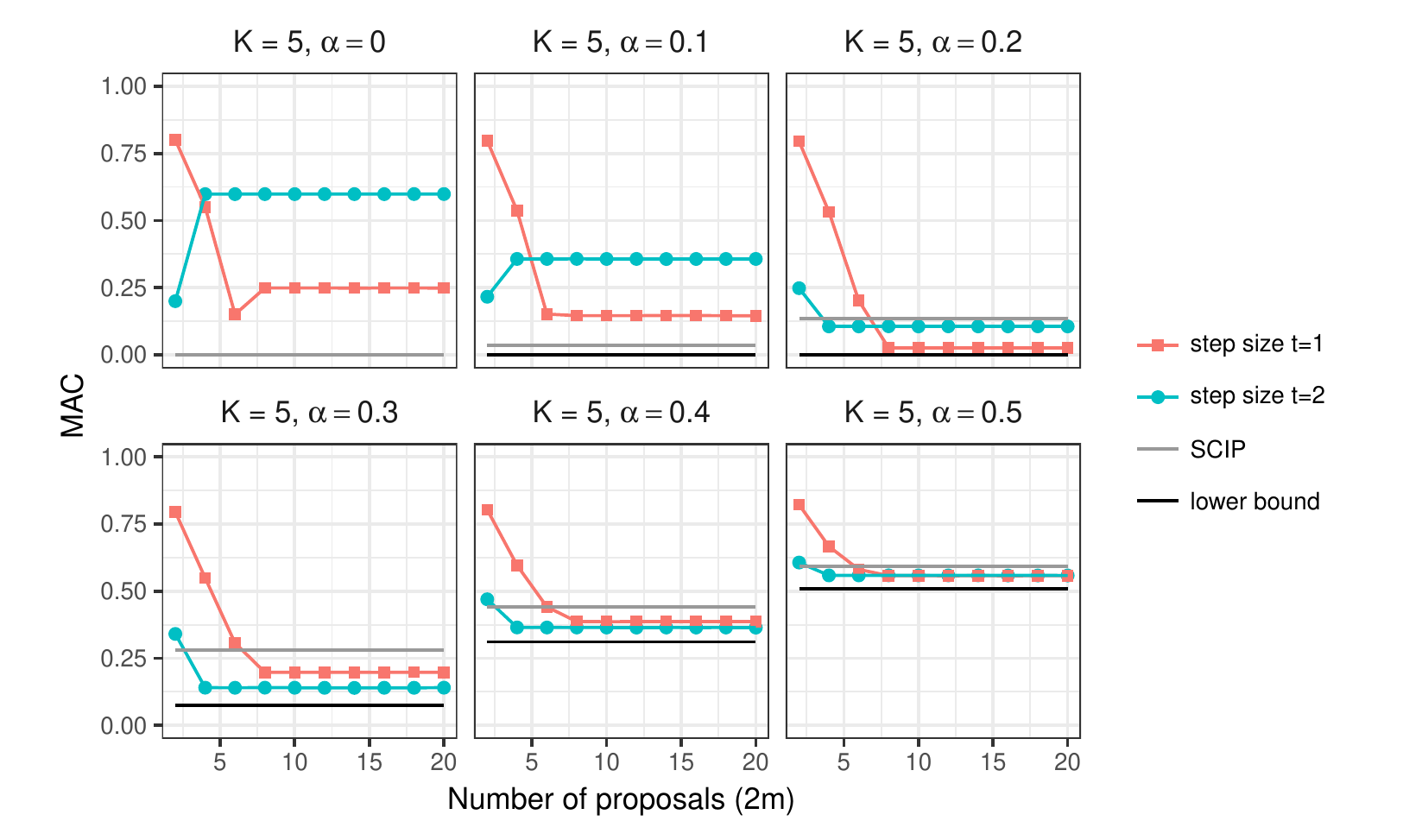}
    \caption{Simulation results for the discrete Markov chains with MTM, $K=5$, $\gamma=0.999$. All standard errors are below $0.001$.}
    \label{figure:DMC5}
\end{figure}

\begin{figure}%[h]
\centering
\includegraphics[width = \textwidth]{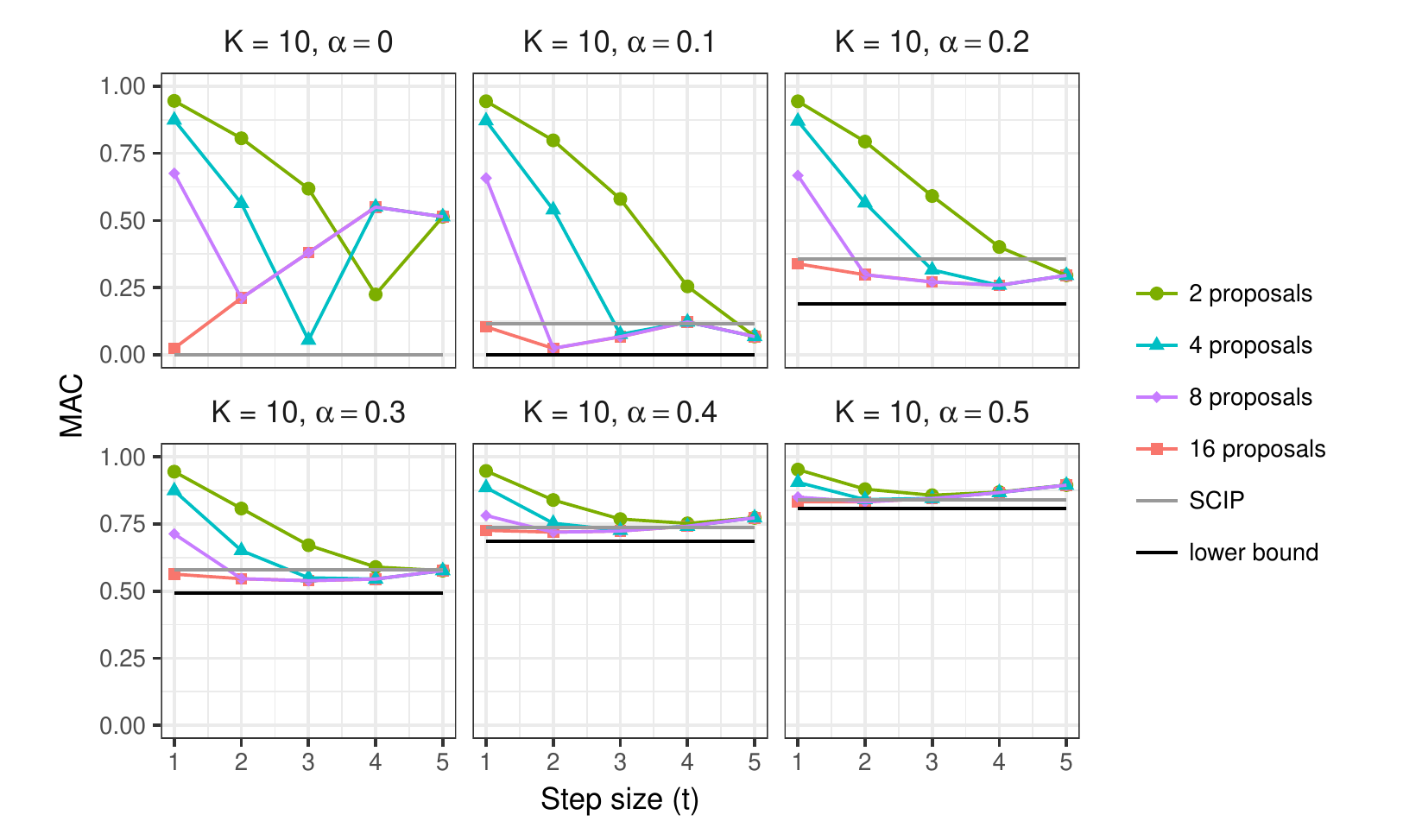}
    \caption{Simulation results for the discrete Markov chains with MTM, $K=10$, $\gamma=0.999$. All standard errors are below $0.001$.}
    \label{figure:DMC10}
\end{figure}

\subsection{Effect of \texorpdfstring{$\gamma$}{gamma}}
\label{subapp:effect-gamma}

The tuning parameter $\gamma$ for the MTM procedure introduced in Section \ref{subsec:MTM} may appear mysterious to the reader and warrants an explanation. In most of our MTM simulations, $\gamma$ is set to be $0.999$, but there are cases where it is necessary to tune $\gamma$ for improved performance. For example, in the discrete Markov chain experiment in Figure \ref{figure:DMC5}, we sometimes observe that MAC increases with the number of proposals. This is surprising, and upon closer inspection, we find that the reason is that many pairs $(X_j,\Xk_j)$ have negative correlations, which leads to increased MAC. In this case, what is happening is that we are proposing and accepting points that are so far away from $X_j$ that they become negatively correlated with $X_j$, which is undesirable. To shift the negative correlations toward zero, $\gamma$ can be decreased so that $X_j = \tilde{X}_j$ more frequently. We illustrate this in Figure \ref{figure:eff-gamma-all}. For example, in the setting where have independent coordinates taking on $K=5$ possible states, the best performance is obtained with $\gamma=4/5$, since with this value of $\gamma$ there will be a probability $1-\gamma=1/5$ of rejection, which makes $X_j$ and $\tilde X_j$ independent. Tuning $\gamma$ may also enable fewer rejections at later stages of the algorithms, since the knockoffs at later coordinates will be less constrained.\footnote{Choosing $\gamma$ less than $1$ means that no matter what the proposal $X_j^*$ is, it will be rejected with positive probability. While we typically want to avoid rejections, rejecting at early stages in the algorithm my lead to better performance by enabling higher quality knockoffs at later steps in the algorithm. In particular, with $\gamma = 1$, at some step $k > j$, it might be the case that none of the points in the proposal set have positive probability, because any point in the proposal set is inconsistent with a rejection that occurred previously in step $j$.  When $\gamma < 1$, however, any proposed value at step $k$ is consistent with a rejection at step $j$, because there is always at least a $1 - \gamma$ chance of rejecting at step $j$, so we avoid the undesirable situation described above.}
 Based on our simulation results, we only recommend tuning $\gamma$ when the variables are discrete with small support and the dependence between variables is weak.

\begin{figure}[h]
\centering
\includegraphics[width = \textwidth]{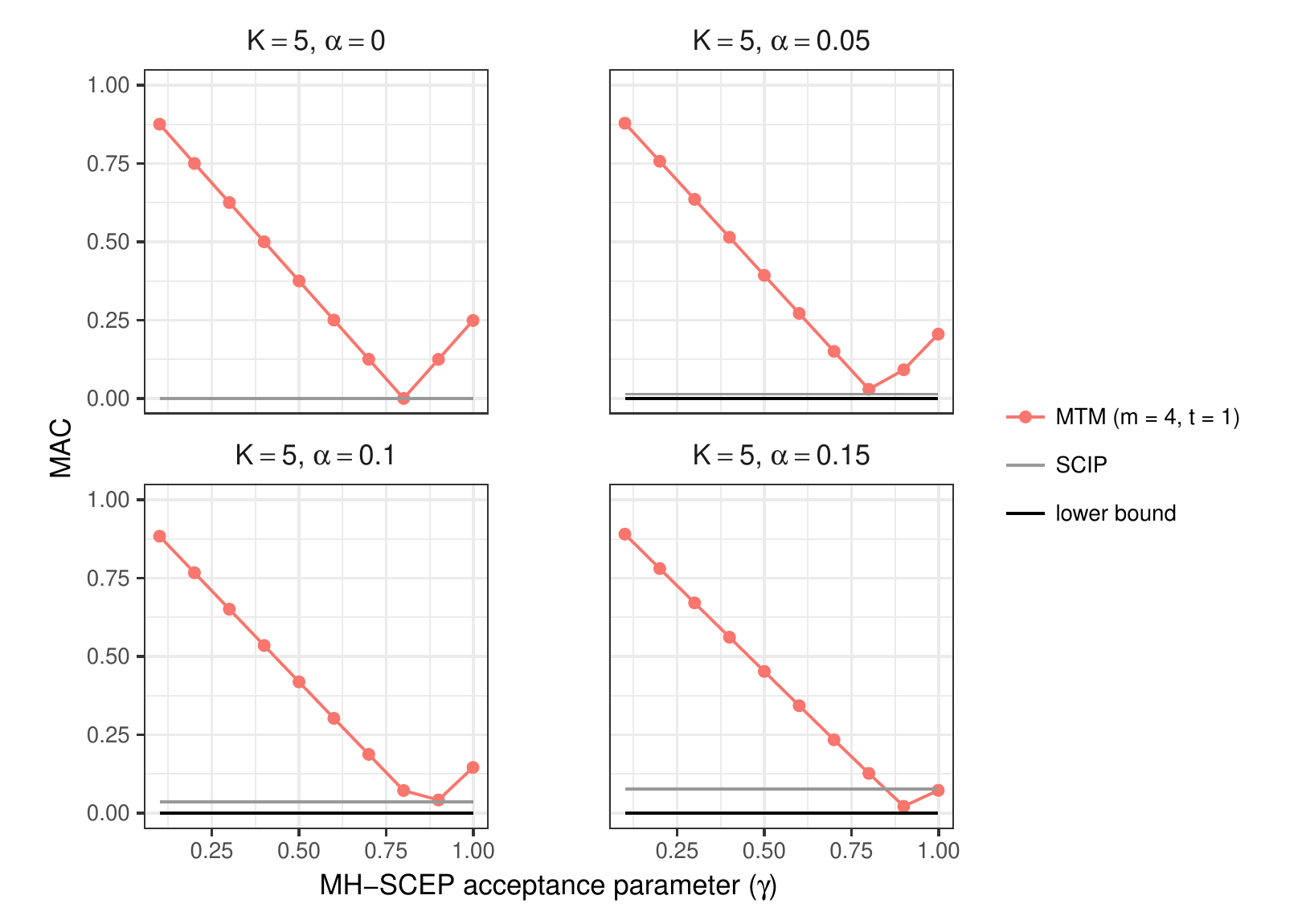}
    \caption{Simulation results showing the effect of the parameter $\gamma$ for the discrete Markov chains with the MTM method. Here, $K=5$, $\alpha=0$, $0.05$, $0.1$, $0.15$, $m=4$ and $t=1$. All standard errors are below $0.001$.}
    \label{figure:eff-gamma-all}
\end{figure}

% \begin{figure}[h]
% \centering
% \includegraphics[width = \textwidth]{GMC_all.pdf}
%     \caption{Simulation results for Gaussian Markov chains. All standard errors are below $0.001$.}
%     \label{figure:GMC_all}
% \end{figure}

% \begin{figure}[h]
% \centering
% \includegraphics[width = \textwidth]{TMC_all.pdf}
%     \caption{Simulation results for $t$-distributed Markov chains. All standard errors are below $0.001$.}
%     \label{figure:TMC_all}
% \end{figure}

% \begin{figure}[h]
% \centering
% \includegraphics[width = \textwidth]{NEMC_all.pdf}
%     \caption{Simulation results for asymmetric Markov chains. All standard errors are below $0.001$.}
%     \label{figure:NEMC_all}
% \end{figure}

\subsection{Ising model simulation details} \label{subapp:ising-sims}
Sampling Ising variables $X$ is done with a Metroplis--Hastings sampler implemented in the \texttt{bayess} R package.

\subsubsection{Divide-and-conquer simulation details}
We set $\Xk_{i_1,i_2} := X_{i_1,i_2}$ for all $(i_1,i_2)$ such that $i_1$ belongs to the set $C$ of columns defined as $C = \{1 \le i \le 100 : i = a_0 + b(w-1), b \in \mathbb N \}$; the spacing $w$ is a fixed constant (see Figure \ref{fig:grid-seperator} for an illustration) and the offset $a_0$ is chosen uniformly from $\{2,\dots,w+1\}$. This implies that  for all sites $i_1, i_2)$, $\p(X_{i_1,i_2} = \Xk_{i_1,i_2}) < 1$. 

The SDP lower bound is not available in this case, because it would require computing a $10000 \times 10000$ covariance matrix and then solving the SDP, which is intractable. Instead, to evaluate the quality of the knockoffs, we compare to Ising model knockoffs on a smaller grid that does not require the divide-and-conquer technique. Our baseline is thus the MAC evaluated at interior nodes  $1 < i_1, i_2 < 10$ achieved by 
the SCIP procedure. We consider interior nodes because we recall that correlations on the edges of the grid  are smaller. We compare this figure of merit to the MAC of the interior variables of the $100 \times 100$ grid. Without the divide-and-conquer strategy, the MAC of the two procedures would be very similar---hence, this is a sensible baseline.

\subsection{Gibbs model simulation details}
In Figure \ref{fig:gibbs}, the best MTM specification is taken from $\{(\gamma,m,t,w):\gamma=0.999,1\le m\le5,1\le t\le 7,w=3\}\cup\{(\gamma,m,t,w):\gamma=0.999,m=1,1\le t\le10,w=5\}$ for $\beta=0.07,0.1,0.3$, and from $\{(\gamma,m,t,w):\gamma=0.999,1\le m\le5,1\le t\le 7,w=3\} \cup \{(\gamma,m,t,w):\gamma=0.999,m=1,1\le t\le 10,w=3,5\}$ for $\beta=0.003,0.01,0.02,0.05$. The best $m=1$, $w=3$ MTM for each $\beta$ is taken from the same set intersecting $\{(\gamma,m,t,w):m=1,w=3\}$.

\end{document}